\documentclass[review]{elsarticle}

\usepackage{lineno, hyperref, verbatim}
\usepackage{amsmath, amssymb, amsthm, mathtools}
\usepackage{float}
\usepackage{graphicx}
\usepackage{calc}
\usepackage{xcolor}
\usepackage{subcaption}

\theoremstyle{definition}
\newtheorem{definition}{Definition}[section]
\newtheorem{proposition}[definition]{Proposition}
\newtheorem{lemma}[definition]{Lemma}
\newtheorem{theorem}[definition]{Theorem}
\newtheorem{corollary}[definition]{Corollary}
\newtheorem{alg}[definition]{Algorithm}

\newcommand{\hopes}{\mathrm{HoPeS}}
\newcommand{\al}{\alpha}
\newcommand{\be}{\beta}
\newcommand{\ga}{\gamma}
\newcommand{\de}{\delta}
\newcommand{\ep}{\epsilon}
\newcommand{\De}{\Delta}
\newcommand{\si}{\sigma}
\newcommand{\birth}{\text{birth}}
\newcommand{\death}{\text{death}}
\newcommand{\PD}{\text{PD}}
\newcommand{\R}{\mathbb{R}}
\newcommand{\Z}{\mathbb{Z}}
\newcommand{\MST}{\text{MST}}
\newcommand{\Cech}{\text{\u{C}ech }}
\newcommand{\Ch}{\text{\emph{\u{C}h}}}
\newcommand{\bd}{\partial}
\newcommand{\pr}{\text{pr}}
\newcommand{\dgap}{\text{dgap}}
\newcommand{\vgap}{\text{vgap}}
\newcommand{\dPD}{\text{DS}}
\newcommand{\vPD}{\text{VS}}
\newcommand{\uv}{\text{vs}}
\newcommand{\ud}{\text{ds}}
\newcommand{\Nbhd}{\text{Nbhd}}
\newcommand{\Reeb}{\text{Reeb}}
\newcommand{\Del}{\mathrm{Del}}
\newcommand{\shopes}{\mathrm{simHoPeS}}

\newcommand{\bsquare}{\hfill $\blacksquare$}
\newcommand{\BSDcaption}[1]{Image #1 from BSD500 with the cloud $C$ of Canny edge pixels in red, the outputs of Mapper, $\al$-Reeb and derived skeleton $\hopes_{1,1}(C)$ with one critical edge in red.}

\journal{Pattern Recognition}

\modulolinenumbers[5]

\begin{document}

\begin{frontmatter}

\title{Skeletonisation Algorithms with Theoretical Guarantees for Unorganised Point Clouds with High Levels of Noise}

\author[address]{P.~Smith}
\address[address]{Department of Computer Science, University of Liverpool, Liverpool L69 3BX, UK}

\author[address]{V.~Kurlin\corref{corresponding}}
\cortext[corresponding]{
The authors were supported by the EPSRC (EP/R018472/1) and Leverhulme Research Centre for Functional Materials Design in the Materials Innovation Factory at Liverpool.
}
\ead{vkurlin@liverpool.ac.uk, http://kurlin.org}

\begin{abstract}
Data Science aims to extract meaningful knowledge from unorganised data.
Real datasets usually come in the form of a cloud of points.
It is a requirement of numerous applications to visualise an overall shape of a noisy cloud of points sampled from a non-linear object that is more complicated than a union of disjoint clusters.
The skeletonisation problem in its hardest form is to find a 1-dimensional skeleton that correctly represents the shape of the cloud. 

This paper compares different algorithms that solve the above skeletonisation problem for any point cloud and guarantee a successful reconstruction.
For example, given a highly noisy point sample of an unknown underlying graph, a reconstructed skeleton should be geometrically close and homotopy equivalent to (has the same number of independent cycles as) the underlying graph.

One of these algorithms produces a Homologically Persistent Skeleton (HoPeS) for any cloud without extra parameters. 
This universal skeleton contains subgraphs that provably represent the 1-dimensional shape of the cloud at any scale.
Other subgraphs of HoPeS reconstruct an unknown graph from its noisy point sample with a correct homotopy type and within a small offset of the sample.
The extensive experiments on synthetic and real data reveal for the first time the maximum level of noise that allows successful graph reconstructions.
\end{abstract}


\end{frontmatter}


\section{Introduction: Motivations, Problem Statement and Contributions}
\label{sec:intro}

The central problem in Data Science is to represent unorganised data in a simple and meaningful form.
Our input data will be any cloud of points given by pairwise distances (a distance matrix) or by a smaller distance graph that spans all given points and specifies distances between endpoints of any edge. 
\smallskip

\noindent
Real data rarely splits into well-defined clusters.
A meaningful representation of such data is a 1D skeleton consisting of straight-line edges that should provably approximate a given cloud.
These 1D skeletons are relevant for curve recognition for surfaces \cite{torrente2018recognition}, topological shapes of micelles \cite{elkin2019fast} and posture identification \cite{patruno2019people}.
\smallskip

\noindent
Branches of a skeleton may reveal new classes of data points.
This paper compares the Mapper algorithm with the most relevant methods ($\al$-Reeb \cite{chazal2015gromov} and HoPeS \cite{kurlin2015one}), because they start from the same input (an unorganised point cloud) and solve the problem below with topological and geometric guarantees.
All required concepts such as homology are introduced in Appendix~\ref{sec:definitions}.
\medskip

\noindent
{\bf Data skeletonisation problem}. 
Given a noisy cloud $C$ of points sampled from a graph $G$ in a metric space, find conditions on $G$ and $C$ such that a reconstructed graph $G'$ has the same homology $H_1(G')=H_1(G)$ and geometrically approximates $G$ in the sense that $G'\subset G^{\al}$ and $G\subset (G')^{\al}$ for a suitable parameter $\al$ depending on $G$ and $C$, see $\al$-offsets defined in Appendix~\ref{sec:definitions}.
\smallskip

\begin{figure}[H]
\centering
\def\svgwidth{\columnwidth}
\includegraphics[scale = 0.108]{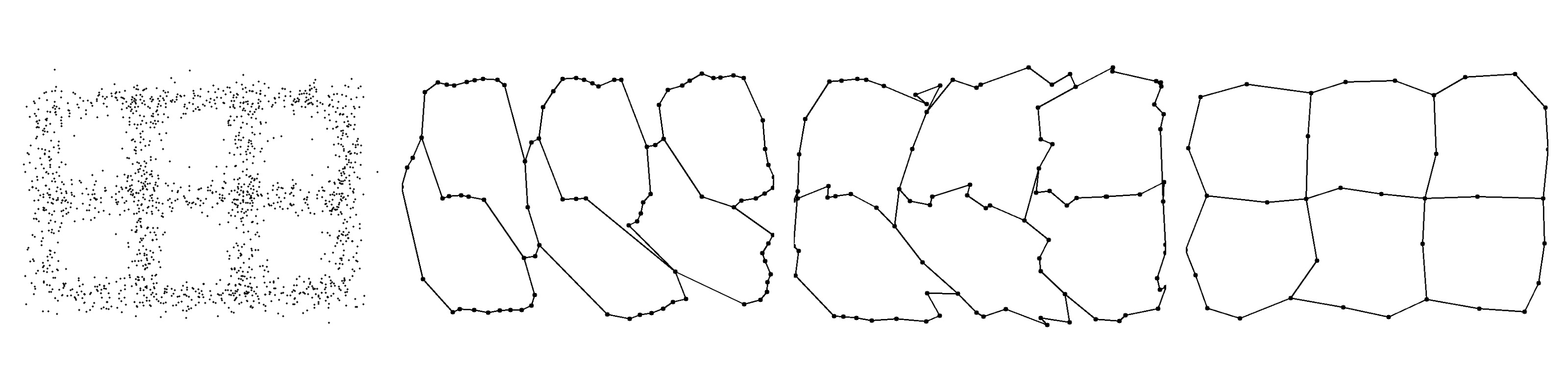}
\caption{{\bf Left}: a sample from the $G_{3, 2}$ graph in Fig.~\ref{fig:graphs} with Gaussian noise ($\sigma = 0.1$). 
{\bf Right}: reconstructions by Mapper (appendix~\ref{sub:Mapper}), $\al$-Reeb (\ref{sub:Reeb}), simplified $\hopes$ (subsection~\ref{sub:simplification}).}
\label{fig:outputsg32}
\end{figure}

\noindent
Informally, the homology condition above implies that the reconstructed graph $G'$ can be continuously deformed to the original graph $G$. 
The first two graphs in Fig.~\ref{fig:graphs} have the same homology $H_1$, because the pair of adjacent edges at each corner vertex in the grid graph $G(2,2)$ can be merged and continuously deformed into one corresponding diagonal edge in the wheel graph $W(4)$. 
Geometric similarity means that $G$ and $G'$ are close to each other with respect to a distance between graphs, e.g. $G'$ is in a small neighbourhood of $G$ and vice versa.
\medskip

\begin{figure}
\centering
\includegraphics[height=22mm]{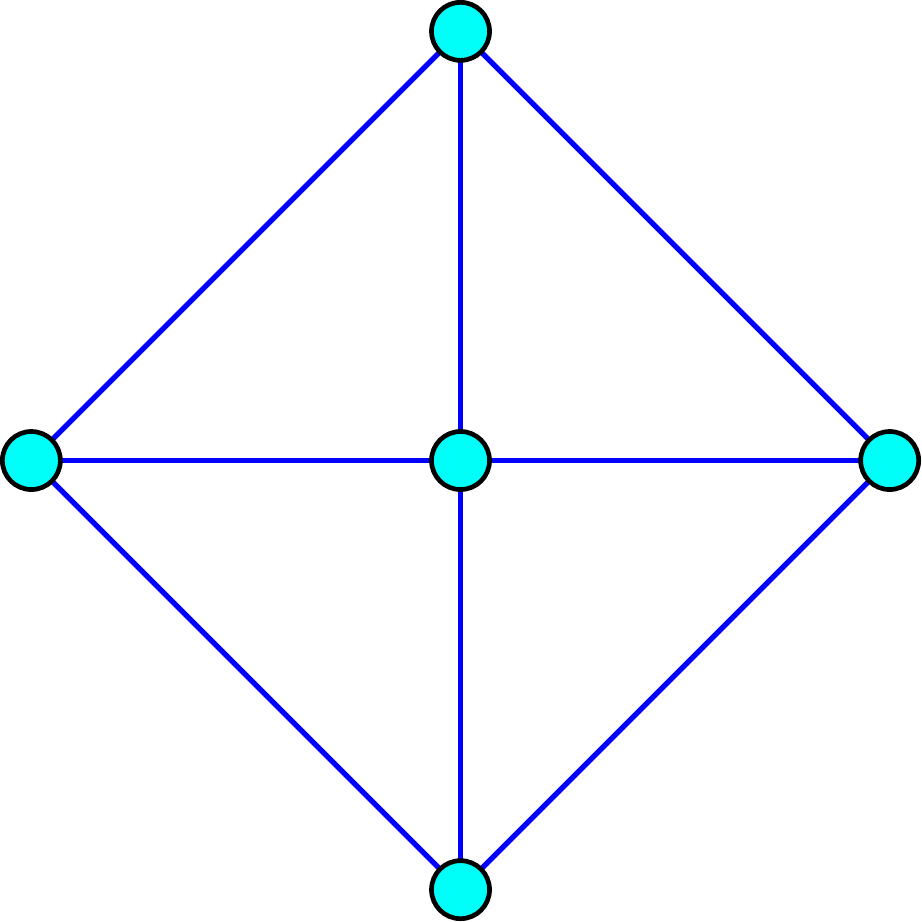}
\hspace*{2mm}
\includegraphics[height=22mm]{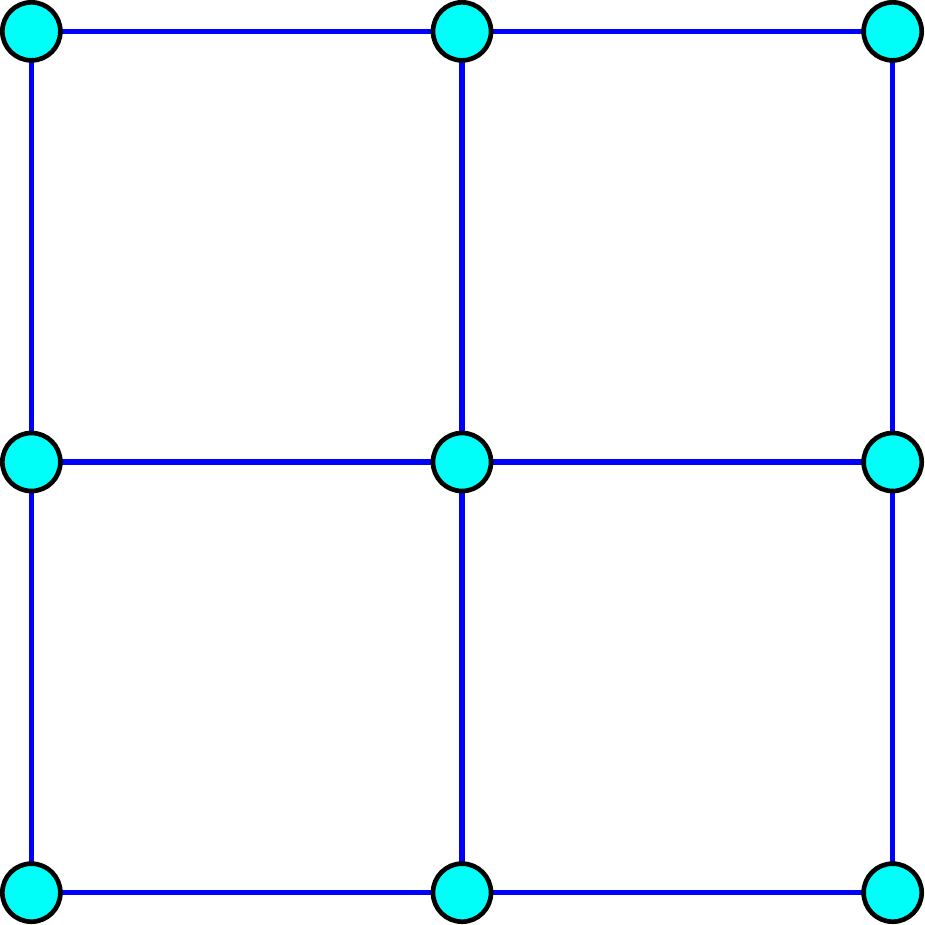}
\hspace*{2mm}
\includegraphics[height=22mm]{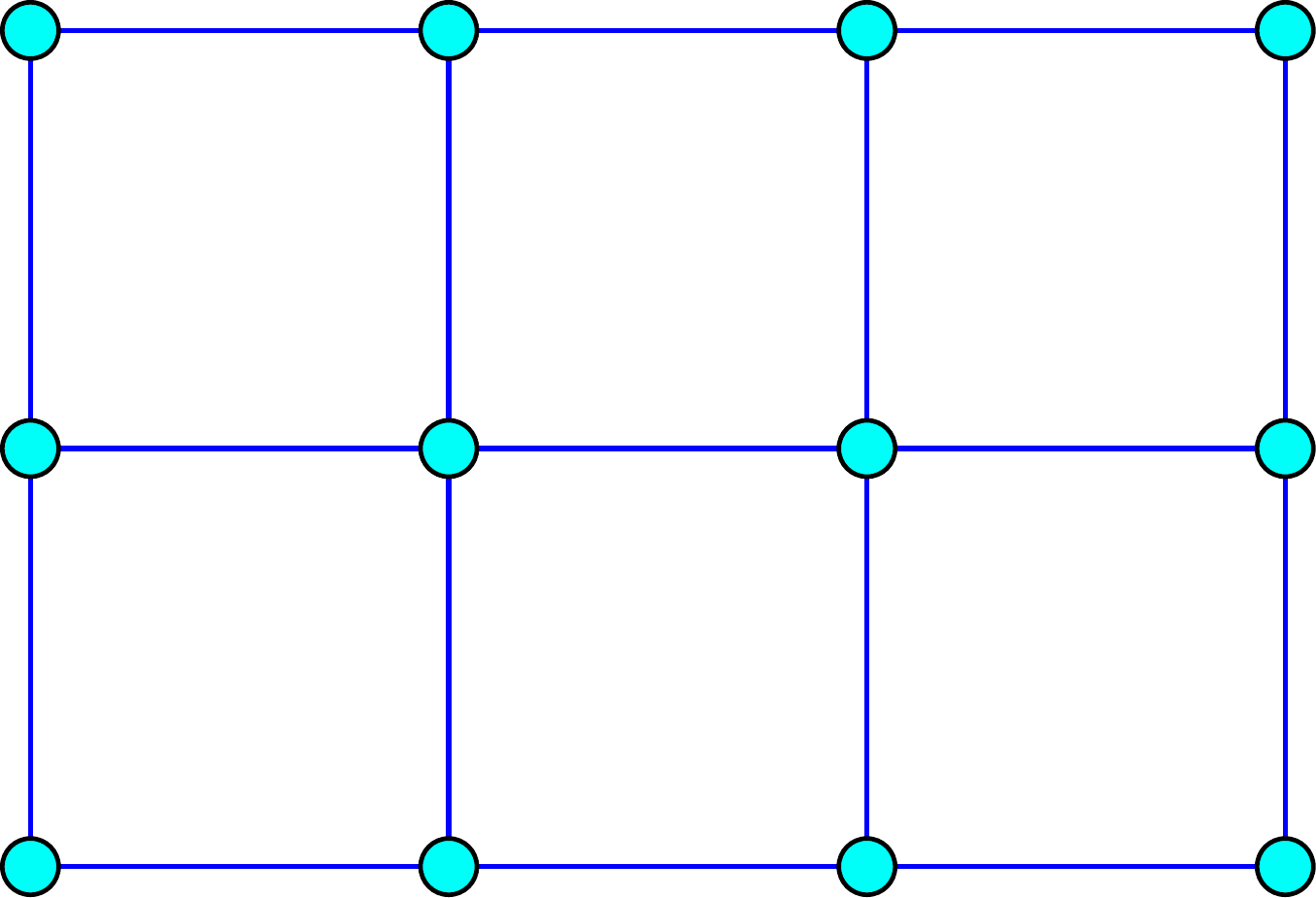}
\hspace*{2mm}
\def\svgscale{0.12}
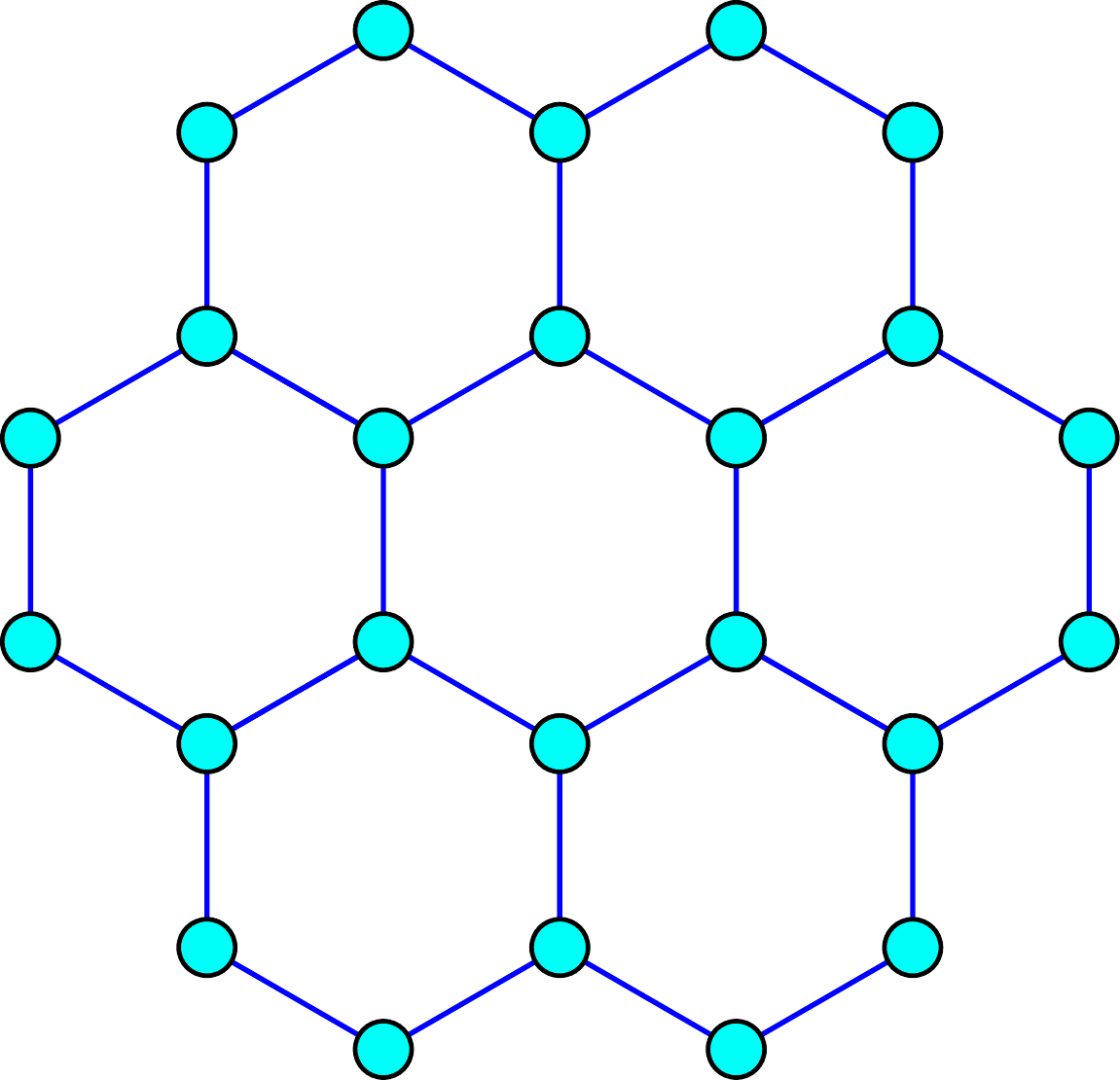
\caption{4-wheel $W(4)$, grid $G(2,2)$, grid $G(3,2)$, 7-hexagons $H(7)$, see subsection~\ref{sub:patterns}.}
\label{fig:graphs}
\end{figure}

\noindent
{\bf The contributions of the paper} to data skeletonisation are the following:
\smallskip

\noindent
$\bullet$
Section~\ref{sec:optimality} gives a much simpler proof of Optimality Theorem~\ref{thm:optimality} than for the higher-dimensional version of a Homologically Persistent Skeleton (HoPeS) by Kalisnik et al. \cite{kalisnik2019higher}.  
This result shows that $\hopes$ extends a classical Minimum Spanning Tree to an optimal graph with cycles representing the 1-dimensional shape of a point cloud $C$ by correctly capturing `cycles' of $C$ at any scale $\al$.
Briefly, the reduced skeleton $\hopes(C,\al)$ has the minimum length among all graphs that span $C$ and have the same $H_0$ and $H_1$ homology as the offset $C^{\al}$.
\smallskip

\noindent
$\bullet$
Key Stability Theorem~\ref{thm:stability} of Topological Data Analysis is extended to Graph Reconstruction Theorems~\ref{thm:reconstruction_1stgap} and \ref{thm:reconstruction_derived}, which are proved in section~\ref{sec:reconstruction} after being announced in \cite{kurlin2015one}.
Corollary~\ref{cor:stability_hopes} proves a global stability of derived subgraphs of HoPeS for the first time, which justifies its applicability to noisy data \cite{kurlin2014auto, kurlin2015homologically}. 
\smallskip

\noindent
$\bullet$
Different algorithms are extensively compared in section~\ref{sec:experiments} by their runtime and topological and geometric measures on the same datasets of real and synthetic clouds with high noise whose generation is explained in section~\ref{sec:dataset}.
Figs.~\ref{fig:noise90} and \ref{fig:noise95} show the maximum noise parameters that allow correct reconstructions.

\section{Review of Related Past Work on Data Skeletonisation Algorithms}
\label{sec:review}

Among many skeletonisation algorithms we review the most relevant ones that can accept any point cloud as in Definition~\ref{dfn:metric} and provide guarantees with respect to the skeletonisation problem in section~\ref{sec:intro}, see Appendix~\ref{sec:definitions}.
\smallskip

\noindent 
\textbf{Past algorithms without guarantees}.
Singh et al. \cite{SCP00} approximated a cloud $C\subset\R^d$ by a subgraph of a Delaunay triangulation, which requires $O(n^{\lceil{d/2}\rceil})$ time for $n$ points of $C$
and three thresholds: a minimum number $K$ of edges in a cycle and $\de_{min},\de_{max}$ for inserting/merging 2nd order Voronoi regions.
Similar parameters are needed for{\em principal curves} \cite{kegl2002piecewise}, which were later extended to iteratively computed elastic maps \cite{GZ09}. Since it is hard to estimate a rate of
convergence for iterative algorithms, we discuss below non-iterative methods.
\smallskip

\noindent 
\textbf{Skeletonisation via Reeb graphs}. Starting from a noisy sample $C$ of an unknown graph $G$ with a scale parameter $\al$, X.~Ge et al. \cite{GSBW11} considered the Reeb graph of the Vietoris-Rips complex on $C$ at a scale $\al$, see Definition~\ref{dfn:Cech+VR}. 
\smallskip

\noindent \textbf{The $\al$-Reeb graph} $G$ was introduced by
F.~Chazal et al. \cite{chazal2015gromov} for a metric space $X$ at a user-defined scale $\al$, see Definition~\ref{dfn:alpha-Reeb}. 
If $X$ is $\ep$-close to an unknown graph with edges of minimum length $8\ep$, the output $G$ is $34(\be(G) + 1)\ep$-close to the input $X$, where $\be(G)$ is the first Betti number of $G$, see \cite[Theorem~4.9]{chazal2015gromov}. 
The $\al$-Reeb graph has a metric, but isn't embedded into any space even if $C \subset \R^2$. 
The algorithm has the fast time $O(n\log n)$ for $n$ input points of $X$. 
\smallskip

\noindent
\textbf{Mapper} \cite{singh2007topological} extends any clustering and outputs a network of interlinked clusters by using a user-defined filter function $f:C\to\R$, which helps to link different clusters of a cloud $C$.
M.~Carri\'ere et al. \cite[Theorem 5.2]{stabilityMapper} found a connection between Mapper and the Reeb graph via MultiNerve Mapper. 
\smallskip

\noindent 
\textbf{Metric graph reconstruction}. M.~Aanjaneya et al. \cite{ACCGGM12} studied a related problem approximating a metric on a large input graph $Y$ by a metric on a small output graph $\hat X$. Let $Y$ be a good $\ep$-approximation to an unknown graph $X$. 
Then \cite[Theorem~2]{ACCGGM12} is the first guarantee for  the existence of a homeomorphism $X \to \hat X$ that distorts the metrics on $X$ and $\hat X$ with a multiplicative factor $1 + c\ep$ for $c > \frac{30}{b}$, where $b > 14.5\ep$ is the length of the shortest edge of $X$. 
\smallskip

\noindent
\textbf{Graph reconstruction by discrete Morse theory}.
A Homological Spanning Forest \cite{molina2012homological} 
uses a given pixel grid of 2D images, hence cannot be applied to an arbitrary point cloud $C$.
Similarly, the recent algorithm by T.~Dey et al. \cite{dey2018graph}
in addition to a cloud $C$ requires a density field $\rho:\R^d\to\R$ (usually on a regular grid), which `concentrates' around a hidden graph.
Section~\ref{sec:experiments} compares the three algorithms that accept an input cloud $C$ without any extra structure.

\section{$\hopes(C)$: a Homologically Persistent Skeleton of a cloud $C$}
\label{sec:hopes}

The Homologically Persistent Skeleton (HoPeS) was introduced in 2015 \cite{kurlin2015one} and extended to higher dimensions by S.~Kalisnik et al. \cite[Definition~4.5]{kalisnik2019higher}, which has crucially clarified potential choices of homology representatives.
The 1-dimensional HoPeS was motivated by extending a reconstruction of hole boundaries in unorganised clouds of edge pixels in images to a 1D skeleton \cite{Kur14cvpr, kurlin2016fast}. 

\subsection{A Minimum Spanning Tree of any point cloud in a metric space}

\begin{definition}[$\MST$]
\label{dfn:MST}
Fix a filtration $\{Q(C; \al)\}$ of complexes on a cloud $C$.
A \textit{Minimum Spanning Tree} $\MST(C)$ is a connected graph with vertex set $C$ and the minimum total length of edges.
The length of an edge $e$ is the minimum $\alpha$ such that $e\subset Q(C; \al)$.
For any scale $\al \geq 0$, a \textit{forest} $\MST(C; \al)$ is obtained from $\MST(C)$ by removing interiors of all edges that are longer than $2\al$.
\bsquare
\end{definition}

$\MST(C)$ may not be unique if different edges enter the filtration $\{Q(C;\al)\}$ at the same scale $\al$, e.g. $\MST(C)$ in Fig.~\ref{fig:cloud10points_MST} may include the upper horizontal edge instead of the lower one.
For all these choices, $\MST(C;\al)$ enjoys the optimality in Lemma~\ref{lem:MST_optimal}.
A graph $G$ \emph{spans} a possibly disconnected complex $Q$ on a cloud $C$ if $G$ has the vertex set $C$, every edge of $G$ belongs to $Q$ and the inclusion $G \subset Q$ induces a 1-1 correspondence between connected components.

\begin{lemma}
\label{lem:MST_optimal}
Fix a filtration $\{Q(C; \al)\}$ of complexes on a cloud $C$ in a metric space. 
For any fixed scale $\al \geq 0$, the forest $\MST(C; \al)$ has the minimum total length of edges among all graphs that span $Q(C; \al)$ at the same scale $\al$.  
\end{lemma}

\begin{proof}
Let $e_1, \dots, e_m \subset \MST(C)$ be all edges longer than $2\al$, so $\MST(C) = e_1 \cup \dots \cup e_m \cup \MST(C; \al)$. 
Assume that there is a graph $G$ that spans $Q(C;\al)$ and is shorter than $\MST(C; \al)$. 
Since both inclusions $G\subset Q(C;\al)$ and $\MST(C; \al)\subset Q(C;\al)$ induce 1-1 maps on connected components, the graph $G \cup e_1 \cup \dots \cup e_m$ is connected and is shorter than $\MST(C)$, which  contradicts Definition~\ref{dfn:MST}. \qedhere
\end{proof}

Lemma~\ref{lem:MST_optimal} implies that, for any scale $\al$, the complex $Q(C;\al)$ can be replaced by the smaller graph $\MST(C;\al)$ that has the same connected components and the minimal total length among all graphs with this connectivity property.
\smallskip

Visually, the 0-dimensional shape (connected components) of each offset $C^{\al}$ on a cloud $C\subset\R^d$ can be represented by the forest $\MST(C;\al)$ instead of the larger $\al$-complex $C(\al)$ homotopically equivalent to $C^{\al}$, see Lemma~\ref{lem:nerve}.

\subsection{Persistent homology and its stability under perturbations}

A Homologically Persistent Skeleton ($\hopes$) will extend the optimality of $\MST$ in Lemma~\ref{lem:MST_optimal} to $H_1$ in Theorem~\ref{thm:optimality}.
$\hopes$ is based on the persistent homology that summarises the evolution of homology 
as formalised below.

\begin{definition}[births and deaths]
\label{dfn:birth-death}
For any filtration $\{Q(C; \al)\}$ of complexes on a cloud $C$ in a metric space, a homology class $\ga \in H_1(Q(C; \al_i))$ is \textit{born} at $\al_i = \birth(\ga)$ if $\ga$ is not in the full image under the induced homomorphism $H_1(Q(C; \al)) \to H_1(Q(C; \al_i))$ for any $\al < \al_i$. The class $\ga$ \textit{dies} at $\al_j = \death(\ga) \geq \al_i$ when the image of $\ga$ under $H_1(Q(C; \al_i)) \to H_1(Q(C; \al_j))$ merges into the image under $H_1(Q(C; \al)) \to H_1(Q(C; \al_j))$ for some $\al < \al_i$.
\bsquare
\end{definition}

The birth-death pairs from Definition~\ref{dfn:birth-death} can be similarly defined for higher-dimensional homology groups $H_k$ with $k>1$ and coefficients in a field, though coefficients in $\Z_2=\{0,1\}$ are used in practice and in our paper.
For $k=0$ and the filtration of offsets $C^{\al}$, the homology group $H_0(C^{\al})$ is generated by connected components of $C^{\al}$.
Then all homology classes of $H_0(C^{\al})$ are born at $\al=0$ (isolated dots) and die at $\al$ equal to half-lengths of edges in $\MST(C)$. 

\begin{definition}[persistence diagram]
\label{dfn:pers-diagram}
Fix a filtration $\{Q(C; \al)\}$ of complexes on a cloud $C$. Let $\al_1, \dots, \al_m$ be all scales when a homology class is born or dies in $H_1(Q(C; \al))$. 
Let $\mu_{ij}$ be the number of independent classes in $H_1(Q(C; \al))$ that are born at $\al_i$ and die at $\al_j$. The \textit{persistence diagram} $\PD\{Q(C; \al)\} \subset \R^2$ is the multi-set consisting of the birth-death pairs $(\al_i, \al_j)$ with integer multiplicities $\mu_{ij} \geq 1$ and all diagonal dots $(x, x)$ with infinite multiplicity.
\bsquare
\end{definition}

For the cloud $C$ in Fig.~\ref{fig:offsets} and the filtration $\{C^{\al}\}$, the homology $H_1$ has 2 classes that persist over $1.5 \leq \al < R$ and $2 \leq \al < R$, where $R = \frac{15}{8}\sqrt{17}$ is the circumradius of the largest Delaunay triangle in $\Del(C)$ with sides $4, \sqrt{17}, 5$. 
We use the value $R \approx 2.577$.
Hence the persistence diagram $\PD\{C^{\al}\}$ has 2 off-diagonal red dots $(1.5, 2.577)$ and $(2, 2.577)$ in the last picture of Fig.~\ref{fig:offsets}. 
\smallskip

\begin{figure}
\centering
\includegraphics[scale = 0.5]{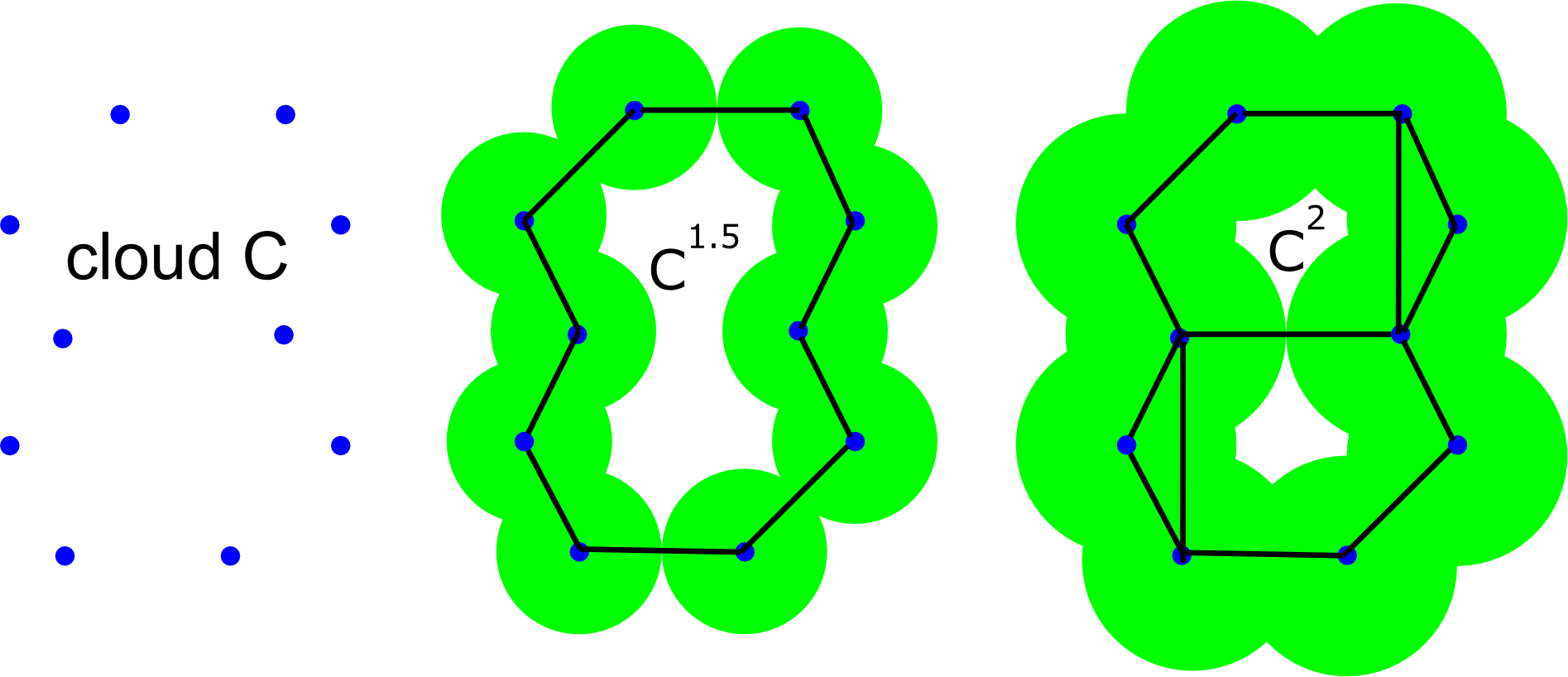}
\includegraphics[scale = 0.7]{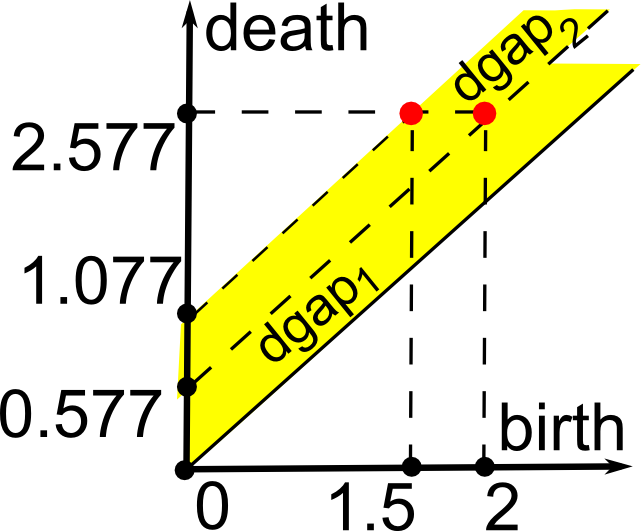}
\caption{A cloud $C$, its $\al$-offsets $C^{1.5}$, $C^2$ and the 1D persistence diagram $\PD\{C^{\al}\}$.}
\label{fig:offsets}
\end{figure}

The persistence diagram $\PD\{C^{\al}\}$ is invariant under any isometry of $C$, i.e. any transformation that preserves distances between points.
The key advantage of $\PD\{C^{\al}\}$ over other geometric invariants
is Stability Theorem~\ref{thm:stability}, which requires the bottleneck distance between persistence diagrams defined below.

\begin{definition}[bottleneck distance $d_B$]
\label{dfn:bottleneck}
For points $p = (x_1, y_1)$, $q = (x_2, y_2)$ in $\R^2$, the \textit{$L_{\infty}$-distance} is $||p - q||_{\infty} = \max\{|x_1 - x_2|, |y_1 - y_2|\}$. 
The \textit{bottleneck distance} between persistence diagrams $\PD, \PD'$ is defined as $d_B = \inf\limits_{\psi}\sup\limits_{q \in \PD} ||q - \psi(q)||_{\infty}$ for all bijections $\psi: \PD \to \PD'$ between multi-sets.
\bsquare
\end{definition}

Stability Theorem~\ref{thm:stability} below informally says that any small perturbation of original data leads to a small perturbation of the persistence diagram. A metric space $M$ is \textit{totally bounded} if $M$ has a finite $\ep$-sample $C \subset M$ for any $\ep > 0$.

\begin{theorem}[stability of persistence]
\label{thm:stability}
\cite[simplified Theorem~5.6]{CdSO14} Let $C$ be any $\ep$-sample of a graph $G$ in a totally bounded metric space $M$. Then the persistence diagrams of \Cech filtrations on $G$ and $C$ are $\ep$-close, namely $d_B(\PD\{\Ch(G, M; \al)\}, \PD\{\Ch(C, M; \al)\}) \leq \ep$.
This inequality holds for the filtrations of $\al$-offsets by Nerve Lemma~\ref{lem:nerve}, namely $d_B(\PD\{G^{\al}\}, \PD\{C^{\al}\}) \leq \ep$.
\end{theorem}

\begin{figure}
\centering
\begin{minipage}{0.47\textwidth}
\centering
\includegraphics[height=30mm]{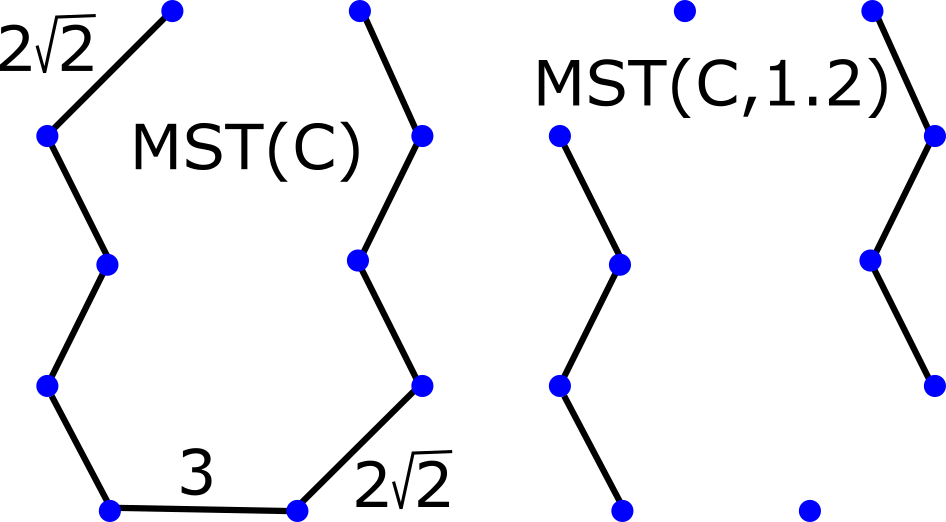}
\captionof{figure}{A full tree $\MST(C)$ and the forest $\MST(C,1.2)$ for the cloud $C$ in Fig.~\ref{fig:offsets}.}
\label{fig:cloud10points_MST}
\end{minipage} \hspace{1em}
\begin{minipage}{0.47\textwidth}
\centering
\includegraphics[height=30mm]{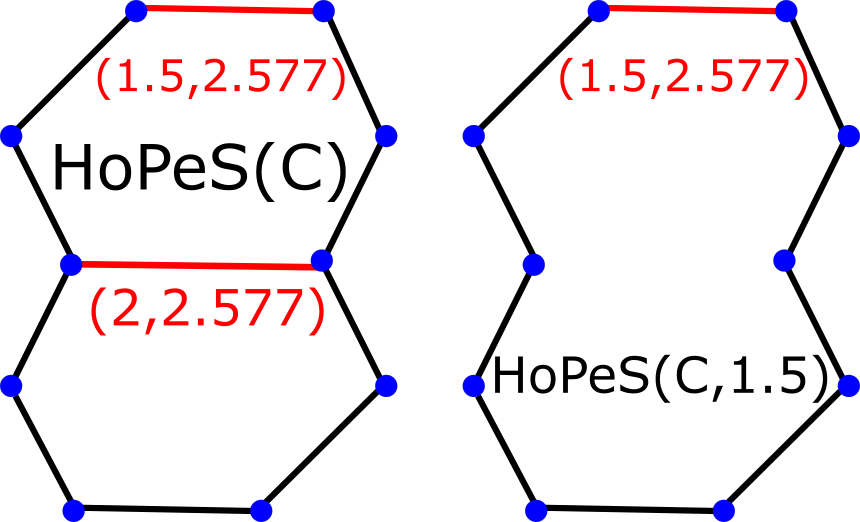}
\captionof{figure}{Full $\hopes(C)$ and reduced subgraph $\hopes(C,1.5)$ for $C$ in Fig.~\ref{fig:offsets}.}
\label{fig:cloud10points_hopes}
\end{minipage}
\end{figure}

\subsection{Gap search method for stable subdiagrams of a persistence diagram}
\label{sub:gap-search}

Stability Theorem~\ref{thm:stability} extends to important subdiagrams with highly persistent dots representing features that persist over a long interval of a scale. 

\begin{definition}[diagonal gaps and subdiagrams]
\label{dfn:diagonal-gaps}
For any metric space $C$, a \textit{diagonal gap} of $\PD\{C^{\al}\}$ is a strip $\{0 \leq a < y - x < b\}$ that has dots in both boundary lines $\{y - x = a\}$ and $\{y - x = b\}$, but not inside the strip. For any $k \geq 1$, the $k$-th \textit{widest diagonal gap} $\dgap_k(C)$ has the $k$-th largest vertical width $|\dgap_k(C)| = b - a$.
The \textit{subdiagram} $\dPD_k(C) \subset \PD\{C^{\al}\}$ consists of only the dots above the lowest of the first $k$ widest $\dgap_i(C)$, $i = 1, \dots, k$.
Each $\dPD_k(C)$ is bounded below by $y - x = a$ and has the \textit{diagonal scale} $\ud_k(C) = a$.
\bsquare
\end{definition}

In Definition~\ref{dfn:diagonal-gaps}, if $\PD\{C^{\al}\}$ has different gaps of the same width, we say that a lower diagonal gap has a larger width. 
If $\PD\{C^{\al}\}$ has dots only in $m$ different lines $\{y - x = a_i > 0\}$, $i = 1, \dots, m$, we have $m$ diagonal gaps $\dgap_i(C)$. 
We ignore the highest gap $\{y - x > \max a_i\}$, so we set $\dgap_i(C) = \emptyset$ 
for $i > m$.

\begin{figure}
\centering
\begin{minipage}{0.47\textwidth}
\centering
\includegraphics[height=18mm]{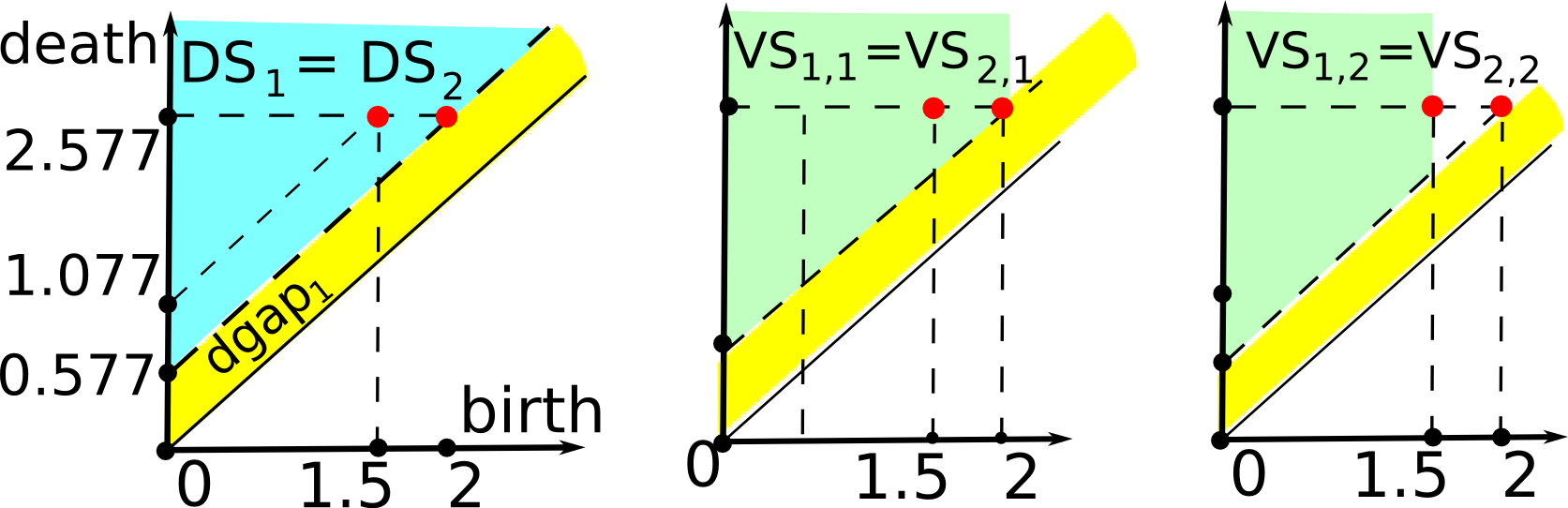}
\captionof{figure}{Subdiagrams $\dPD_k(C), \vPD_{k, l}(C)$ from Definitions~\ref{dfn:diagonal-gaps} and \ref{dfn:vertical-gaps} for $C$ in Fig.~\ref{fig:offsets}.
}
\label{fig:cloud10-allgaps}
\end{minipage} \hspace{1em}
\begin{minipage}{0.47\textwidth}
\centering
\includegraphics[height=18mm]{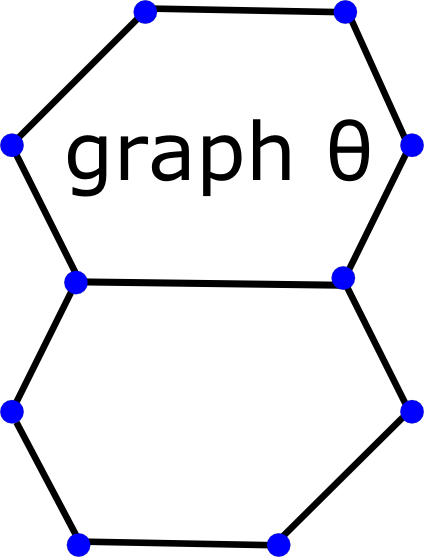} 
\hspace*{1pt}
\includegraphics[height=18mm]{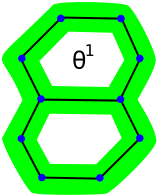} \hspace*{1pt}
\includegraphics[height=18mm]{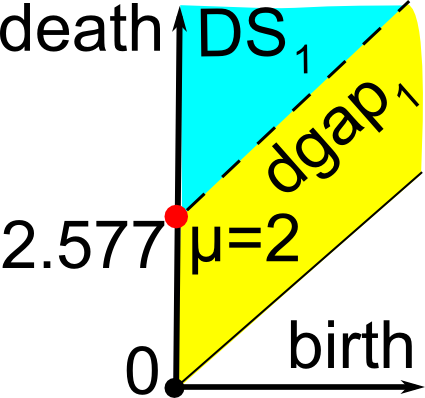}
\captionof{figure}{The graph $G\subset \R^2$ with the offset $G^{1}$ and 1D persistence diagram $\PD\{G^{\al}\}$.}
\label{fig:theta+offsets+persistence}
\end{minipage}
\end{figure}

The cloud $C$ in Fig.~\ref{fig:offsets} has the persistence diagram $\PD\{C^{\al}\}$ with only 2 off-diagonal dots $(1.5, 2.577)$ and $(2, 2.577)$. Then $\dgap_1(C) = \{0 < y - x < 0.577\}$ is below $\dgap_2(C) = \{0.577 < y - x < 1.577\}$. So $\dPD_1(C)=\dPD_2(C)$ consists of $(1.5, 2.577), (2, 2.577)$, $\ud_1(C) = \ud_2(C) = 0.577$ in Fig.~\ref{fig:cloud10-allgaps}. 
Definition~\ref{dfn:diagonal-gaps} makes sense in any persistence diagram, say for $\al$-offsets $G^{\al}$ of a graph $G \subset M$. 
The graph $G$ in Fig.~\ref{fig:theta+offsets+persistence} has the 1D diagram $\PD\{G^{\al}\}$ containing only one off-diagonal dot $(0, 2.577)$ of multiplicity 2.
Then $\dgap_1(G)=\{0 < y - x < 2.577\}$, $|\dgap_2(G)| = 0$, $\dPD_1(G) = \{(0, 2.577), (0, 2.577)\}$ and $\ud_1(G) = 2.577$. \smallskip

If we need to reconstruct cycles of $G$ from a noisy sample $C$ using the perturbed diagram $\PD\{C^{\al}\}$, we should look for birth-death pairs having a high persistence $\death-\birth$ and small $\birth$. 
Hence we search for a vertical gap to separate the dots $(\birth, \death)$ with a small $\birth$ from all other dots.

\begin{definition}[vertical gap $\vgap_{k, l}$, subdiagram $\vPD_{k, l}$, scale $\uv_{k, l}$]
\label{dfn:vertical-gaps}
In the diagonal subdiagram $\dPD_k(C)$ from Definition~\ref{dfn:diagonal-gaps}, a \textit{vertical gap} is the widest vertical strip $\{a < x < b\}$ whose boundary contains a dot from $\dPD_k(C)$ in the line $\{x = a\}$, but not inside the strip. For $l \geq 1$, the $l$-th \textit{widest vertical gap} $\vgap_{k, l}(C)$ has the $l$-th largest horizontal width $|\vgap_{k,l}(C)| = b - a$. The \textit{vertical subdiagram} $\vPD_{k, l}(C) \subset \dPD_k(C)$ consists of only the dots to the left of the first $l$ widest vertical gaps $\vgap_{k, j}(C)$, $j = 1, \dots, l$. So each $\vPD_{k, l}(C)$ is bounded on the right by $x = b$ and has the \textit{vertical scale} $\uv_{k, l}(C) = a$.
\bsquare
\end{definition}

In Definition~\ref{dfn:vertical-gaps}, if there are different vertical gaps with the same horizontal width, we say that the leftmost vertical gap has a larger width. We prefer the leftmost vertical gap, while in Definition~\ref{dfn:diagonal-gaps} we prefer the lowest diagonal gap.
\smallskip

We allow the case $b = +\infty$, so the widest vertical gap $\vgap_{k, 1}(C)$ always has the form $\{x > a\}$, $VS_{k, 1}(C) = \dPD_k(C)$ and we set $|\vgap_{k, 1}(C)| = +\infty$. 
\smallskip

For the cloud $C$ in Fig.~\ref{fig:offsets}, the diagonal subdiagram $\dPD_1(C)$ has $\vgap_{1, 1}(C) = \{x > 2\}$, $\vgap_{1, 2}(C) = \{1.5 < x <2\}$, $\vPD_{1, 1}(C) = \{(1.5, 2.577), (2, 2.577)\}$ and  $\vPD_{1, 2}(C) = \{(1.5, 2.577)\}$, so $\uv_{1, 1}(C) = 2$, $\uv_{1, 2}(C) = 1.5$, see Fig.~\ref{fig:cloud10-allgaps}. 
Any cloud $C$ has no cycles at $\al = 0$, hence $\PD\{C^{\al}\}$ has no dots in the vertical axis, so any $\vgap_{k, l}(C)$ has a left boundary line $\{x = a > 0\}$ and all $\uv_{k, l}(C) > 0$. 
\smallskip

Definition~\ref{dfn:vertical-gaps} makes sense for any persistence diagram, say for $\al$-offsets $G^{\al}$ of a graph $G$ in a metric space. 
The diagonal subdiagram $\dPD_1(\theta)$ for the graph $\theta$ in Fig.~\ref{fig:theta+offsets+persistence} consists of the doubled dot $(0,2.577)$. 
The only vertical subdiagram $\vPD_{1, 1}(\theta)$ consists of the same doubled dot $(0, 2.577)$, so $\uv_{1, 1}(\theta) = 0$.

\subsection{$\hopes(C)$ is the persistence-based extension of $\MST(C)$ for a cloud $C$}

$\hopes(C)$ will be obtained from $\MST(C)$ by adding critical edges below.

\begin{definition}[critical edges]
\label{dfn:critical}
\cite[modified from Definition 4.5]{kalisnik2019higher} 
Fix a filtration $\{Q(C; \al)\}$ of complexes on a cloud $C$. 
An edge $e$ is \textit{critical} if the addition of $e$ to the filtration $\{Q(C; \al)\}$ creates a new homology class that does not immediately die, see two red critical edges in Fig.~\ref{fig:cloud10points_hopes}. 
The \textit{birth} of a critical edge $e$ is the scale $\al$ when $e$ enters the filtration.
Define $E(\al) = \{e_1, ..., e_s\}$ to be the set of critical edges with $\birth(e_i) \leq \al$ that have not yet been assigned a death value. 
Then $[e_1], ..., [e_s]$ form a basis of $H_1((\MST(C; \al) \cup E(\al)) \: / \: \MST(C; \al))$. 
Denote 
$$f: H_1((\MST(C; \al) \cup E(\al)) \: / \: \allowbreak \MST(C; \al)) \to H_1(Q(C; \al) \: / \: \allowbreak \MST(C; \al))$$ to be the map on homology induced by the inclusion 
$$(\MST(C; \al) \cup E(\al) ) \: / \: \allowbreak \MST(C; \al) \to Q(C; \al) \: / \: \MST(C; \al).$$ 

Let $\{b_1, ..., b_r\}$ be a basis of $\text{ker}(f)$, where $r = \text{dim ker}(f)$ is equal to the number of critical edges that die at exactly scale $\al$, so we have $r \leq s$. We can expand each basis element as $b_i = \sum_{j = 1}^{s} c_{ij}[e_j]$, with $c_{ij} \in\Z_2$ (coefficients for $H_1$). 
We consider (in $\Z_2$) the system of equations $\sum_{j = 1}^{s} c_{ij}x_j = 0$ for $i \in \{1, ..., r\}$. Since basis elements are linearly independent, so are these equations. Then this system of equations can be solved with $r$ leading variables, each of which is expressed in terms of the remaining $s - r$ free variables. 
Let $\mathit{I} \subseteq \{1, ..., s\}$ be the set of all indices of leading variables ($\mathit{I}$ could be empty if $f$ is trivial). 
For each $i \in \mathit{I}$, we declare the \emph{death} value of the critical edge $e_i$ to be $\al$.
\bsquare
\end{definition}

The \emph{elder rule} selects those leading variables whose corresponding set of critical edges has the greatest combined birth value. 
If there are distinct sets with the greatest combined birth value, then we have a genuine choice.
However, Optimality Theorem~\ref{thm:optimality} will be valid regardless of the choice.

\begin{definition}[HoPeS]
\label{dfn:hopes}
For any filtration $\{Q(C; \al)\}$ of complexes on a cloud $C$, a \textit{Homologically Persistent Skeleton} $\hopes(C)$ is the union of a minimum spanning tree $\MST(C)$ and all critical edges with labels $(\birth,\death)$.
For any $\al \geq 0$, the \textit{reduced} skeleton $\hopes(C; \al)$ is obtained from $\hopes(C)$ by removing all edges longer than $2\al$ and all critical edges $e$ with $\death(e) \leq \al$.
\bsquare
\end{definition}

\section{Optimality guarantees for reduced subgraphs of $\hopes(C)$}
\label{sec:optimality}

Optimality Theorem~\ref{thm:optimality} will prove that,  for any scale $\al$, the reduced subgraphs $\hopes(C;\al)$ are minimal graphs that span a cloud $C$ and have the same components and cycles as the offset $C^{\al}$.
We fix a filtration $\{Q(C; \al)\}$ of complexes on a cloud $C$, but use the simpler notation $\hopes(C)$ for the skeleton. 

\begin{lemma}
\label{lem:length=birth}
Let a class $\ga \in H_1(Q(C;\al))$ be born due to a critical edge $e$ added to $Q(C;\al)$. 
Then $2\birth(\ga)$ is the length $|e|$ relative to the filtration $\{Q(C; \al)\}$.
\end{lemma}

\begin{proof}
By Definition~\ref{dfn:critical} the critical edge $e$ is the last edge added to a cycle $L \subset Q(C; \al)$ giving birth to the homology class $\ga$ at $\al = \birth(\ga)$.
The length $|e|$ equals the doubled scale $2\al$ when $e$ enters $Q(C; \al)$, so $|e| = 2 \cdot \birth(\ga)$. 
\qedhere
\end{proof}

\begin{lemma}
\label{lem:hopes_in_complex} 
The reduced skeleton $\hopes(C; \al)$ is a subgraph of the complex $Q(C;\al)$ for any $\al$, e.g. $\hopes(C; \al)\subset\Del(C)$ for the filtration of $\al$-complexes. 
\end{lemma}

\begin{proof}
By Definition~\ref{dfn:hopes}, the skeleton $\hopes(C; \al)$ for any $\al$ consists of the forest $\MST(C; \al)$ and all critical edges with lengths $\death(e) \leq |e| \leq 2\al$. Lemma~\ref{lem:MST_optimal} implies that $\MST(C; \al) \subset Q(C; \al)$. Any critical edge $e \subset \hopes(C)$ belongs to $Q(C; \al)$ for $2\al \geq |e|$. Hence $\hopes(C; \al) \subset Q(C; \al)$. 
\qedhere
\end{proof}
 
The inclusion $\hopes(C; \al) \subset Q(C; \al)$ in Lemma~\ref{lem:hopes_in_complex} induces a homomorphism $f_*$ in $H_1$. Lemmas~\ref{lem:critical-addition} and~\ref{lem:critical-deletion} analyse what happens with $f_*$ when a critical edge $e$ is added to (or deleted from) $Q = Q(C; \al)$ and $G = \hopes(C; \al)$.

\begin{lemma}[addition of a critical edge]
\label{lem:critical-addition}
Let an inclusion $f:G \to Q$ of a graph $G$ into a simplicial complex $Q$ induce an isomorphism $f_*: H_1(G) \to H_1(Q)$. Between some vertices $u, v \in G$ let us add an edge $e$ to $G$ and $Q$ that creates $\ga \in H_1(Q \cup e)$. Then $f_*$ extends to an isomorphism $H_1(G \cup e) \to H_1(Q \cup e)$.
\end{lemma}

\begin{proof}
Let $L \subset G \cup e$ be a cycle containing the added edge $e$. Then $f$ extends to the inclusion $G \cup e \to Q \cup e$ and induces the isomorphism $H_1(G \cup e) \cong H_1(G) \oplus \langle[L]\rangle$. Now $f(L)$ is a cycle in $Q \cup e$ and $H_1(Q \cup e) \cong H_1(Q) \oplus \langle[f(L)]\rangle$. Mapping $[L]$ to $[f(L)] \in H_1(Q \cup e)$, we extend $f_*$ to a required isomorphism $H_1(G) \oplus \langle[L]\rangle \to H_1(Q) \oplus \langle[f(L)]\rangle$. 
\qedhere
\end{proof}
 
\begin{lemma}[deletion of a critical edge]
\label{lem:critical-deletion}
Let an inclusion $f: G \to Q$ of a graph $G$ into a simplicial complex $Q$ induce an isomorphism $f_*: H_1(G) \to H_1(Q)$. 
Let $\ga \in H_1(Q)$ die after adding a triangle $T$ to $Q$. 
Let $e$ be a longest open edge of the triangle $T$. Then $f_*$ descends to an isomorphism $H_1(G - e)\to H_1(Q \cup T)$.
\end{lemma}

\begin{proof}
Adding $T$ to $Q$ kills the homology class $\bd T$ of the boundary $\bd T \subset G$. Then $H_1(Q \cup T) \cong H_1(Q) / \langle[\bd T]\rangle$.  
Deleting the open edge $e$ from the boundary $\bd T \subset G$ makes the homology smaller: $H_1(G - e) \cong H_1(G) / \langle[\bd T]\rangle$. The isomorphism $f_*$ descends to the isomorphism $H_1(G) / \langle[\bd T]\rangle \to H_1(Q) / \langle[\bd T]\rangle$. \qedhere
\end{proof}

\begin{proposition}
\label{pro:hopes-homology}
The inclusion $\hopes(C; \al) \to Q(C; \al)$ from Lemma~\ref{lem:hopes_in_complex} induces an isomorphism of 1D homology groups: $H_1(\hopes(C; \al)) \to H_1(Q(C;\al))$.
\end{proposition}

\begin{proof}
$\hopes(C; 0)$ and the complex $Q(C;0)$ coincide with the cloud $C$, so their 1-dimensional homology groups are trivial. Each time a homology class is born or dies in $H_1(Q(C; \al))$, the isomorphism in $H_1$ induced by the inclusion $\hopes(C; \al) \to Q(C; \al)$ is preserved by Lemmas~\ref{lem:critical-addition} and \ref{lem:critical-deletion}. \qedhere
\end{proof}

\begin{lemma}
\label{lem:length>=birth}
Let a cycle $L \subset Q(C; \al)$ represent a homology class $\ga \in H_1(Q(C; \al))$ in $\PD\{Q(C; \al)\}$. Then any longest edge $e \subset L$ has length $|e| \geq 2 \cdot \birth(\ga)$.
\end{lemma}

\begin{proof}
Let a longest edge $e$ of a cycle $L \subset Q(C; \al)$ representing the class $\ga$ have a half-length $\al < \birth(\ga)$. Then $L$ enters $Q(C; \al)$ earlier than $\birth(\ga)$ and can not represent the class $\ga$ that starts living from $\al = \birth(\ga)$. \qedhere
\end{proof}

Recall that the forest $\MST(C; \al)$ on a cloud $C$ at a scale $\al$ is obtained from a minimum spanning tree $\MST(C)$ by removing all open edges longer than $2\al$.
An edge $e$ is \textit{splitting} a graph $G$ if removing the open edge $e$ makes $G$ disconnected. Otherwise the edge $e$ is called \textit{non-splitting} and belongs to a cycle of $G$.

\begin{proposition}
\label{pro:length-bound}
Let a graph $G \subset Q(C; \al)$ span $Q(C; \al)$ and $H_1(G) \to H_1(Q(C; \al))$ be an isomorphism induced by the inclusion. Let $(b_i, d_i)$, $i = 1, \dots, m$, be all dots in the persistence diagram $\PD\{Q(C; \al)\}$ such that $\birth \leq \al < \death$. Then the total length of $G$ is bounded below by the total length of edges of the forest $\MST(C; \al)$ plus $2\sum\limits_{i=1}^m b_i$.
\end{proposition}

\begin{proof}
Let the subgraph $G_1 \subset G$ consist of all non-splitting edges of $G$ and $e_1 \subset G_1$ be a longest open edge. Removing $e_1$ from $G$ makes $H_1(G)$ smaller. Hence there is a cycle $L_1 \subset G_1$ containing $e_1$ and representing a class $\ga_1 \in H_1(Q(C; \al))$ that corresponds to some off-diagonal dot in $\PD\{Q(C;\al)\}$, say $(b_1,d_1)$. Then $\ga_1$ lives over the interval $b_1 = \birth(\ga_1) \leq \al < d_1 = \death(\ga_1)$. Lemma~\ref{lem:length>=birth} implies that $|e_1| \geq 2b_1$.
Let the graph $G_2 \subset G - e_1$ consist of all non-splitting edges and $e_2 \subset G_2$ be a longest open edge. Similarly, find the corresponding point $(b_2, d_2)$, conclude that $|e_2| \geq 2b_2$ and so on until we get $\sum\limits_{i=1}^m |e_i| \geq 2\sum\limits_{i=1}^m b_i$. After removing all open edges $e_1, \dots, e_m$, the remaining graph $G - (e_1 \cup \dots \cup e_m)$ still spans the (possibly disconnected) complex $Q(\al)$. Indeed, each time we removed a non-splitting edge. So the total length of $G - (e_1 \cup \dots \cup e_m)$ is not smaller than the total length of $\MST(C; \al)$ by Lemma~\ref{lem:MST_optimal}. 
\qedhere
\end{proof}

\begin{theorem}[optimality of reduced skeletons]
\label{thm:optimality}
Fix a filtration $\{Q(C; \al)\}$ of complexes on a cloud $C$ in a metric space. 
For any $\al > 0$, the reduced skeleton $\hopes(C; \al)$ has the minimum total length of edges over all graphs $G \subset Q(C; \al)$ that span $Q(C; \al)$ and induce an isomorphism $H_1(G) \to H_1(Q(C; \al))$.
\end{theorem}

\begin{proof}
For any $\al > 0$, the inclusion $\hopes(C; \al)\to Q(C; \al)$ induces an isomorphism in $H_1$ by Proposition~\ref{pro:hopes-homology}. 
Let classes $\ga_1, \dots, \ga_m$ represent all $m$ dots counted with multiplicities in the `rectangle' $\{\birth \leq \al < \death\} \subset \PD\{Q(C; \al)\}$. Then $\ga_1, \dots, \ga_m$ form a basis of $H_1(Q(C; \al)) \cong H_1(\hopes(C; \al))$ by Definition~\ref{dfn:pers-diagram}.
The total length of $\hopes(C; \al)$ equals the length of $\MST(C; \al)$ plus $2\sum\limits_{i=1}^m \birth(\ga_i)$ by Lemma~\ref{lem:length=birth}.
By Proposition~\ref{pro:length-bound} this length is minimal over all graphs $G \subset Q(C; \al)$ that span $Q(C; \al)$ and have the same $H_1$ as $Q(C; \al)$. 
\qedhere
\end{proof}

\vspace*{-5mm}
\begin{figure}[h]
\centering
\includegraphics[width=\textwidth]{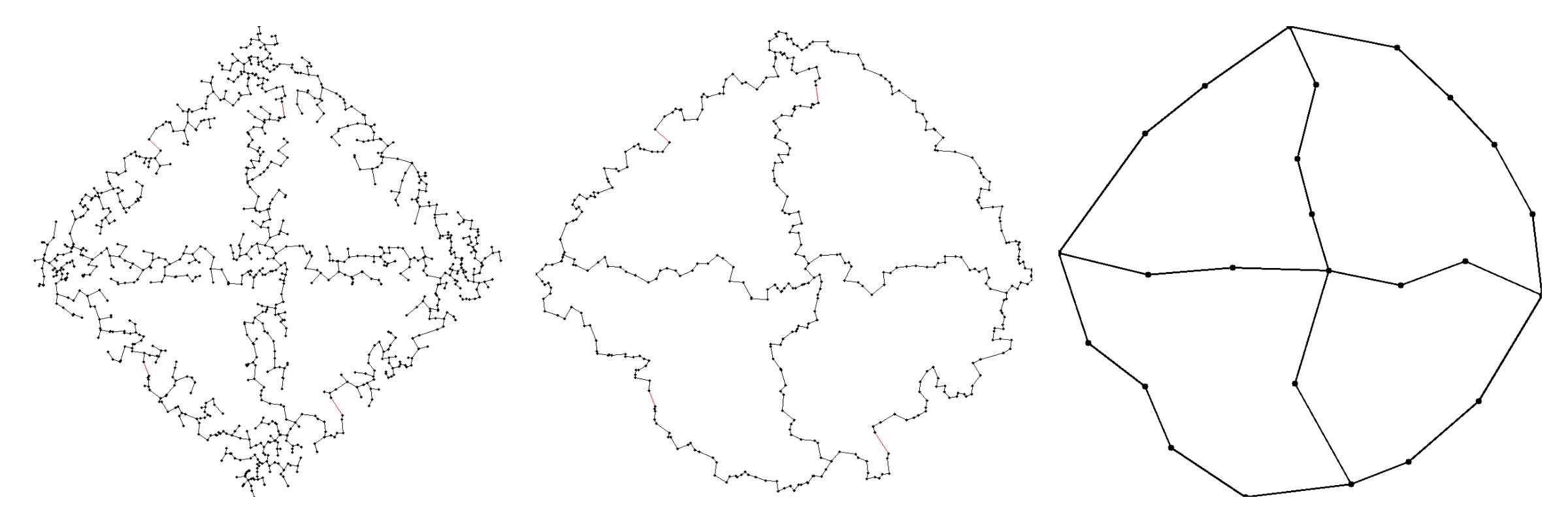}
\caption{
Left: derived $\hopes_{1,1}(C)$ for the sample of $W(4)$ in the first picture of Fig.~\ref{fig:clouds}. Middle: all degree 1 vertices have been removed. Right: $\shopes$ obtained by Algorithm~\ref{alg:contract}.}
\label{fig:hopes-simplification}
\end{figure}

\vspace*{-5mm}
\section{Guarantees for graph reconstructions using derived $\hopes$}
\label{sec:reconstruction}

This section shows that other subgraphs (derived skeletons) of HoPeS in Definition~\ref{dfn:derived-hopes} guarantee correct reconstructions in Theorems~\ref{thm:reconstruction_1stgap}, \ref{thm:reconstruction_derived}, see Fig.~\ref{fig:hopes-simplification}.

\begin{definition}[derived skeletons $\hopes_{k,l}$]
\label{dfn:derived-hopes}
Let $C$ be a finite cloud in a metric space. 
For $k, l \geq 1$, the \textit{derived skeleton} $\hopes_{k, l}(C)$ is obtained from $\hopes(C)$ by removing all edges that are longer than $2\uv_{k, l}(C)$ and by keeping only critical edges with $(\birth, \death) \in \vPD_{k, l}(C)$ and with $\death > \uv_{k,l}(C)$, see Fig.~\ref{fig:hopes-simplification}.
\bsquare
\end{definition}

\begin{definition}[radius $\rho$ of a cycle and thickness $\theta$ of a graph]
\label{dfn:radius and thickness}
For a graph $G$ in a metric space $M$, the \textit{radius} $\rho$ of a non-self-intersecting cycle $L \subset G^\al$ is the persistence of its corresponding homology class in the filtration, see Fig.~\ref{fig:radius_thickness}. 
When $\al$ is increasing, new holes can be born in $G^\al$. The \textit{thickness} $\theta(G)$ is the maximum persistence of any hole born after $\al = 0$ in the diagram $\PD\{G^\al\}$.
\bsquare
\end{definition}

If a cycle $L\subset\R^2$ encloses a convex region, the only hole of $L^{\al}$ completely dies when $\al$ equals the radius $\rho(L)$, so $\theta(L)=0$.
The `heart' cycle $L$ in Fig.~\ref{fig:radius_thickness} encloses a non-convex region, however no new holes are born in $L^{\al}$, so $\theta(L)=0$.
The figure-eight-shaped graph $G$ in Fig.~\ref{fig:radius_thickness} has a positive thickness $\theta(G)$ equal to the radius $\rho(L_1)$ of the largest cycle born in $G^{\al}$ when $\al$ is increasing.

\begin{lemma}
\label{lem:dotsondeathaxis}
For a graph $G$ in a metric space $M$, the 1D diagram $\PD\{G^\al\}$ has $m$ dots (with multiplicities) on the death axis, where $m=\text{rank}\; H_1(G)$. 
\end{lemma}

\begin{proof}
By Definition~\ref{dfn:birth-death} all $m$ dots $(0,d_i)$ in $\PD\{G^\al\}$ correspond to homology classes born at $\al=0$, hence present in $H_1(G)$ of the initial graph $G^0=G$.
\qedhere
\end{proof}

\begin{figure}
\centering
\begin{minipage}{0.52\textwidth}
\centering
\includegraphics[height=16mm]{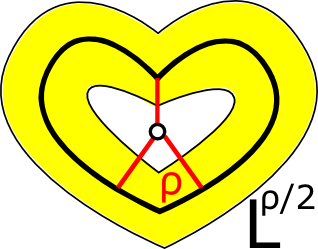}
\hspace*{1mm}
\includegraphics[height=16mm]{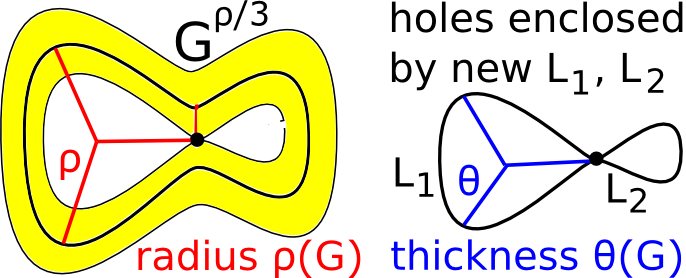}
\captionof{figure}{The `heart' graph has thickness $\theta=0$. The `figure-eight' graph has $\theta>0$, see Def.~\ref{dfn:radius and thickness}.}
\label{fig:radius_thickness}
\end{minipage} 
\hspace{4pt}
\begin{minipage}{0.43\textwidth}
\centering
\includegraphics[height=16mm]{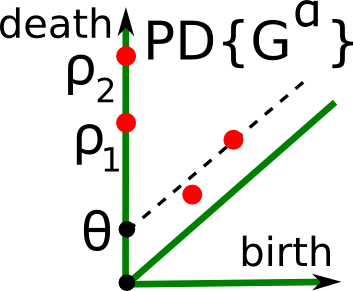}
\hspace*{8mm}
\includegraphics[height=16mm]{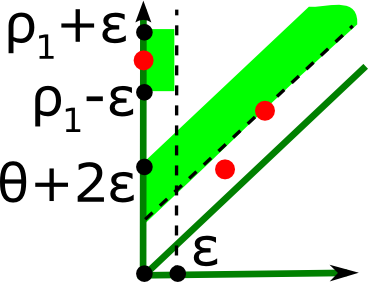}
\captionof{figure}{$\PD\{G^{\al}\}$ vs $\PD\{C^{\al}\}$ for an $\ep$-sample $C$ of $G$ satisfying Theorem~\ref{thm:reconstruction_1stgap}.}
\label{fig:diagram_perturbation}
\end{minipage} 
\end{figure}

\begin{theorem}
\label{thm:reconstruction_1stgap}
Let $M$ be a metric space, and let $C$ be any $\ep$-sample of a connected graph $G \subset M$ with $\text{rank}\; H_1(G) = m$. 
Let $m$ dots on the vertical axis of $\PD\{G^\al\}$ from Lemma~\ref{lem:dotsondeathaxis} have ordered deaths $y_1\leq\cdots\leq y_m$. 
If $y_1 > 7\ep + 2\theta(G) + \underset{i = 1,\ldots, m - 1}{\text{max}}\{y_{i + 1} - y_i\}$, we get the lower bound for noise $\uv_{1, 1}(C)\leq \ep$. 
The derived skeleton $\hopes_{1, 1}(C) \subset G^{2\ep}$ has the same homology $H_1$ as $G$.
\end{theorem}

\begin{proof}
Apart from the dots on the $y$-axis, all other dots in the 1D persistence diagram $\PD\{G^\al\}$ come from cycles born later in the $\al$-offsets $G^\al$. The maximum persistence of these dots is bounded above by the thickness $\theta(G)$, see Fig.~\ref{fig:diagram_perturbation}.
\smallskip

The inequality $y_1 > 7\ep + 2\theta(G) + \underset{i = 1, ..., m - 1}{\text{max}}\{y_{i + 1} - y_i\}$ guarantees that the gap $\{\theta(G) < y - x < y_1\}$ in $\PD\{G^\al\}$ is wider than any other gaps including the higher gaps $y_{i + 1} - y_i$ between the dots $(0, y_i) \in \PD\{G^\al\}, i = 1, \dots, m - 1$, since the inequality implies $y_1 - \theta(G) > \underset{i = 1, ..., m - 1}{\text{max}}\{y_{i + 1} - y_i\}$ and $y_1 - \theta(G) > \theta(G)$.
\smallskip

By Stability Theorem~\ref{thm:stability}, the perturbed diagram $\PD\{C^\al\}$ is in the $\ep$-offset of $\PD\{G^\al\} \subset \cup_{i = 1}^m (0, y_i ) \cup \{y - x < \theta(G)\}$ with respect to the $L_\infty$ metric on $\mathbb{R}^2$. 
All noisy dots near the diagonal in $\PD\{C^\al\}$ can not be higher that $\theta(G) + 2\ep$ after projecting along the diagonal $\{x = y\}$ to the vertical axis $\{x = 0\}$. 
The remaining dots can not be lower than $y_1 - 2\ep$ after the same projection $(x, y) \to y - x$. 
So the smaller diagonal strip $\{\theta(G) + 2\ep < y - x < y_1 - 2\ep\}$ of the vertical width $y_1 - 4\ep - \theta(G)$ is empty in the perturbed diagram $\PD\{C^\al\}$.
\smallskip

By Stability Theorem~\ref{thm:stability}, any dot $(0, y_i) \in \PD\{G^\al\}, i \geq 1$, can not jump lower than the line $y - x = y_i - 2\ep$ or higher than $y - x = y_i + \ep$ (not $y - x = y_i + 2\ep$ because the dot is on the $y$-axis and hence can not move in the negative $x$ direction). 
Then the widest diagonal gap between these perturbed dots has a vertical width at most $\underset{i = 1, ..., m - 1}{\text{max}}\{y_{i + 1} - y_i\} + 3\ep$, see the last picture of Fig.~\ref{fig:diagram_perturbation}.
All dots near the diagonal have diagonal gaps not wider than $\theta(G) + 2 \ep$.
The given inequality $y_1 - 4\ep - \theta(G) > \theta(G) + \underset{i = 1,\ldots, m - 1}{\text{max}}\{y_{i + 1} - y_i\}+3\ep$ implies that all other diagonal gaps have a vertical width smaller than $y_1 - 4\ep - \theta(G)$. 
\smallskip

Hence the first widest gap in $\PD\{C^\al\}$ covers the diagonal strip $\{\theta(G) + 2\ep < y - x < y_1 - 2\ep\}$ within the first widest gap $\{\theta(G) < y -x < y_1\}$ in the original diagram $\PD\{G^\al\}$.
Then the subdiagram $\dPD_1(C^\al)$ above the line $y - x = y_1 - 2\ep$ contains exactly $m$ perturbations of the original dots $(0, y_i)$ in the vertical strip $\{0 \leq x \leq \ep\}$. Hence the reduced skeleton $\hopes_{1, 1}(C)$ contains exactly m critical edges corresponding to all m dots of the subdiagram $\dPD_1\{C^\al\}$.
\smallskip

It remains to prove that $\hopes_{1, 1}(C)$ is $2\ep$-close to $G$. Let the critical scale $\al(C)$ be the maximum birth over all dots in $\dPD_1\{C^\al\}$. These dots are at most $\ep$ away from their corresponding points $(0, y_i)$ in the vertical axis, so critical scale $\al(C)$ is at most $\ep$. A longest edge $e$ of any cycle in $\hopes_{1, 1}(C)$ persisting over birth $\leq \al <$ death has the half-length equal to its birth, because adding $e$ gave birth to the homology class. Then all edges in $\hopes_{1, 1}(C)$ have half-lengths at most $\al(C) \leq \ep$. 
Hence $\hopes_{1, 1}(C)$ is covered by the disks that have the radius $\ep$ and centres at all points of $C$, so $\hopes_{1, 1}(C) \subset C^\ep \subset G^{2\ep}$. \qedhere 
\end{proof}

Reconstruction Theorem~\ref{thm:reconstruction_derived} for other derived subgraphs $\hopes_{k,l}(C)$ will follow from Lemmas~\ref{lem:diagonal-subdiagram}, \ref{lem:vertical-subdiagram}, \ref{lem:derived-hopes}, \ref{lem:approximation} and will also imply Corollary~\ref{cor:stability_hopes}.

\begin{lemma}
\label{lem:diagonal-subdiagram}
Let $C$ be any $\ep$-sample of a connected graph $G$ in a metric space. If $|\dgap_k(G)| - |\dgap_{k + 1}(G)| > 8\ep$, then there is a bijection $\psi: \dPD_k(G) \to \dPD_k(C)$ so that $||q - \psi(q)||_{\infty}\leq \ep$ for all $q \in \dPD_k(G)$.
\end{lemma}

\begin{proof}
By Stability Theorem~\ref{thm:stability} there is a bijection $\psi: \PD\{G^{\al}\} \to \PD\{C^{\al}\}$ such that $q, \psi(q)$ are $\ep$-close in the $L_{\infty}$-distance on $\R^2$ for all $q \in \PD\{G^{\al}\}$. 
The $\ep$-neighbourhood of a dot $q = (x, y)$ in the $L_{\infty}$ distance is the square $[x - \ep, x + \ep] \times [y - \ep, y + \ep]$. 
Under the diagonal projection $\pr(x, y) = y - x$, this square maps to the interval $[y - x - 2\ep, y - x + 2\ep]$. Hence any diagonal gap $\{a < y - x < b\}$ in $\PD\{G^{\al}\}$ can become thinner or wider in $\PD\{C^{\al}\}$ by at most $4\ep$ due to dots $q$ `jumping' to $\psi(q)$ by at most $2\ep$ under the projection $\pr$. 
\smallskip

By the given inequality, the first $k$ widest gaps $\dgap_i(G)$ in $\PD\{G^{\al}\}$ for $i = 1, \dots, k$ are wider by at least $8\ep$ than all other $\dgap_i(G)$ for $i > k$.
Hence all dots between any two successive gaps from the first $k$ widest can not `jump' over these wide gaps and remain `trapped' between corresponding diagonal gaps in the perturbed diagram $\PD\{C^{\al}\}$.  
Although the order of the first $k$ widest gaps $\dgap_i(G)$, $i = 1, \dots, k$, may not be preserved under $\psi$, the lowest $\{a < y - x < b\}$ of these $k$ diagonal gaps is respected by $\psi$ as follows. The thinner strip $S = \{a + 2\ep < y - x < b - 2\ep\}$ has no dots from $\PD\{C^{\al}\}$ and has the vertical width $|S| \geq |\dgap_k(G)| - 4\ep > |\dgap_{k + 1}(G)| + 4\ep \geq |\dgap_{k + 1}(C)|$.
\smallskip

Then the diagonal strip $S = \{a + 2\ep < y - x < b - 2\ep\}$ is within the lowest gap of the first $k$ widest gaps $\dgap_i(C)$, $i = 1, \dots, k$. Hence all dots above $S$ remain above $S$ under the bijection $\psi$. By Definition~\ref{dfn:diagonal-gaps}, all these dots above $S$ in $\PD\{G^{\al}\}$ and in $\PD\{C^{\al}\}$ form the diagonal subdiagrams $\dPD_k(G)$ and $\dPD_k(C)$ respectively. So $\psi$ descends to a bijection $\dPD_k(G) \to \dPD_k(C)$. 
\qedhere
\end{proof}

\begin{lemma}
\label{lem:vertical-subdiagram}
Let $\psi: \dPD_k(G) \to \dPD_k(C)$ be a bijection such that $||q - \psi(q)||_{\infty} \leq \ep$ holds for all dots $q \in \dPD_k(G)$ as in Lemma~\ref{lem:diagonal-subdiagram}. If $|\vgap_{k, l}(G)| - |\vgap_{k, l + 1}(G)| > 4\ep$ for some $l \geq 1$, then $\psi$ restricts to a bijection $\vPD_{k, l}(G) \to \vPD_{k, l}(C)$.
\end{lemma}

\begin{proof}
The $x$-coordinate of any dot $q \in \dPD_k(G)$ changes under the given bijection $\psi$ by at most $\ep$. 
Similarly to the proof of Lemma~\ref{lem:diagonal-subdiagram} each vertical gap $\vgap_{k, l}(G)$ becomes thinner or wider by at most $2\ep$. 
By the given inequality the first $l$ gaps $\vgap_{k, j}(G)$ in $\PD_k\{G^{\al}\}$ for $j = 1,\dots, l$ are wider by at least $4\ep$ than all other $\vgap_{k, j}(G)$ for $j > l$. Hence all dots between any two successive gaps from the first $k$ widest can not `jump' over these wide gaps and remain `trapped' between corresponding vertical gaps in the perturbed diagram $\PD\{C^{\al}\}$.  
\smallskip

Although the order of the first $l$ widest $\vgap_{k, j}(G)$, $j = 1, \dots, l$, may not be preserved under $\psi$, the leftmost $\{a < x < b\}$ of these $l$ vertical gaps is respected by $\psi$ in the following sense. 
The thinner strip $S = \{a + \ep < x < b - \ep\}$ contains no dots from $\dPD_k(C)$ and has the horizontal width $|S| \geq |\vgap_{k, l}(G)| - 2\ep > |\dgap_{k, l + 1}(G)| + 2\ep \geq |\dgap_{k, l + 1}(C)|$.
Then the vertical strip $S = \{a + \ep < x < b - \ep\}$ is within the leftmost of the first $l$ widest $\vgap_{k, j}(C)$, $j = 1, \dots, l$. Hence all dots to the left of $S$ remain to the left of $S$ under the bijection $\psi$. 
By Definition~\ref{dfn:vertical-gaps}, all these dots to the left of $S$ in $\PD\{G^{\al}\}$ and in $\PD\{C^{\al}\}$ form the vertical subdiagrams $\vPD_{k, l}(G)$ and $\vPD_{k, l}(C)$ respectively. 
So $\psi$ descends to a bijection $\vPD_{k, l}(G) \to \vPD_{k, l}(C)$ between smaller vertical subdiagrams. 
\qedhere
\end{proof}

\begin{lemma}
\label{lem:derived-hopes}
The derived skeleton $\hopes_{k, l}(C)$ is within $\hopes(C; \uv_{k, l}(C))$.
\end{lemma}

\begin{proof}
By Definition~\ref{dfn:hopes} all edges of the reduced skeleton $\hopes(C; \uv_{k, l}(C))$ have a half-length at most $\uv_{k, l}(C)$ and $\death > \uv_{k, l}(C)$ for all critical edges. Definition~\ref{dfn:derived-hopes} also imposes the extra restriction on critical edges of $\hopes_{k, l}(C)$, namely each dot $(\birth, \death)$ is in the vertical subdiagram $\vPD_{k, l}(C)$. So $\hopes_{k, l}(C) \subset \hopes(C; \uv_{k, l}(C))$. 
\qedhere
\end{proof}

\begin{lemma}[approximation by reduced $\hopes$]
\label{lem:approximation}
Let $C$ be any finite $\ep$-sample of a subspace $G$ in a metric space $M$. Then the reduced skeleton $\hopes(C; \al)$ for any scale $\al > 0$ is contained within the $(\ep + \al)$-offset $G^{\ep + \al} \subset M$.
\end{lemma}

\begin{proof}
Any edge $e\subset\hopes(C; \al)$ has a length at most $2\al$ by Definition~\ref{dfn:hopes}, hence is covered by the balls with the radius $\al$ and centres at the endpoints of $e$. 
Then $\hopes(C; \al) \subset C^{\al} \subset G^{\ep + \al}$ since $C \subset G^{\ep}$ is an $\ep$-sample of $G$. 
\qedhere
\end{proof}

Theorem~\ref{thm:reconstruction_derived} guarantees a reconstruction with a correct homology $H_1$ and within a $2\ep$-offset for any $\ep$-sample of a graph $G$ with $k$ cycles whose $\PD\{G^{\al}\}$ has $k$ dots separable by the $k$-th widest diagonal gap from noisy artefacts.
This is formalised below in the 4 conditions, which fail if the noise level $\ep$ is too high.

\begin{theorem}[reconstruction by derived skeletons]
\label{thm:reconstruction_derived}
Let $C$ be any $\ep$-sample of an unknown graph $G$ in a metric space. Let $G$ satisfy the conditions below. 
\smallskip

\noindent 
(1) All cycles $L \subset G$ are `persistent', i.e. $\text{death}(L) \geq \ud_k(G)$ for some $k \geq 1$. 
\smallskip

\noindent 
(2) $|\dgap_k(G)|$ `jumps', i.e. $|\dgap_k(G)| - |\dgap_{k + 1}(G)| > 8\ep$ for $k$ from (1). 
\smallskip

\noindent 
(3) No cycles are born in offsets $G^{\al}$ for small $\al $, i.e. $\uv_{k, l}(G) = 0$ for some $l$. 
\smallskip

\noindent
 (4) $|\vgap_{k, l}(G)|$ `jumps': $|\vgap_{k, l}(G)| - |\vgap_{k, l + 1}(G)| > 4\ep$ for $k, l$ in (1),(3). 
 \smallskip

\noindent Then we get the lower bound for noise $\ep\geq\uv_{k, l}(C)$ and the derived skeleton $\hopes_{k, l}(C) \subset G^{2\ep}$ has the same homology $H_1$ as the underlying graph $G$.
\end{theorem}

\begin{proof}
Due to condition~(2), Lemma~\ref{lem:diagonal-subdiagram} implies that there is a bijection $\psi: \dPD_{k}(G) \to \dPD_k(C)$ so that $||q - \psi(q)||_{\infty} \leq \ep$ for all $q \in \dPD_{k}(G)$. Lemma~\ref{lem:vertical-subdiagram} due to condition~(4) implies that $\psi$ descends to a bijection $\vPD_{k,l}(G) \to \vPD_{k, l}(C)$. All cycles in a graph $G$ give birth to corresponding homology classes in $H_1(G^{\al})$ at the scale $\al = 0$. These classes may split later at $\al > 0$, but will eventually die and always give dots $(0, \death) \in \PD\{G^{\al}\}$ in the vertical death axis. For any cycle $L \subset G$, let $\death(L)$ be the maximum $\al$ such that
$H_1(G^{\al})$ has the class $[L]$. 
The graph $\theta$ in Fig.~\ref{fig:theta+offsets+persistence} has 3 cycles with $\death(L)\approx 2.577$. 
\smallskip

Let $L_1,\dots,L_m\subset G$ be all $m$ cycles generating $H_1(G)$. Then the 1D diagram $\PD\{G^{\al}\}$ contains $m$ dots $(0, \death(L_i))$, $i = 1, \dots, m$, because each class $[L_i]$ persists over $0 \leq \al < \death(L_i)$ by Definition~\ref{dfn:pers-diagram}. Condition (1) implies that all dots $(0, \death(L_i))$ 
belong to the subdiagram $\dPD_k(G)$, hence to $\vPD_{k, l}(C)$.
\smallskip

Condition (3) $\uv_{k, l}(G) = 0$ means that the leftmost of the first $l$ widest $\vgap_{k, j}(G)$, $j = 1, \dots, l$, is attached to the vertical death axis in $\PD\{G^{\al}\}$, which should contain the vertical subdiagram $\vPD_{k, l}(G)$. So $\vPD_{k, l}(G)$ consists of the $m$ dots $(0, \death(L_i))$, $i = 1, \dots, m$. 
Then the vertical subdiagram $\vPD_{k, l}(C)$ for the cloud $C$ also has exactly $m$ dots, which are `noisy' images $\psi(0, \death(L_i))$, $i = 1, \dots, m$. 
Moreover, all these dots in the vertical subdiagram $\vPD_{k, l}(C)$ are at most $\ep$ away from the vertical axis, so the vertical scale $\uv_{k, l}(C)$ is at most $\ep$.

The lowest of the points $\psi(0, \death(L_i))$ has a death with the lower bound $\min\limits_{i = 1, \dots, m} \death(L_i) - \ep \geq \dPD_k(G) - \ep > |\dgap_k(G)| - \ep > 7\ep > \uv_{k,l}(C)$. 

So the condition $\death > \uv_{k, l}(C)$ from Definition~\ref{dfn:derived-hopes} is satisfied. Hence the reduced skeleton $\hopes_{k, l}(C)$ contains exactly $m$ critical edges corresponding to all $m$ dots in the subdiagram $\vPD_{k, l}(C)$. 
All these $m$ critical edges of $\hopes_{k, l}(C)$ generate $H_1$ of the required rank $m$. The geometric approximation $\hopes_{k, l}(C) \subset G^{2\ep}$ follows from Lemmas~\ref{lem:derived-hopes},~\ref{lem:approximation} for $\al = \uv_{k, l}(C) \leq \ep$. 
\qedhere
\end{proof}

The first attempt in a graph reconstruction from a cloud $C$ by Theorem~\ref{thm:reconstruction_derived} can be the derived skeleton $\hopes_{1,1}(C)$.
If the conditions fail for $k=l=1$, Theorem~\ref{thm:reconstruction_derived} might work for other pairs of the parameters $k,l$, see Fig.~\ref{fig:42049_hopes}.

\begin{corollary}
\label{cor:stability_hopes}
In the conditions of Theorem~\ref{thm:reconstruction_derived}, if another cloud $\tilde C$ is $\de$-close to $C$, then the perturbed skeleton $\hopes_{k, l}(\tilde C)$ is $(2\de + 4\ep)$-close to $\hopes_{k, l}(C)$.
\end{corollary}

\begin{proof}

The condition that the perturbed cloud $\tilde C$ is $\de$-close to the original cloud $C$, which is $\ep$-close to the graph $G$, implies that $\tilde C$ is $(\de + \ep)$-close to $G$.

Reconstruction Theorem~\ref{thm:reconstruction_derived} for the $\ep$-sample $C$ and $(\de + \ep)$-sample $\tilde C$ of $G$ says that $\hopes_{k, l}(C)$ is $2\ep$-close to $G$ and $\hopes_{k, l}(\tilde C)$ is $(2\de + 2\ep)$-close to $G$. Hence $\hopes_{k, l}(C)$ and $\hopes_{k, l}(\tilde C)$ are $(2\de + 4\ep)$-close as required. \qedhere

\end{proof}

\section{The 79K Dataset of Random Noisy Point Samples of Planar Graphs}
\label{sec:dataset}

Mapper, $\al$-Reeb and $\hopes$ have guarantees on the homotopy type (or the number of independent cycles in homology $H_1$) of reconstructed graphs, not on a more advanced homeomorphism type that captures branches or filaments outside closed cycles. 
Hence, for experimental comparisons, it is enough to generate noisy samples of planar graphs without degree~1 vertices.
\smallskip

Subsection~\ref{sub:patterns} introduces three patterns of planar graphs around which point clouds are randomly sampled: wheel, grid and hexagonal, see examples in Fig.~\ref{fig:graphs}.
Subsection~\ref{sub:noise} discusses two noise models: uniform and Gaussian.

\subsection{Patterns: families of connected planar graphs for generating noisy samples}
\label{sub:patterns}

\smallskip
\noindent
$\bullet$
The {\em $k$-wheel} graph $W(k)\subset\R^2$ has $k\geq 3$ circumference vertices equally distributed along the unit circle centred at the central vertex at the origin. 
$W(k)$ has edges between the central vertex and all circumference vertices, and also edges between successive circumference vertices, see $W(4)$ in Fig.~\ref{fig:graphs}.
\smallskip

\noindent
$\bullet$
For $k,l\geq 1$, the {\em $(k,l)$-grid} graph $G(k,l)$ has vertices at the points $(i,j)$ with integer coordinates $0\leq i\leq k$, $0\leq j\leq l$.
Each vertex $(i,j)$ is connected to up to 4 neighbours $(i\pm 1,j)$, $(i,j\pm 1)$ in the rectangle $[0,k]\times[0,l]\subset\R^2$, see Fig.~\ref{fig:graphs}.\smallskip

\noindent
$\bullet$
The \textit{$k$-hexagons} graph $H(k)$ consists of the boundaries of $k$ regular hexagons with edges of unit length, e.g. $H(1)$ is one regular hexagon.
The $(k + 1)$-hexagons graph is obtained from the $k$-hexagons graph by adding the boundary of a new hexagon.
The last picture of Fig.~\ref{fig:graphs} shows the order of added hexagons for $k\leq 7$.
For $k>7$, more hexagons are similarly added in circular layers.

\subsection{Generating random points in edges of a graph according to length}
\label{sub:generate}

For a fixed graph $G\subset\R^2$ from one of the patterns above, each random point is uniformly generated along edges of $G$ as follows.
Fix any order of edges of $G$. 
Then the continuous graph $G$ is parameterised by a single variable $t$ that takes values in the interval $[0,\sum\limits_{i=1}^k l_i]$, where $l_1,\dots,l_k$ are the edge-lengths of $G$.

If the uniform random variable $t$ belongs to the $j$-th interval $[\sum\limits_{i=1}^{j-1} l_i, \sum\limits_{i=1}^{j} l_i]$, a point is chosen in the $j$-th edge as the weighted combination $w\vec u+(1-w)\vec v$ of the endpoints $\vec u,\vec v\in\R^2$ with $w=(t-\sum\limits_{i=1}^{j-1} l_i)/l_{j}$.
We sample 100 points per unit length of the embedded graph $G\subset\R^2$, e.g. 400 points for $G_{1,1}$.

\subsection{Noise models: uniform and Gaussian perturbations of points}
\label{sub:noise}

For each sampled point $p$ lying  on an edge $e$ of a graph, we generate two independent random shifts $d_e,d_{\perp}$ with the same distribution as follows.
\smallskip

\noindent
$\bullet$
{\em Uniform noise} with an upper bound $\mu$: $d_e, d_{\perp}$ are uniform in $[-\mu, \mu]$.
\smallskip

\noindent
$\bullet$
{\em Gaussian noise} with mean 0 and a standard deviation $\sigma$: $d_e, d_{\perp}\in\R$ have the Gaussian density $f(t, \sigma) = \frac{1}{\sqrt{2\pi\sigma^2}}e^{-\frac{t^2}{2\sigma^2}}$ for $t\in\R$. 
\smallskip

Then the point $p\in e$ is shifted by $d_e$ units parallel to $e$ (in one of two fixed directions), and by $d_{\perp}$ units in the straight line perpendicular to $e$, see Fig.~\ref{fig:clouds}.

\begin{figure}
\includegraphics[width = 0.49\textwidth]{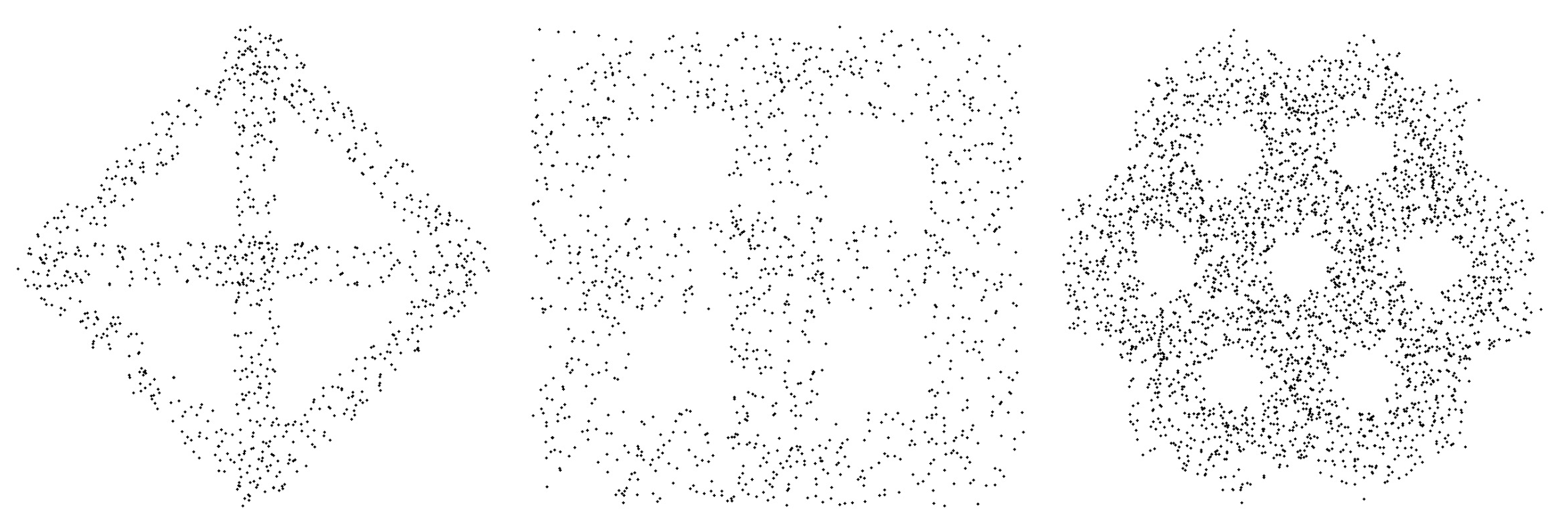}
\includegraphics[width = 0.49\textwidth]{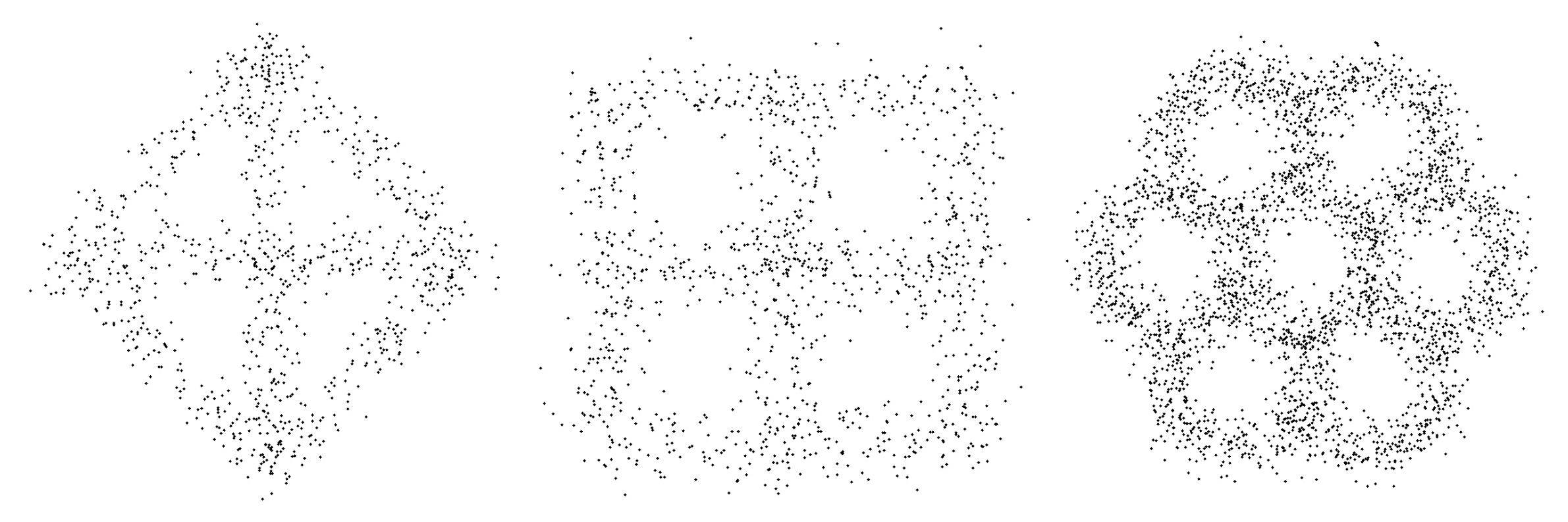}
	\caption{
		\textbf{First three}: samples with uniform noise and bounds $\mu = 0.1$, $\mu = 0.25$, $\mu = 0.5$.
		{\bf Last three}: samples with Gaussian noise and deviations $\sigma = 0.08$, $\sigma = 0.12$, $\sigma = 0.2$.
	}
	\label{fig:clouds}
\end{figure}

\subsection{Types of randomly sampled clouds obtained by varying parameters}
\label{sub:types_clouds}

The Planar Graph Cloud (PGC) dataset contains 79K clouds, 200 clouds for each of the following 395 types of clouds counted over 3 graph patterns.
\smallskip

\noindent
$\bullet$
70 wheel types: $k$-wheel graphs $W_k$ have 7 values of $k$ from $3$ to $9$; the bound $\mu$ of uniform noise has 5 values from $0.05$ to $0.25$ in $0.05$ intervals;
the deviation $\si$ of Gaussian noise has 5 values from $0.02$ to $0.1$ 
in $0.02$ intervals.
\smallskip

\noindent
$\bullet$
108 grid types: grids have 6 pairs $(k,l)=(1,1),(2,1),(3,1),(2,2),(3,2),(3,3)$;
the bound $\mu$ of uniform noise has 8 values from $0.05$ to $0.4$ in 
$0.05$ intervals;
the deviation $\si$ of Gaussian noise has 10 values from $0.02$ to $0.2$ in $0.02$ intervals.
\smallskip

\noindent
$\bullet$
217 hexagonal types: the $k$-triangles graph $H_k$ has 7 values of $k$ 
from 1 to 7;
the bound $\mu$ of uniform noise has 15 values from $0.05$ to $0.75$ in $0.05$ intervals;
the deviation $\si$ of Gaussian noise has 16 values from $0.02$ to 
$0.32$ in $0.02$ intervals.
\smallskip

The above intervals for noise parameters were chosen to report maximum noise when the skeletonisation algorithms produce correct reconstructions (for the 1st Betti number) with high success rates, see Fig.~\ref{fig:noise90},~\ref{fig:noise95} in section~\ref{sec:experiments}.

\section{Drawing and simplifying skeletons of point clouds}
\label{sec:simplification}

By Lemma~\ref{lem:hopes_in_complex}, for any cloud $C\subset\R^d$ and the filtration $\{C^{\al}\}$ of $\al$-offsets, $\hopes(C)$ is a subgraph of a Delaunay triangulation $\Del(C)\subset\R^d$.
Hence the skeleton $\hopes(C)$ is embedded into $\R^d$ (has no intersection of edges by definition), i.e. visualises the shape of $C$ directly in the space where $C$ lives. 

\subsection{Drawing the abstract Mapper and $\al$-Reeb graphs in the plane}

Both Mapper and $\al$-Reeb outputs are abstract graphs, and so do not naturally embed in $\R^d$. Therefore, in order to draw the graphs in the plane when $C\subset\R^2$, it is necessary for us to project them to $\R^2$ as naturally as possible.
\smallskip

Each vertex $v$ of the Mapper graph corresponds to a cluster $C_v$ in a given cloud $C\subset\R^d$ and is naturally placed at the geometric centre $\frac{1}{|C_v|}\sum\limits_{p \in C_v} p$.
\smallskip

For $\al$-Reeb graphs, it is far less natural to appropriately embed them in $\R^2$, because final stage~5 in section~\ref{sub:Reeb} identifies several disjoint intervals. 
Every central point $v$ of such an interval $J_v$ is mapped to the geometric centre of the corresponding connected subgraph similarly to the Mapper graph.
\smallskip

In all cases the edges between vertices are drawn as straight-line segments.
\smallskip

Since the algorithms are compared on noisy samples around planar graphs without degree~1 vertices, the outputs of the 3 algorithms Mapper, $\al$-Reeb and $\hopes$ before comparison will be simplified by removing all filamentary branches in subsection~\ref{sub:simplification}.
Then the $\al$-Reeb graph in Fig.~\ref{fig:alpha-Reeb} will become a topological circle that better resembles the original noisy sample of an ellipse.

\subsection{Simplification of output skeletons by pruning branches}
\label{sub:simplification}

After running all algorithms, we iteratively remove all degree 1 vertices (with their single edge), to get a graph without degree 1 vertices, see Fig.~\ref{fig:hopes-simplification}, \ref{fig:mapper-simplification}, \ref{fig:aR-simplification}.
\smallskip

\begin{figure}[h]
\centering
\begin{minipage}{0.47\textwidth}
\centering
\includegraphics[scale = 0.1]{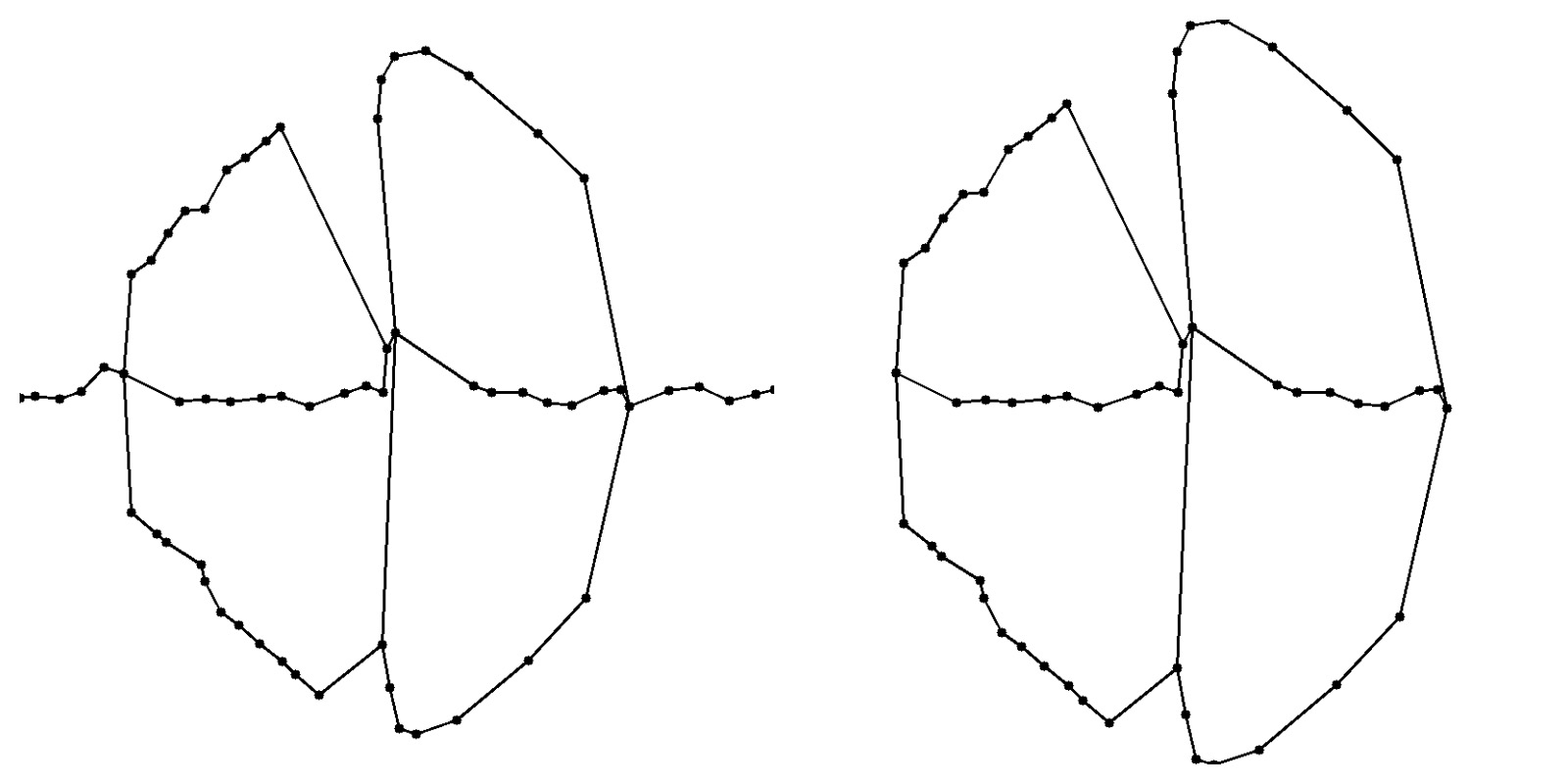}
\captionof{figure}{Left: an original Mapper output for the noisy sample of $W(4)$ in the first picture of Fig.~\ref{fig:clouds}. Right: the simplified Mapper output by pruning vertices of degree 1.}
\label{fig:mapper-simplification}
\end{minipage} \hspace{1em}
\begin{minipage}{0.47\textwidth}
\centering
\includegraphics[scale = 0.1]{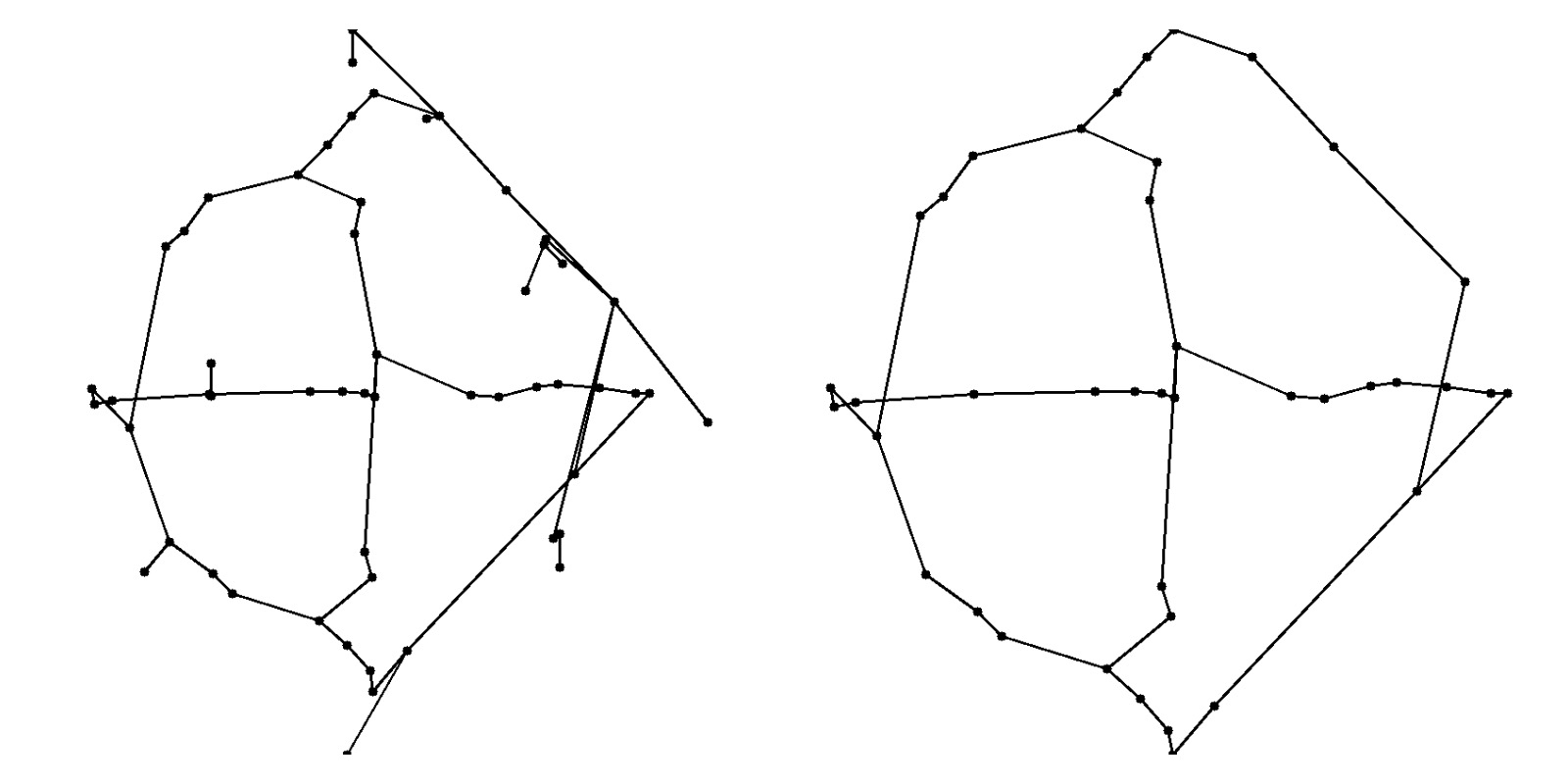}
\captionof{figure}{Left: an original $\al$-Reeb output the noisy sample of $W(4)$ in the first picture of Fig.~\ref{fig:clouds}. Right: the simplified $\al$-Reeb output by pruning vertices of degree 1.}
\label{fig:aR-simplification}
\end{minipage}
\end{figure}

By Definition~\ref{dfn:hopes}, $\hopes(C)$ contains $\MST(C)$, which spans all points of a cloud $C$.
Both $\MST(C)$ and $\hopes(C)$ approximate $C$ with zero geometric error, i.e. any point of $C$ is at distance 0 from $\hopes(C)$.
Hence pruning all branches in $\hopes(C)$ above will make the experimental comparison fairer. 
\smallskip

An ideal output should have much fewer vertices than a given point cloud $C$.
We have decided to simplify $\hopes$ by Algorithm~\ref{alg:contract} below, which preserves the homology by Corollary~\ref{cor:simhopes_homology}.
This extra simplification is applied only to $\hopes$ and makes a geometric approximation worse, while Mapper and $\al$-Reeb graphs keep their geometric errors with many more intermediate vertices.

\begin{alg}
\label{alg:contract}
For a given threshold $\ep$, we iteratively collapse all edges of $\hopes(C)$ with lengths less than $\ep$, starting from shortest edges as follows. 
\smallskip
 
\noindent
\textbf{Step (\ref{alg:contract}a)}. 
Any edge shorter than $\ep$ and contained in a triangular cycle will not be collapsed by steps below in order to preserve the homology group $H_1$.
\smallskip
 
\noindent
\textbf{Step (\ref{alg:contract}b)}. 
To collapse an edge $e$ with endpoints $v_1, v_2$, we first remove $v_1, v_2$ and all edges incident on them. 
We now add a new vertex $v$ as follows.
\smallskip
 
\noindent
\textbf{Step (\ref{alg:contract}c)}. 
If $\text{deg}(v_1) = \text{deg}(v_2) = 2$, or $\text{deg}(v_1) \neq 2$ and $\text{deg}(v_2) \neq 2$, the new vertex $v$ is placed at the midpoint of the straight-line edge between $v_1,v_2$. 
\smallskip
 
\noindent
\textbf{Step (\ref{alg:contract}d)}. 
If (say) $\deg(v_1) \neq 2$ but $\deg(v_2) = 2$, then $v$ is put at the original location of $v_1$, to best preserve the geometry of the skeleton $\hopes(C)$. 
\smallskip
 
\noindent
\textbf{Step (\ref{alg:contract}e)}. 
We add an edge from the new vertex $v$ to each of the original neighbours of the old vertices $v_1$ and $v_2$, see the final picture in Fig.~\ref{fig:hopes-simplification}.  
\smallskip
 
\noindent
\textbf{Step (\ref{alg:contract}f)}. 
If any edges become intersecting, the collapse is reversed.
\smallskip

The output of Algorithm~\ref{alg:contract} can be denoted by $\shopes_\ep(C)$. 
We always use $\ep$ equal to the maximum death of any dot above the first widest diagonal gap in the persistence diagram $\PD\{C^\al\}$ and denote the result by $\shopes(C)$ for simplicity.
This choice of $\ep$ is justified by the fact that maximum death is the scale $\al$ when all (most persistent) cycles of $\hopes(C)$ become contractible in the offset $C^{\al}$, so there is no sense to collapse edges longer than $\al$.
\end{alg}

\begin{corollary}
\label{cor:simhopes_homology}
Under the conditions of Theorems \ref{thm:reconstruction_1stgap} and \ref{thm:reconstruction_derived}, $\shopes(C)$ from Algorithm~\ref{alg:contract} has the same homology $H_1$ as the underlying graph $G$.
\end{corollary}

\begin{proof}
Removing degree 1 vertices does not change the homology group $H_1$ of a graph. 
Collapsing a short edge in Algorithm~\ref{alg:contract} would only change $H_1$ if this edge was in a triangular cycle or, if considering the output as an embedded graph in $\R^2$, we get intersections of edges.
Algorithm~\ref{alg:contract} prevents both operations by Steps (\ref{alg:contract}a) and (\ref{alg:contract}f).
Hence $H_1$ remains unchanged under all simplifications and is isomorphic to $H_1(G)$ in both cases specified by Theorems \ref{thm:reconstruction_1stgap}, \ref{thm:reconstruction_derived}. 
\qedhere
\end{proof}

Fig.~\ref{fig:wheel5},~\ref{fig:grid33},~\ref{fig:hex6} 
show typical outputs. 
Since Mapper and $\al$-Reeb graphs are abstract, their projections to $\R^2$ can have self-intersections. 
The full skeleton $\hopes(C)$ and simplified skeleton $\shopes(C)$ have no self-intersections.

\begin{figure}[H]
\centering
\def\svgwidth{\columnwidth}
\includegraphics[width=\textwidth]{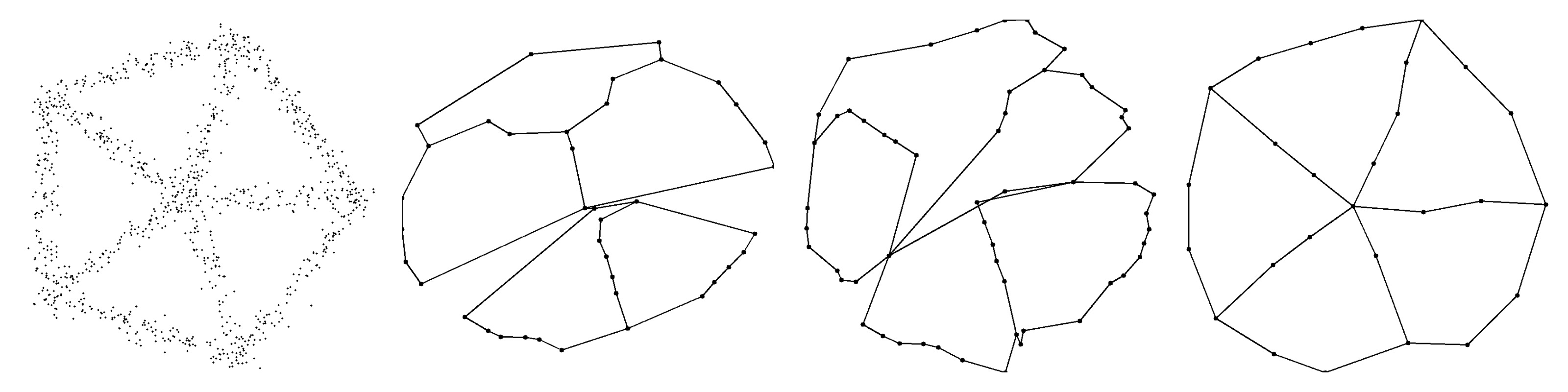}
\caption{\textbf{1st}: a cloud sampled from the $W(5)$ graph with Gaussian noise with $\sigma = 0.04$; 
\textbf{2nd}: Mapper output on $C$; 
\textbf{3rd}: $\al$-Reeb$(C)$; 
\textbf{4th}: $\shopes(C)$ by Algorithm~\ref{alg:contract}.}
\label{fig:wheel5}
\end{figure}

\begin{figure}
\centering
\def\svgwidth{\columnwidth}
\includegraphics[width=\textwidth]{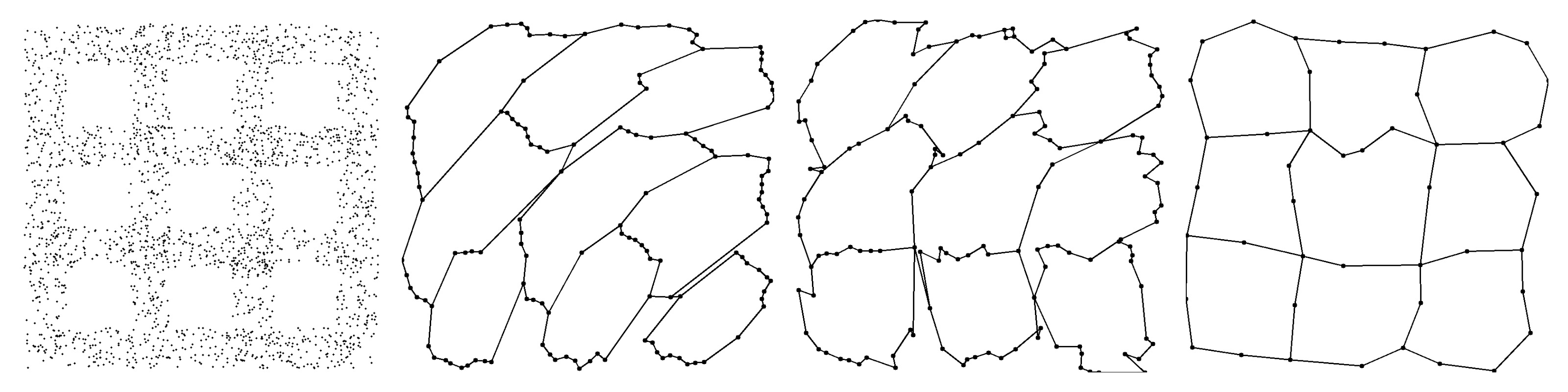}
\caption{\textbf{1st}: a cloud $C$ sampled from the $G(3,3)$ graph with uniform noise with bound $\mu = 0.2$; 
\textbf{2nd}: Mapper output on $C$; 
\textbf{3rd}: $\al$-Reeb$(C)$; 
\textbf{4th}: $\shopes(C)$ by Algorithm~\ref{alg:contract}.}
\label{fig:grid33}
\end{figure}

\begin{figure}
\centering
\def\svgwidth{\columnwidth}
\includegraphics[width=\textwidth]{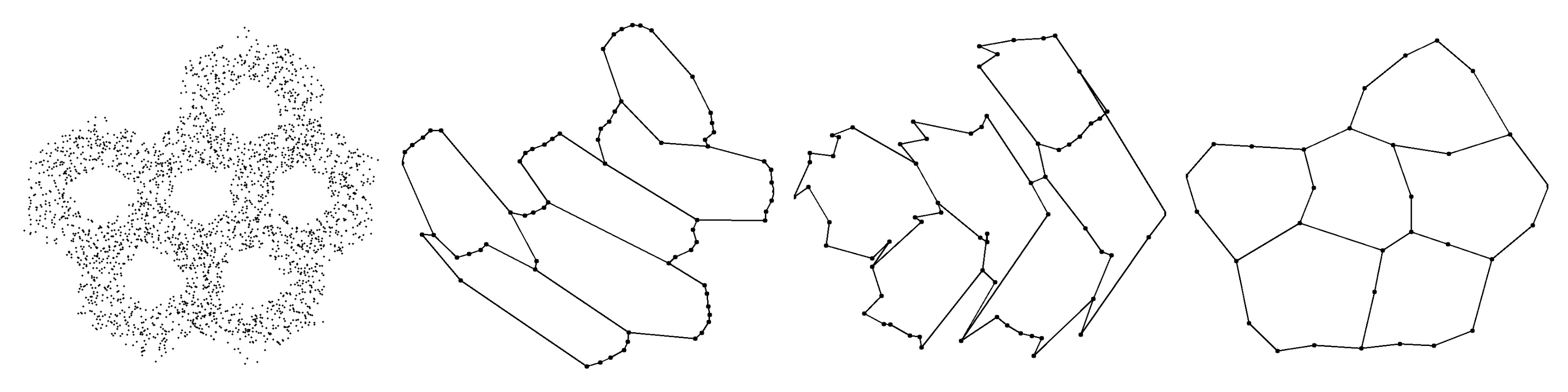}
\caption{\textbf{1st}: a cloud sampled from the $H(6)$ graph with uniform noise with bound $\mu = 0.4$; 
\textbf{2nd}: Mapper output on $C$; 
\textbf{3rd}: $\al$-Reeb$(C)$; 
\textbf{4th}: $\shopes(C)$ by Algorithm~\ref{alg:contract}.}
\label{fig:hex6}
\end{figure}

\section{Comparison of the algorithms on synthetic and real data}
\label{sec:experiments}

\subsection{Running over wide ranges of parameters of Mapper and $\al$-Reeb graphs}
\label{sub:parameters}

Unlike $\hopes$, the Mapper and $\al$-Reeb algorithms require additional parameters that essentially affect outputs. 
To get adequate outputs of Mapper and $\al$-Reeb, one can manually tune these parameters for every input cloud. 
However, for different types of synthetic and real clouds, choosing the same parameters would be inadequate. 
If the parameters such as $\al$ of the $\al$-Reeb algorithm and $\ep$, used for DBSCAN in the Mapper algorithm, are too small, then noisy cycles will also be captured in addition to the dominant cycles. 
If these parameters are too large, then even the dominant cycles will not be captured. 
$\hopes$ is always run for the two heuristics: the 1st widest diagonal gap to the derived skeleton $\hopes_{1,1}(C)$, which is simplified with the threshold $\ep$ equal to the max death by Algorithm~\ref{alg:contract}. 
The parameters of the Mapper and $\al$-Reeb in the experiments below were searched over wide ranges to separately optimise each measure.
Hence only best outputs of Mapper and $\al$-Reeb were counted.     
\smallskip

For Mapper, the amount of overlap of the intervals in the covering of the range of the filter function was fixed at 50\%. 
For a cloud of $n$ points, the number of intervals was chosen as $\frac{tn}{100}$, rounded to the nearest integer, where $t$ took $10$ values between $1.5$ and $3.3$ in $0.2$ increments.
The clustering parameter $\ep$ also took $10$ values between $0.05$ and $0.5$ in $0.05$ increments. 
So in total we have used $100$ configurations of the parameters of Mapper. 
For the $\al$-Reeb algorithm, the scale $\al$ took $10$ values between $0.15$ and $0.6$ in $0.05$ increments.

\subsection{Four measures for the experimental comparison of the algorithms}
\label{sec:measures}

\noindent
For each fixed graph, the measures below are averaged over 200 noisy samples.  
\smallskip

\noindent
\textbf{The Betti success rate} (in percentages) counts the cases when an output contains the expected number of independent cycles, which is called the first Betti number and equal to the rank of the homology group $H_1$.
The $k$-wheel graph $W(k)$ and $k$-hexagons graph $H(k)$ have the first Betti number $k$.
\smallskip

\noindent
\textbf{Homeomorphism success rate} (in percentages) is the empirical  conditional probability that the output skeleton is homeomorphic to an underlying graph $G$ from which a cloud was sampled assuming the first Betti number was correct. 
\smallskip

\noindent
\textbf{The root mean square (RMS) distance} measures how close geometrically an output is to the cloud $C$. 
For each point $p\in C$, the distance $d(p,S)$ is computed from $p$ to the closest straight-line edge of the output $S$. 
Then the \textit{RMS error} is $\sqrt{\sum_{p \in C} d(p,S)^2}$. 
For each type of cloud, we average the RMS distance over all outputs that are successful on the first Betti number. 
\smallskip

\noindent
\textbf{Running time (ms)} is the time averaged over all clouds and parameters.
\smallskip

Each point in Fig.~\ref{fig:wheel_uniform}-\ref{fig:hexagons_Gaussian} has the $y$-coordinate equal to the average value over 200 clouds (with optimal parameters of Mapper and $\al$-Reeb for each cloud) sampled around a graph $G$, the $x$-coordinate is the first Betti number of $G$.

\begin{figure}
\centering
\def\svgwidth{\columnwidth}
\includegraphics[scale = 0.38]{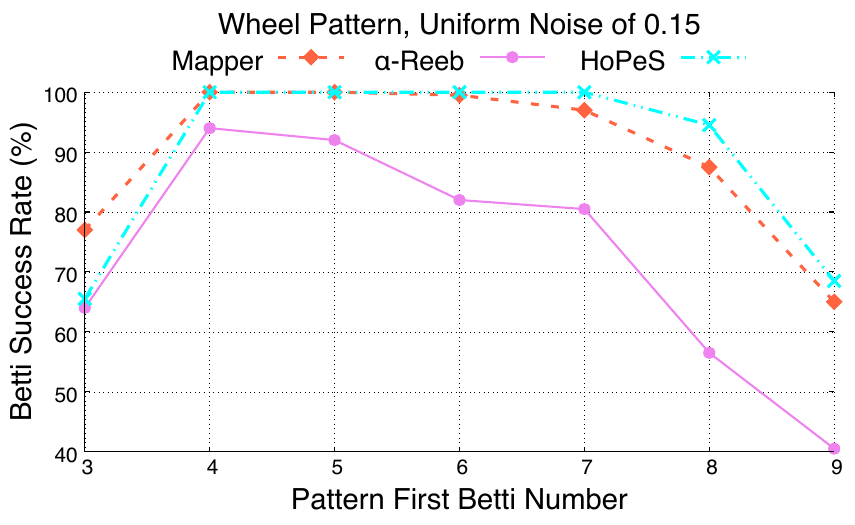} \hspace{1em}
\includegraphics[scale = 0.38]{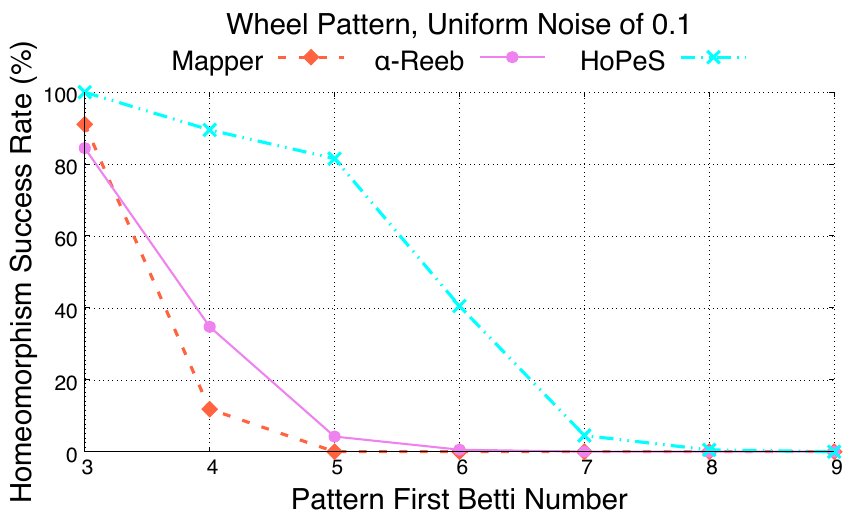}
\includegraphics[scale = 0.38]{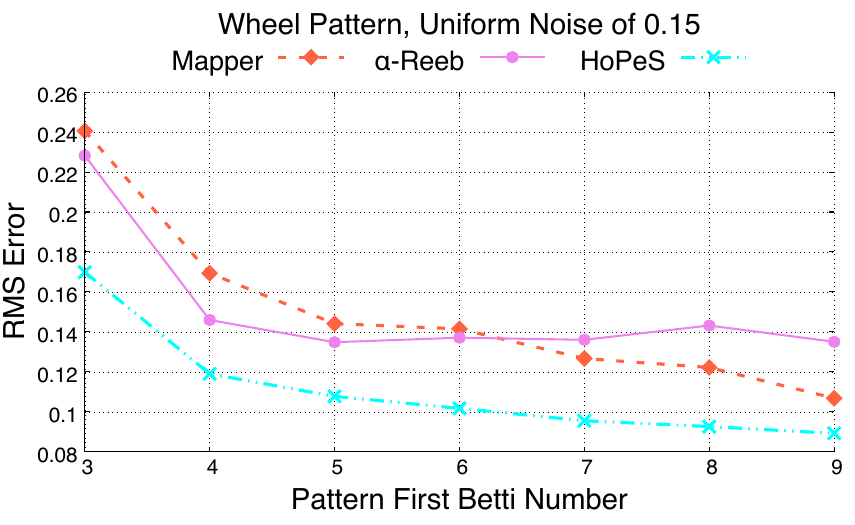} \hspace{1em}
\includegraphics[scale = 0.38]{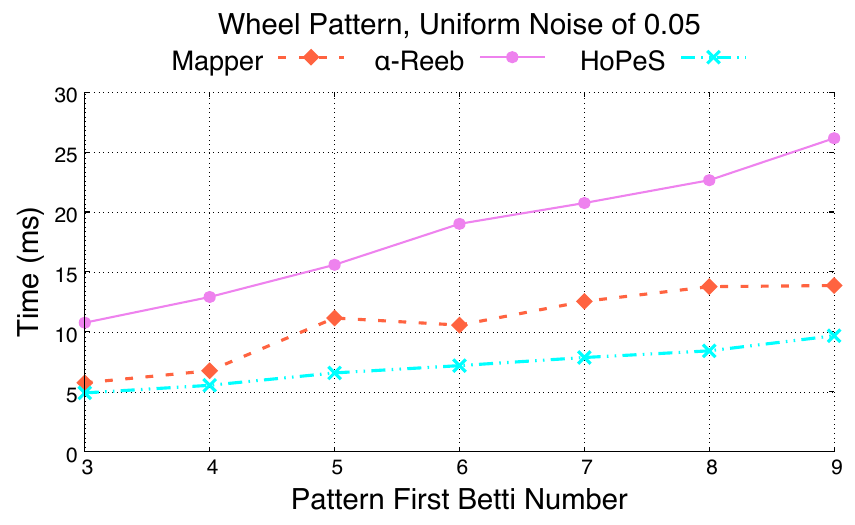}
\caption{All clouds are sampled with uniform noise around wheel graphs $W(k)$, $k=3,\dots,9$.
}
\label{fig:wheel_uniform}
\end{figure}

\begin{figure}
\centering
\def\svgwidth{\columnwidth}
\includegraphics[scale = 0.38]{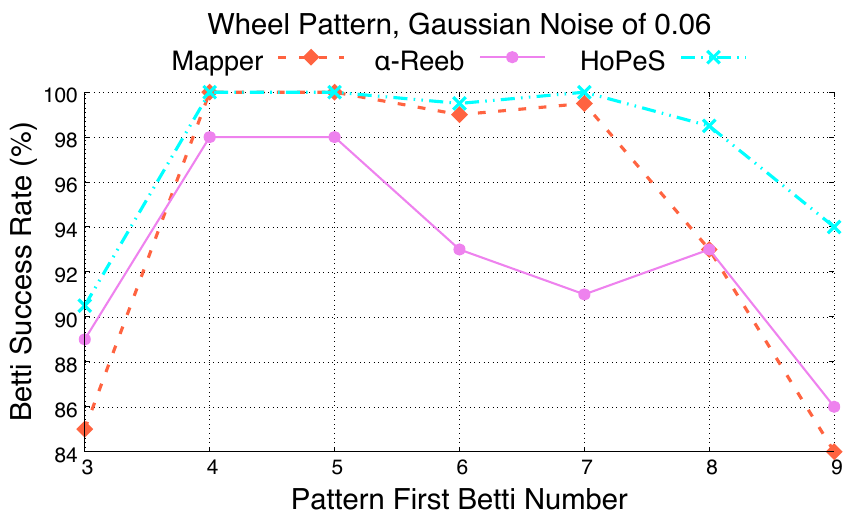} \hspace{1em}
\includegraphics[scale = 0.38]{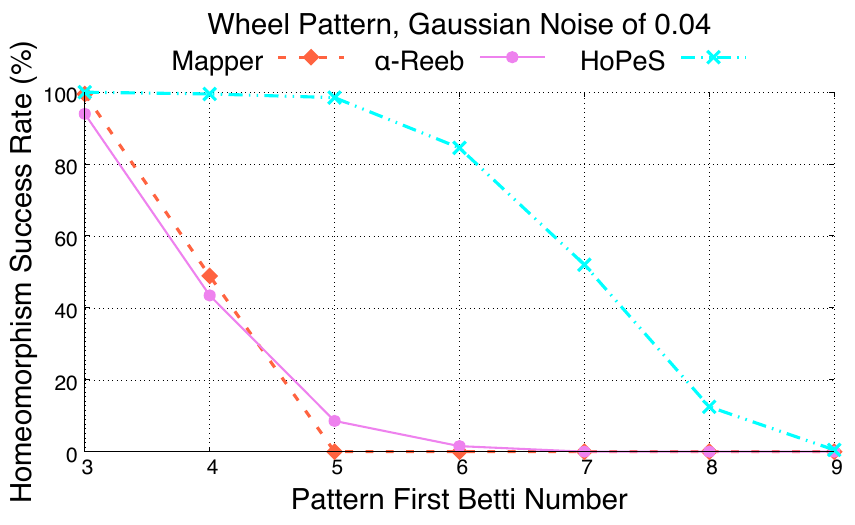}
\includegraphics[scale = 0.38]{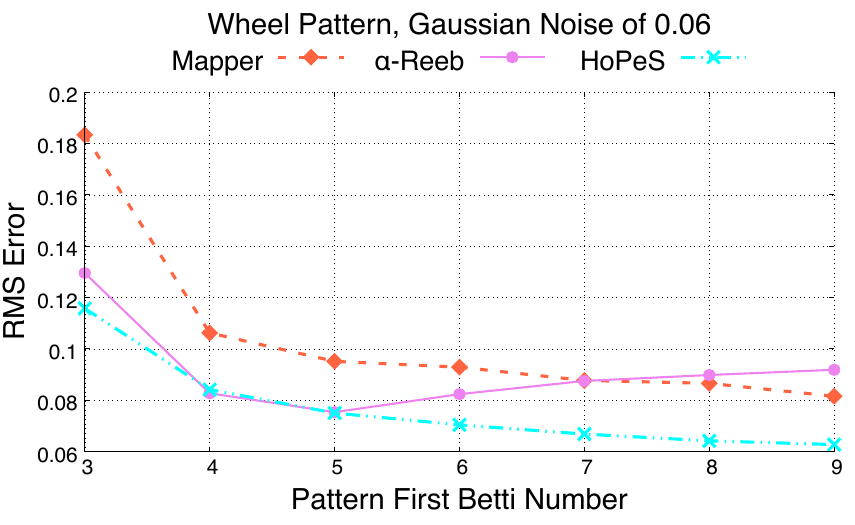} \hspace{1em}
\includegraphics[scale = 0.38]{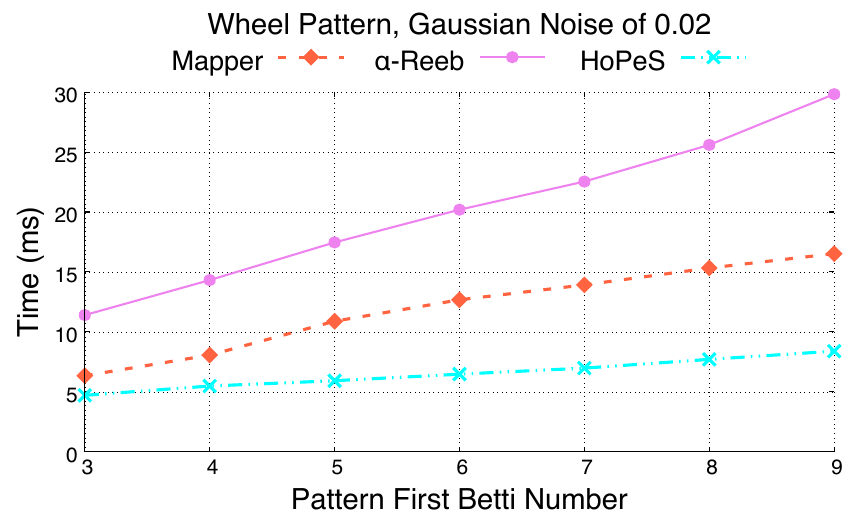}
\caption{All clouds are sampled with Gaussian noise around wheel graphs $W(k)$, $k=3,\dots,9$}
\label{fig:wheelg}
\end{figure}

\begin{figure}
\centering
\def\svgwidth{\columnwidth}
\includegraphics[scale = 0.38]{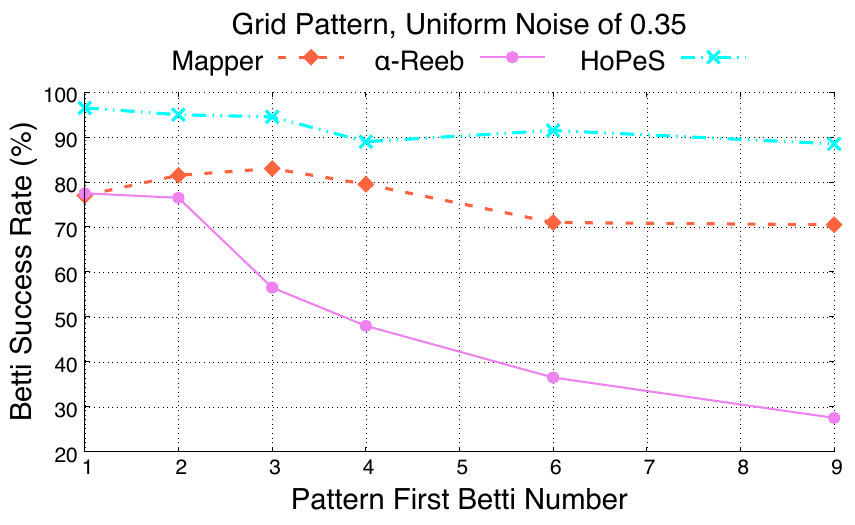} \hspace{1em}
\includegraphics[scale = 0.38]{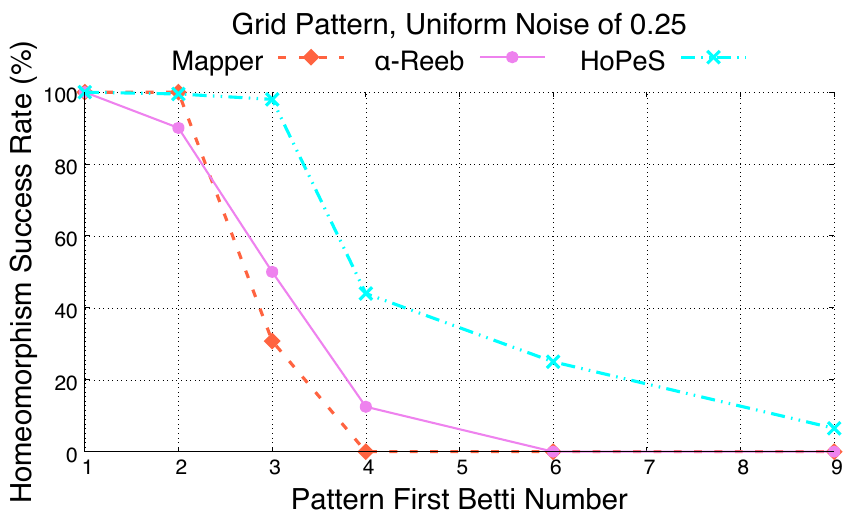}
\includegraphics[scale = 0.38]{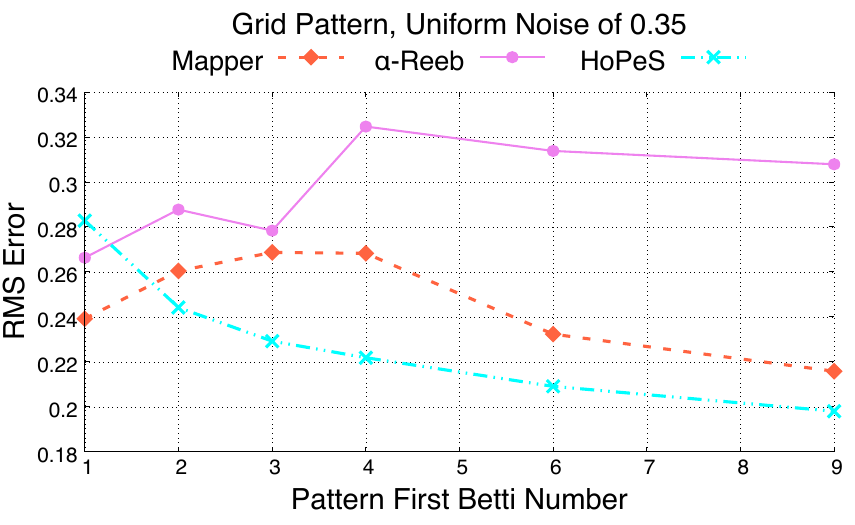} \hspace{1em}
\includegraphics[scale = 0.38]{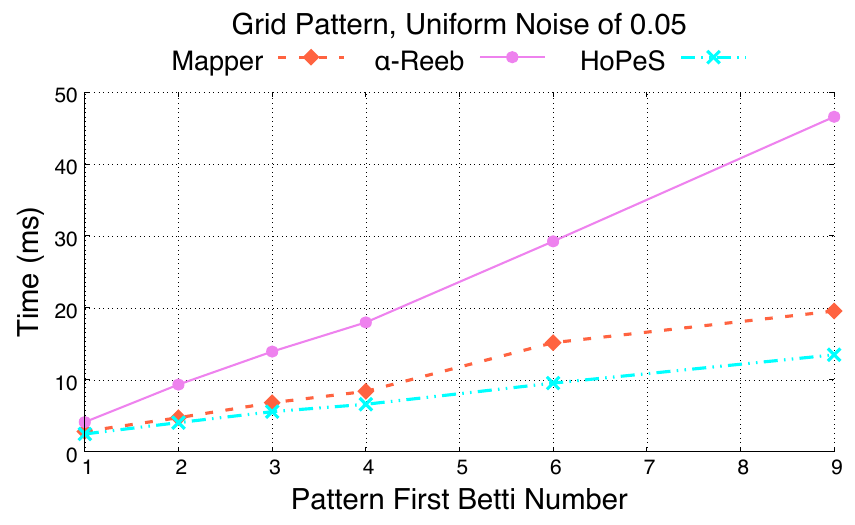}
\caption{All clouds are sampled with uniform noise around grid graphs $G(k,l)$, $1\leq k,l\leq 3$.}
\label{fig:gridu}
\end{figure}

\begin{figure}
\centering
\def\svgwidth{\columnwidth}
\includegraphics[scale = 0.38]{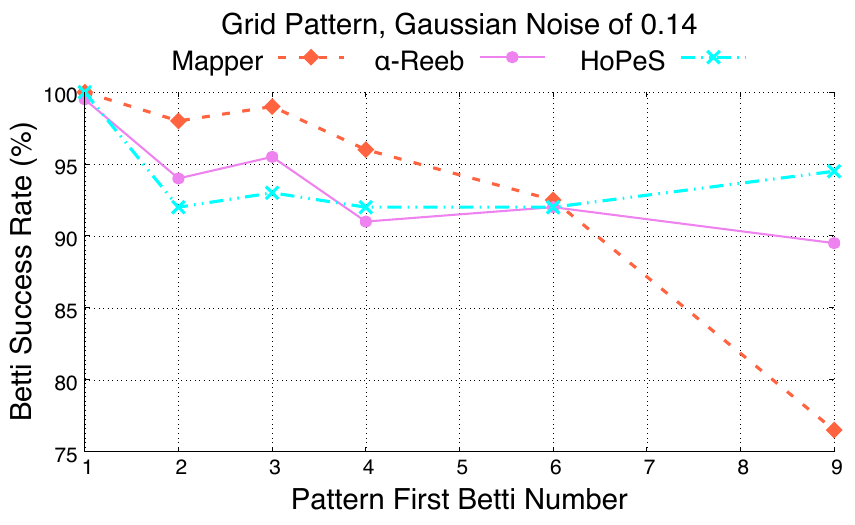} \hspace{1em}
\includegraphics[scale = 0.38]{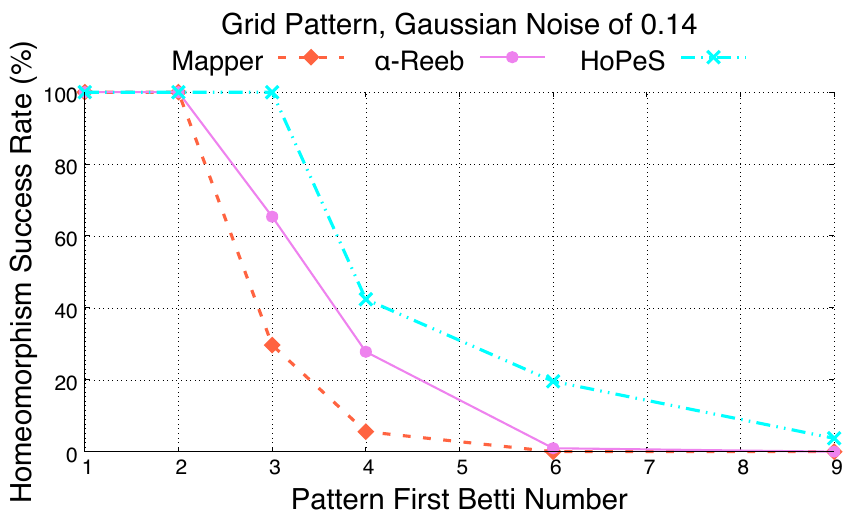}
\includegraphics[scale = 0.38]{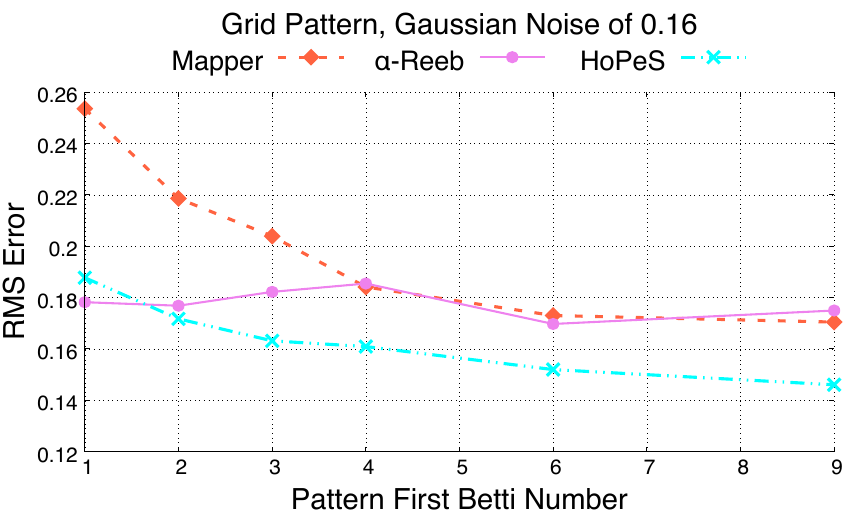} \hspace{1em}
\includegraphics[scale = 0.38]{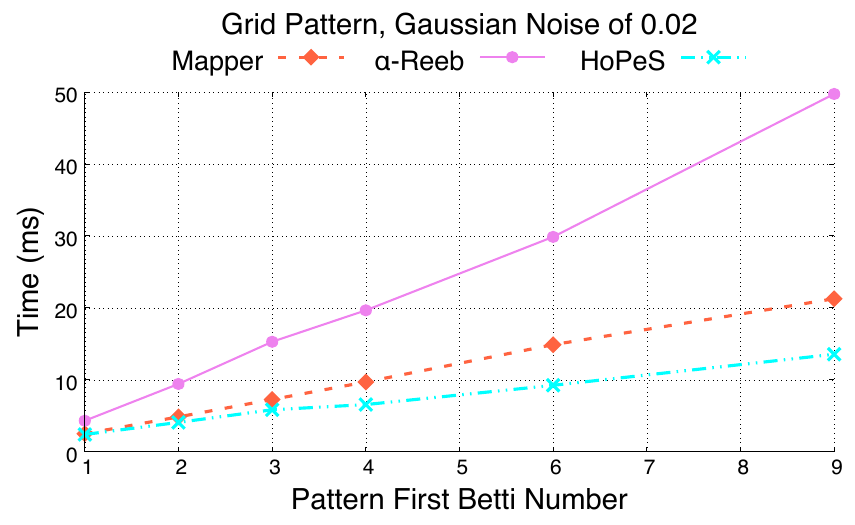}
\caption{All clouds are sampled with Gaussian noise around grid graphs $G(k,l)$, $1\leq k,l\leq 3$.}
\label{fig:gridg}
\end{figure}

\begin{figure}
\centering
\def\svgwidth{\columnwidth}
\includegraphics[scale = 0.38]{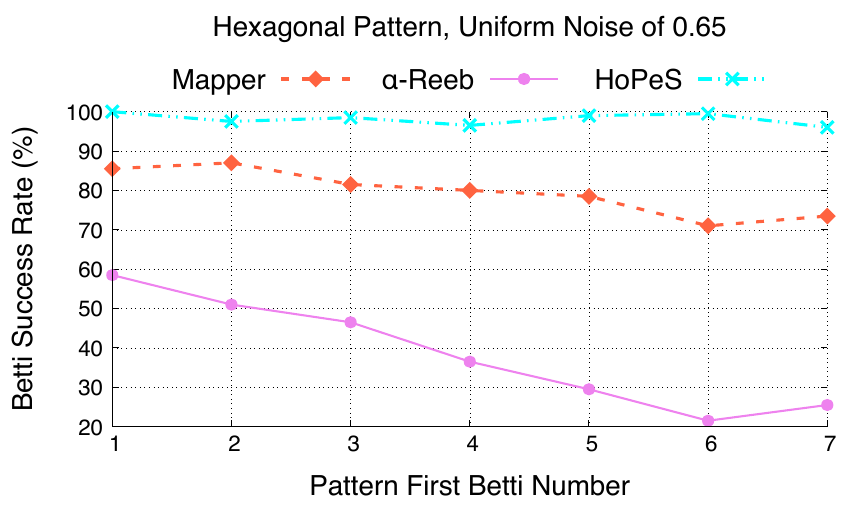} \hspace{1em}
\includegraphics[scale = 0.38]{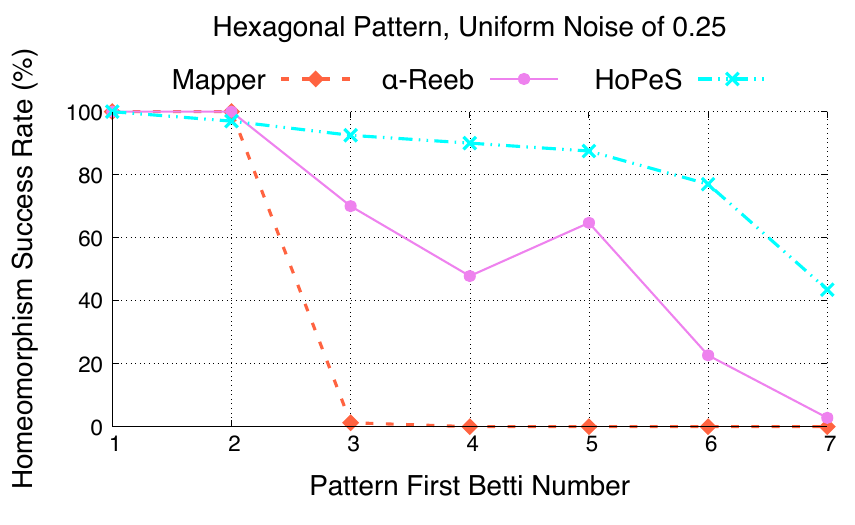}
\includegraphics[scale = 0.38]{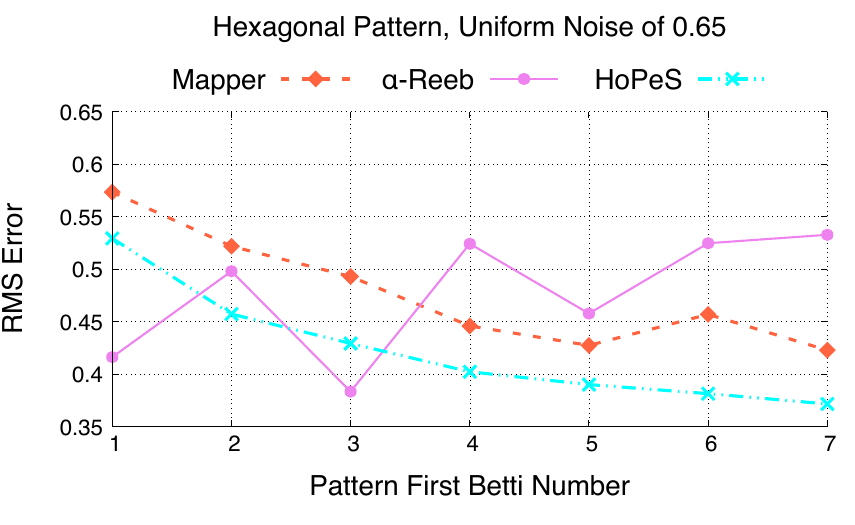} \hspace{1em}
\includegraphics[scale = 0.38]{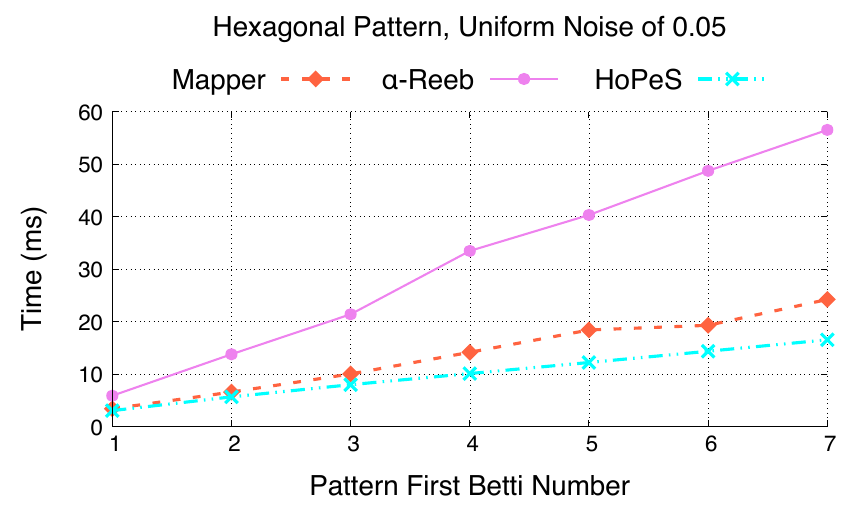}
\caption{Clouds are sampled with uniform noise around hexagon graphs $H(k)$, $k=1,\dots,7$.}
\label{fig:hexagons_uniform}
\end{figure}

\begin{figure}
\centering
\def\svgwidth{\columnwidth}
\includegraphics[scale = 0.38]{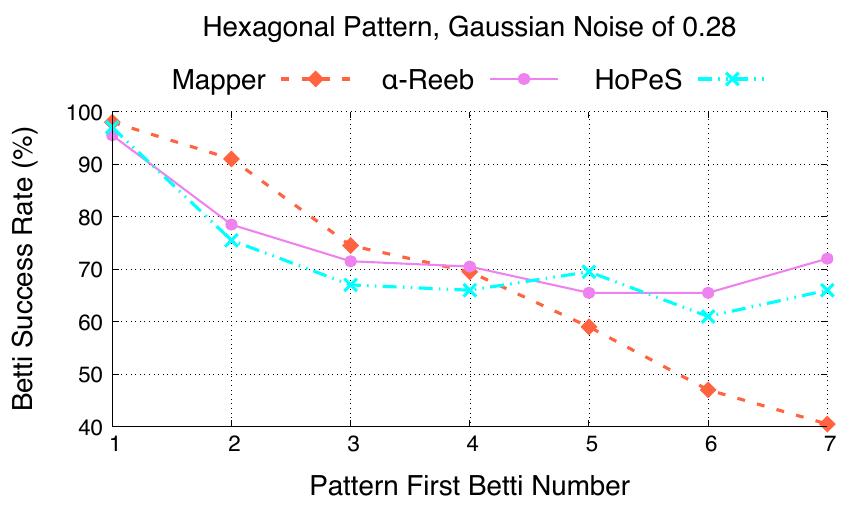} \hspace{1em}
\includegraphics[scale = 0.38]{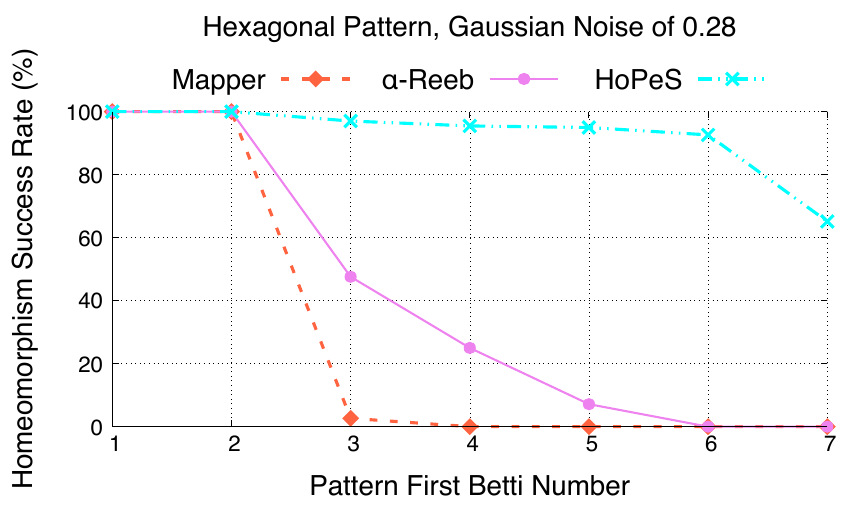}
\includegraphics[scale = 0.38]{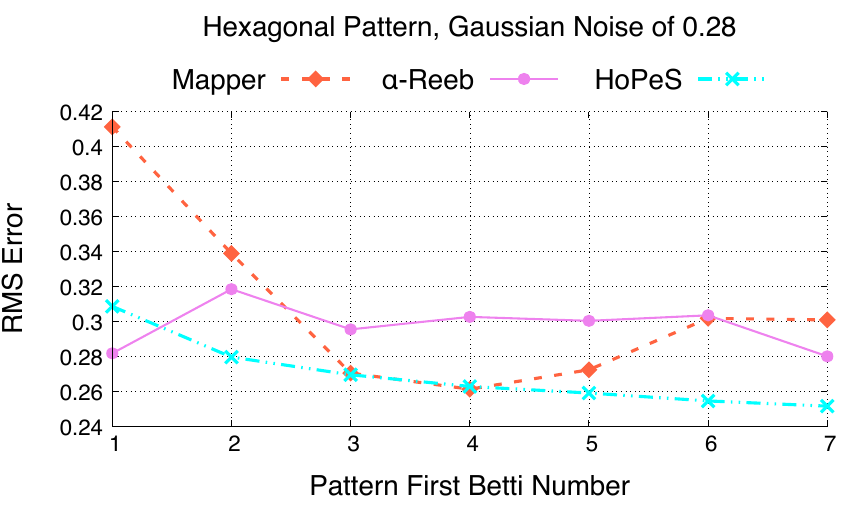} \hspace{1em}
\includegraphics[scale = 0.38]{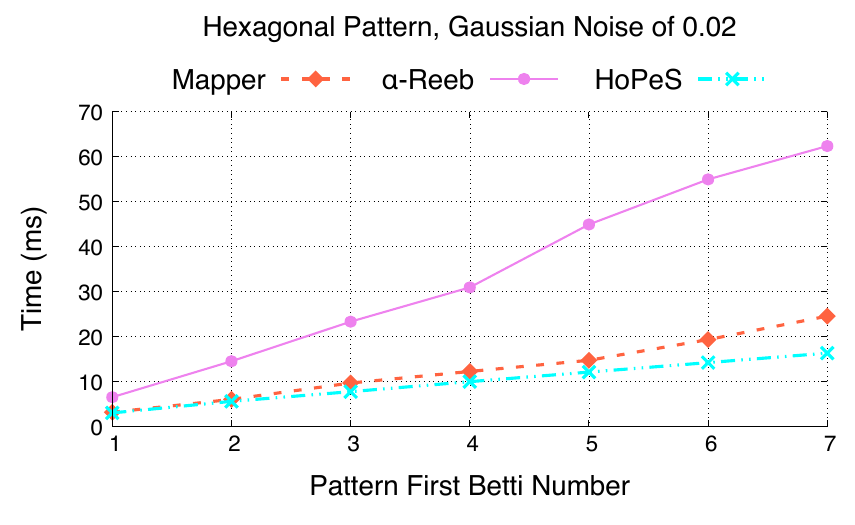}
\caption{Clouds are sampled with Gaussian noise around hexagon graphs $H(k)$, $k=1,\dots,7$.}
\label{fig:hexagons_Gaussian}
\end{figure}

\subsection{Comparison of geometric errors on clouds of edge pixels in BSD images}
\label{sub:edge_clouds}

We ran the algorithms over the clouds $C$ of Canny edge pixels extracted from 500 real images in the Berkeley Segmentation Database (BSD500).
The last pictures in Fig.~\ref{fig:car}, \ref{fig:plane}, \ref{fig:woman} show the derived skeleton $\hopes_{1,1}$ with critical edges in red.
Table~\ref{tab:bsd} shows the RMS distance  from the output graphs to $C$ (in pixels, averaged over images and parameters of Mapper and $\al$-Reeb). 

\begin{center}
\captionof{table}{Averages of 3 algorithms on Canny edge pixels over 500 real images in BSD500}
\label{tab:bsd}
\begin{tabular}{|c|c|c|c|c|}
\hline
Measures / Algorithms & Mapper & $\al$-Reeb & $\shopes$ & $\hopes$ \\
\hline
RMS distance (pixels) & 10.726 & 6.48247 & 5.61771 & 0 \\
Max distance (pixels) & 55.892 & 45.0883 & 29.1306 & 0 \\
Time (milliseconds) & 310 & 4110 & 1256 & 88 \\
\hline
\end{tabular}
\end{center}

\begin{figure}[H]
\centering
\def\svgwidth{\columnwidth}
\includegraphics[scale = 0.45]{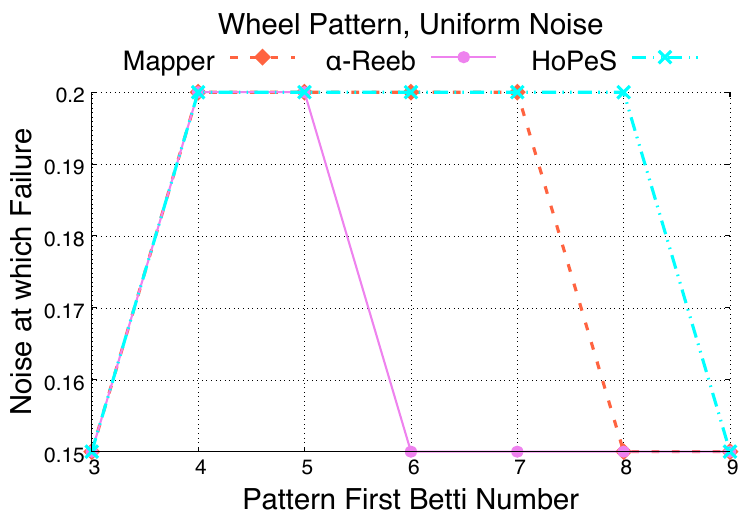} \hspace{1em}
\includegraphics[scale = 0.45]{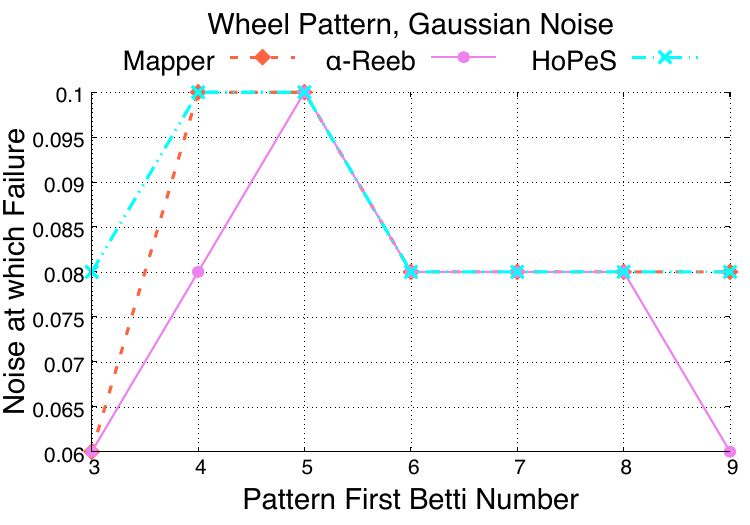}
\includegraphics[scale = 0.45]{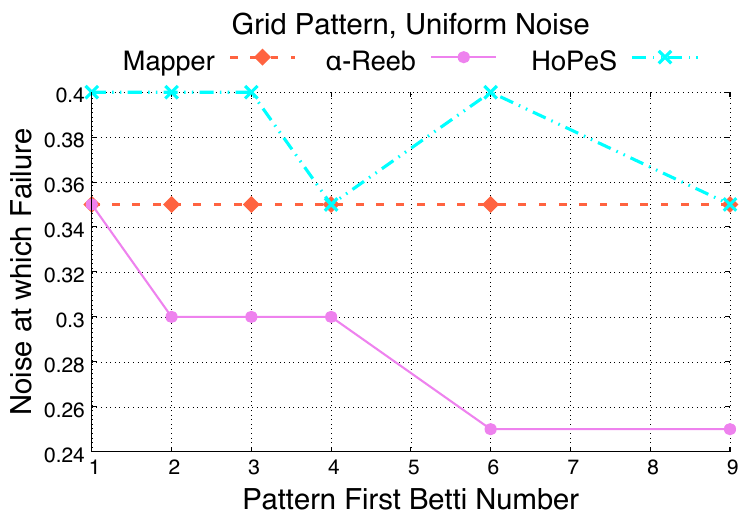} \hspace{1em}
\includegraphics[scale = 0.45]{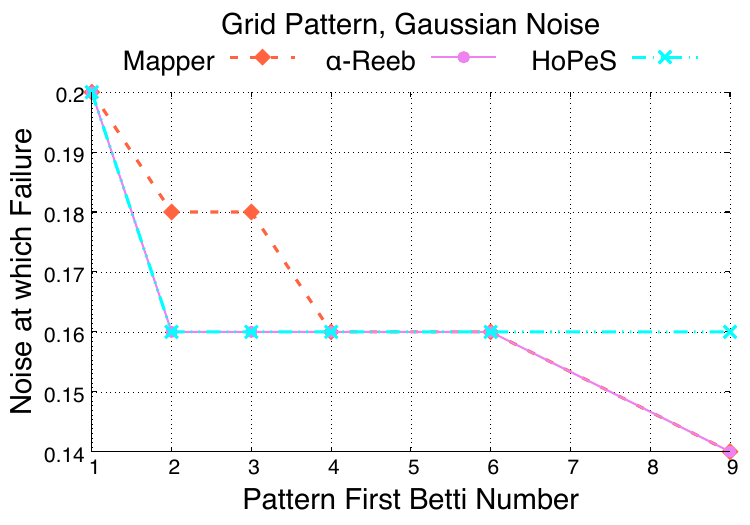}
\includegraphics[scale = 0.45]{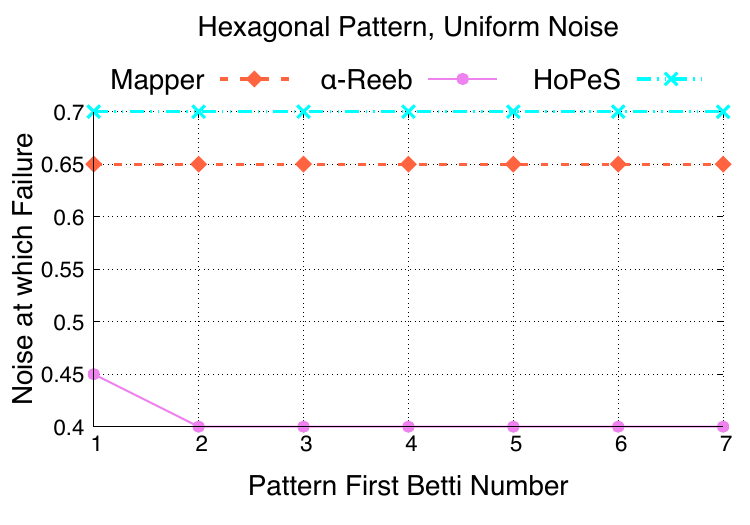} \hspace{1em}
\includegraphics[scale = 0.45]{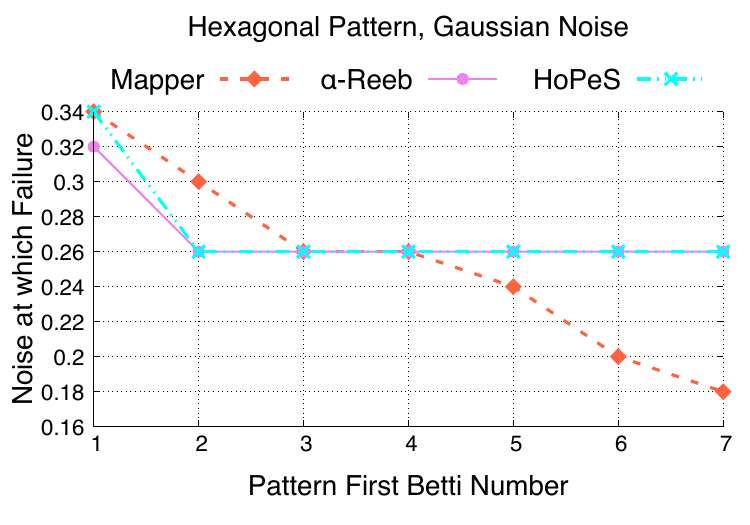}
\caption{The magnitude of noise at which the Betti success rate drops below 90\%. Top-left: wheel pattern with uniform noise; Top-right: wheel pattern with Gaussian noise; Middle-left: grid pattern with uniform noise; Middle-right: grid pattern with Gaussian noise; Bottom-left: hexagons pattern with uniform noise; Bottom-right: hexagons pattern with Gaussian noise.}
\label{fig:noise90}
\end{figure}

\begin{figure}[H]
	\centering
	\def\svgwidth{\columnwidth}
	\includegraphics[height=19mm]{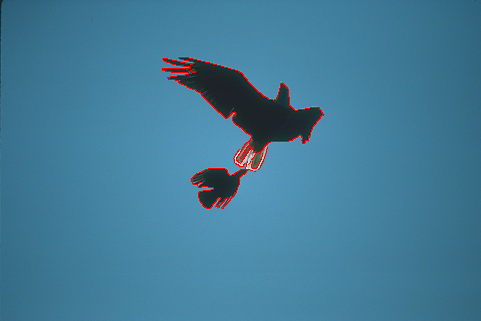}
	\includegraphics[height=19mm]{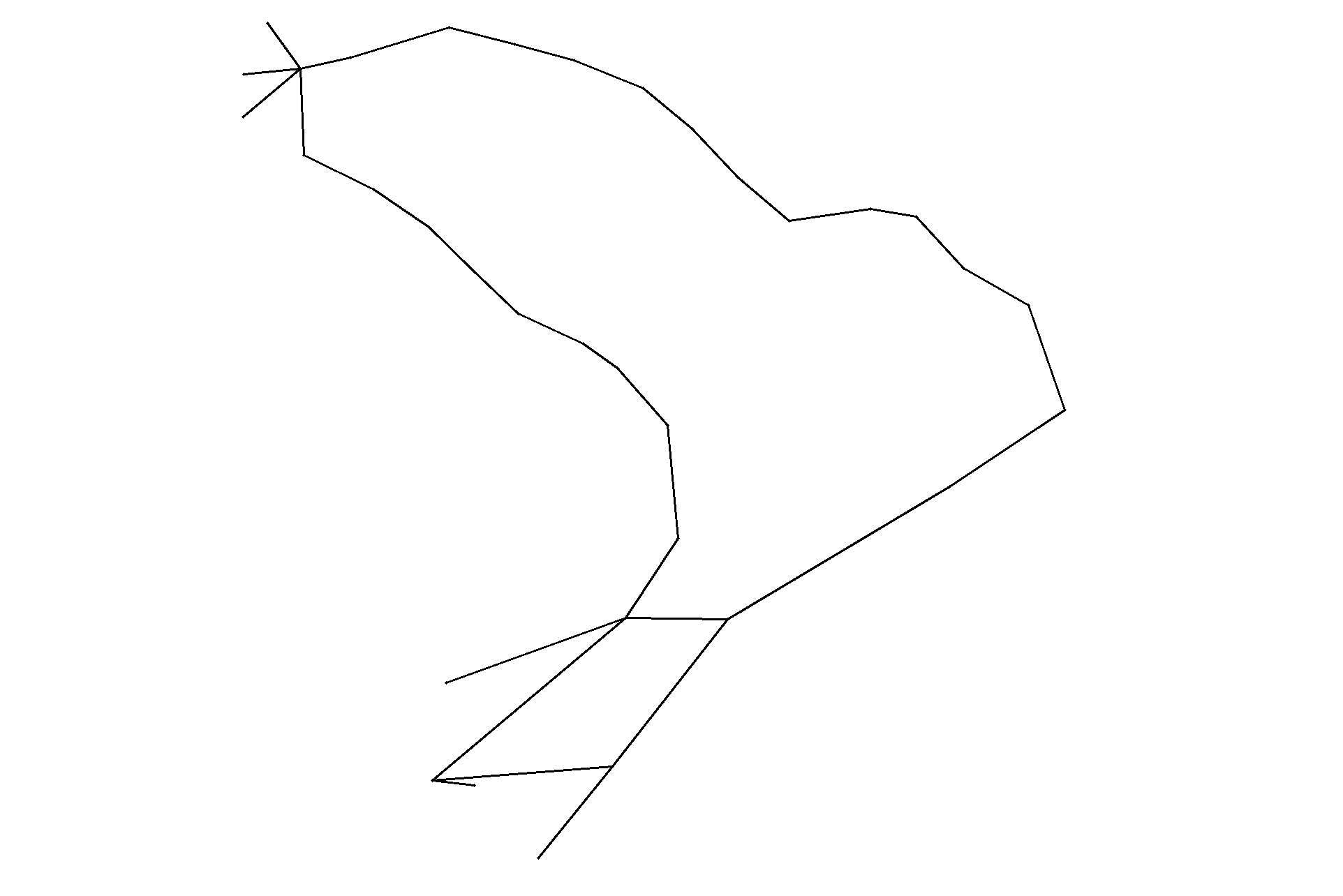}
	\includegraphics[height=19mm]{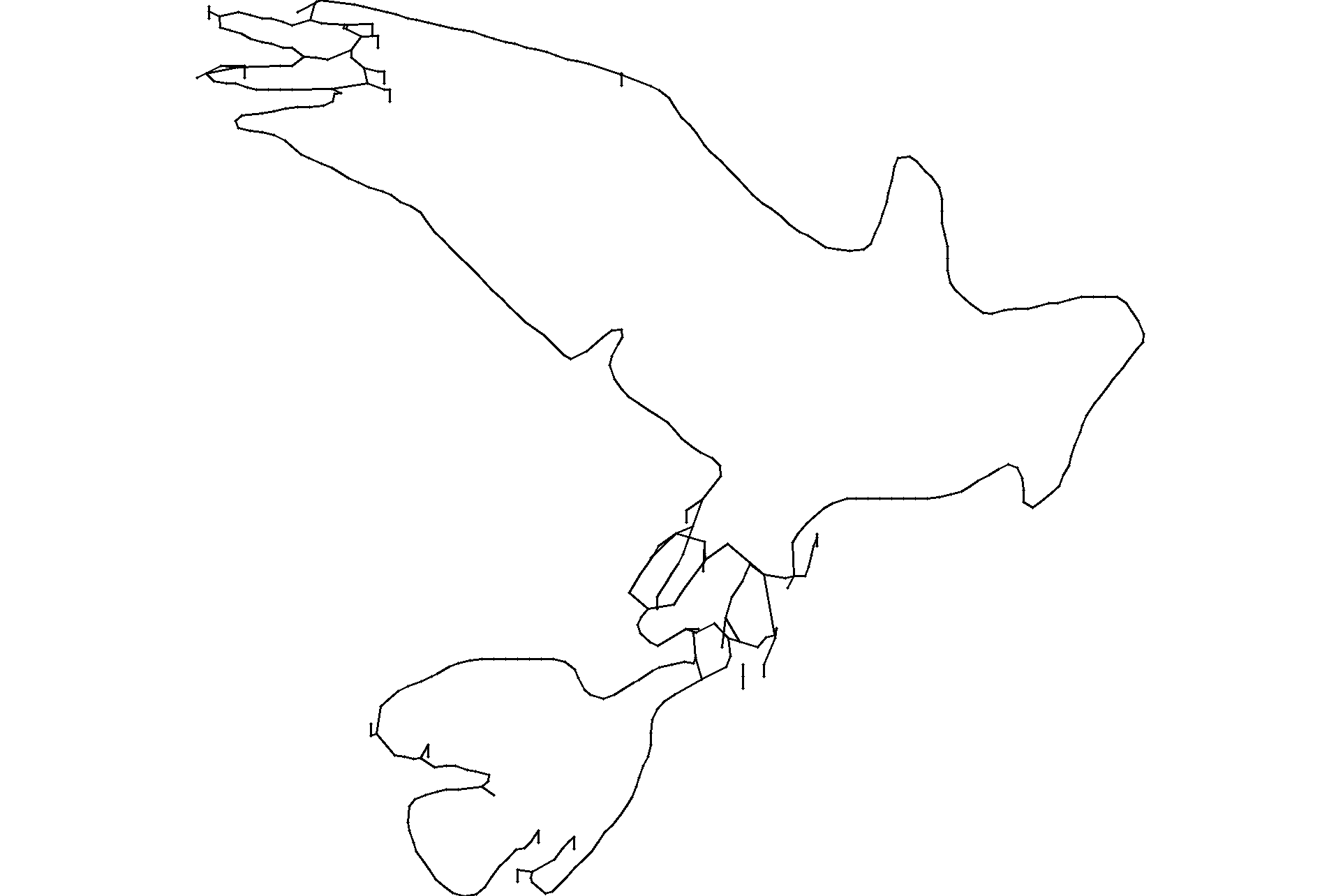}
	\includegraphics[height=19mm]{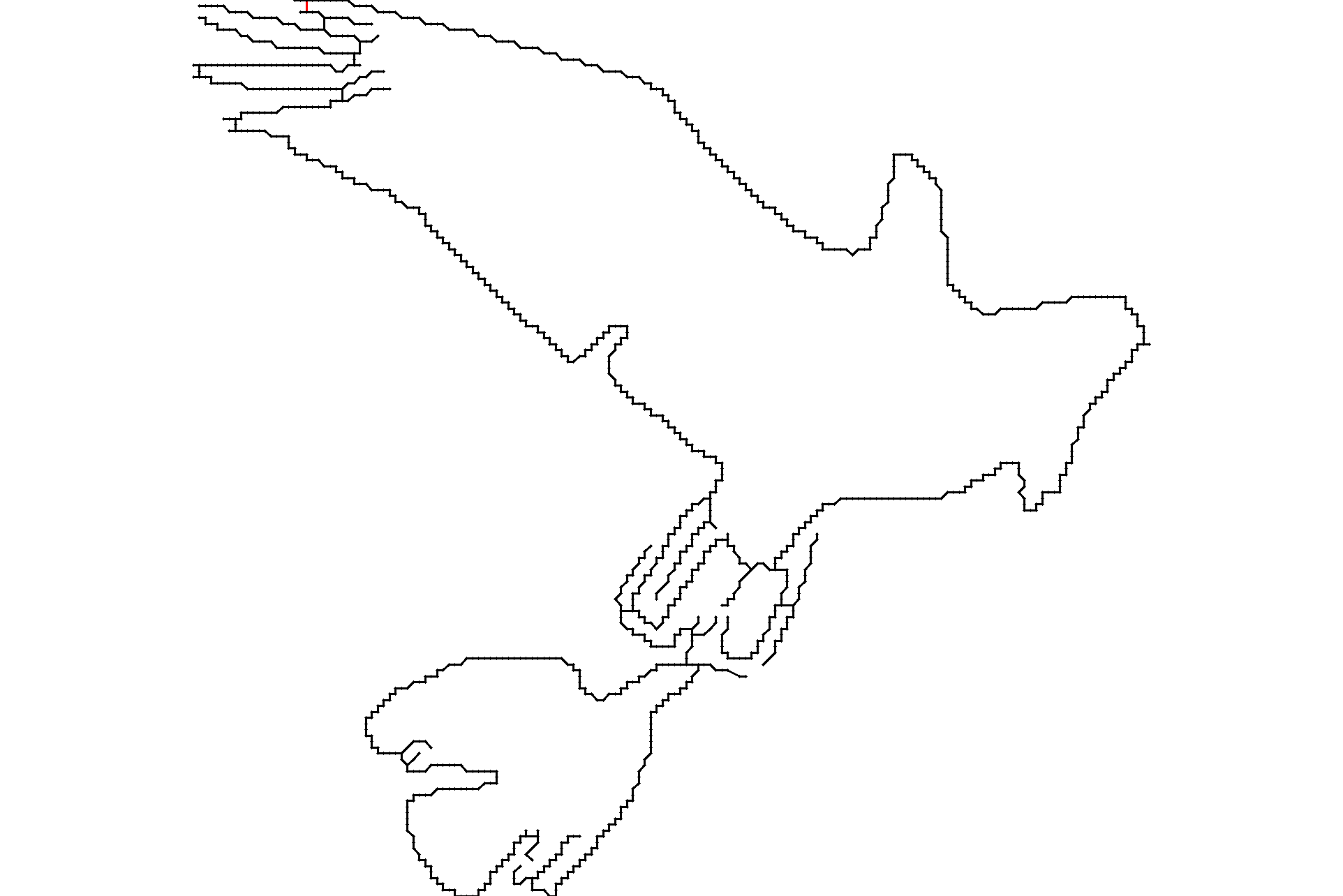}
	\caption{\BSDcaption{135069}}
	\label{fig:bird_2}
\end{figure}

\begin{figure}[H]
	\centering
	\def\svgwidth{\columnwidth}
	\includegraphics[scale = 0.45]{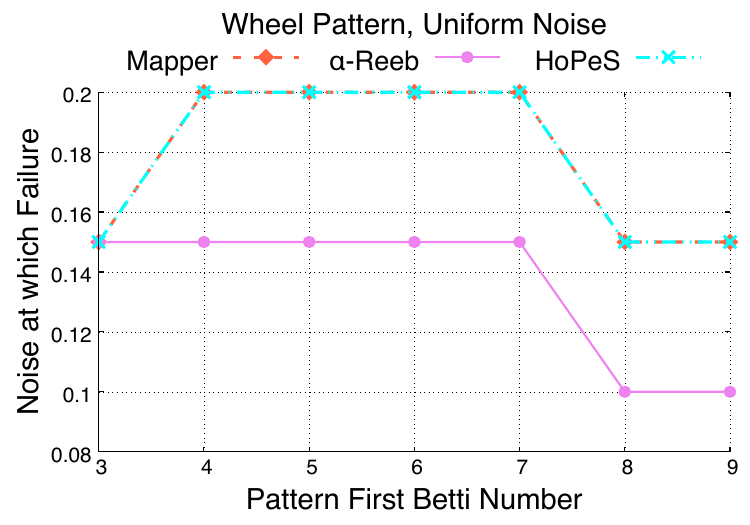} \hspace{1em}
	\includegraphics[scale = 0.45]{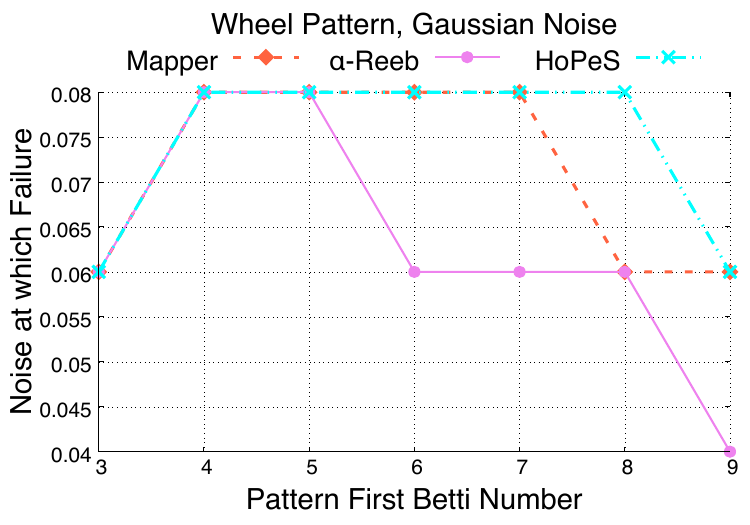}
	\includegraphics[scale = 0.45]{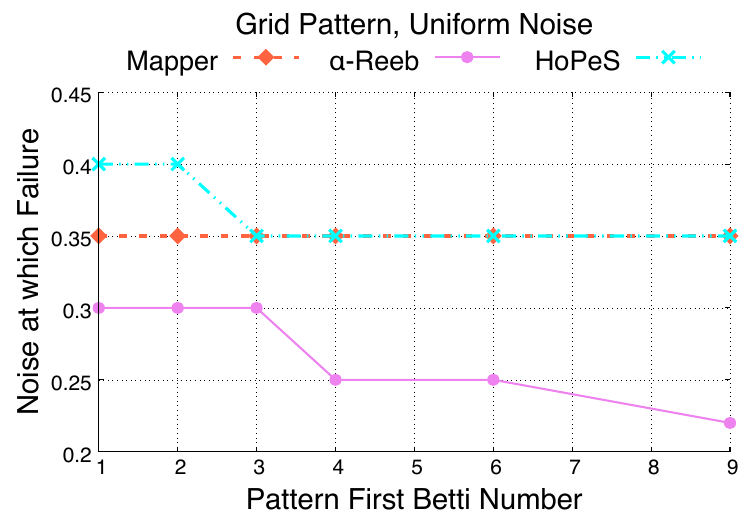} \hspace{1em}
	\includegraphics[scale = 0.45]{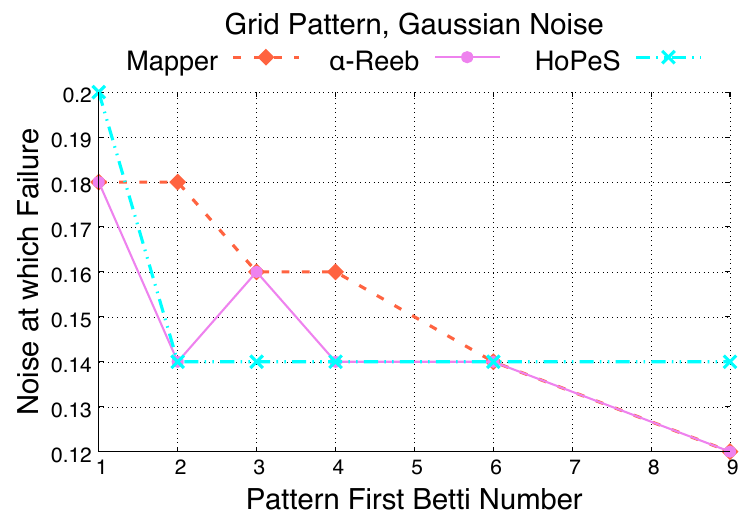}
	\includegraphics[scale = 0.45]{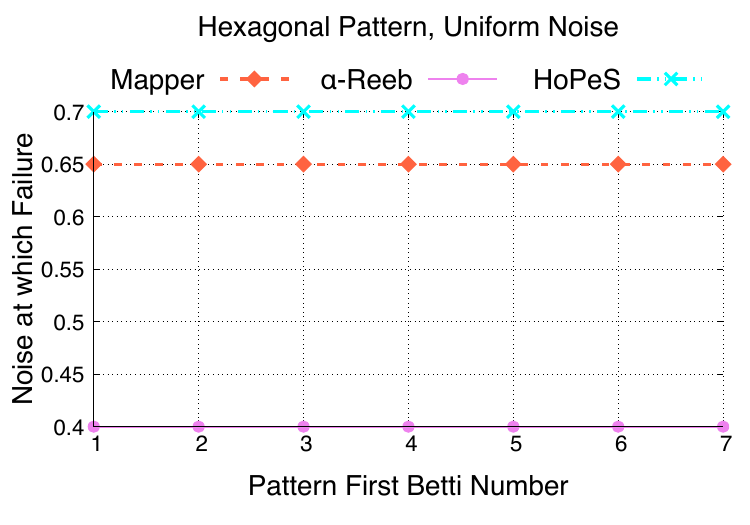} \hspace{1em}
	\includegraphics[scale = 0.45]{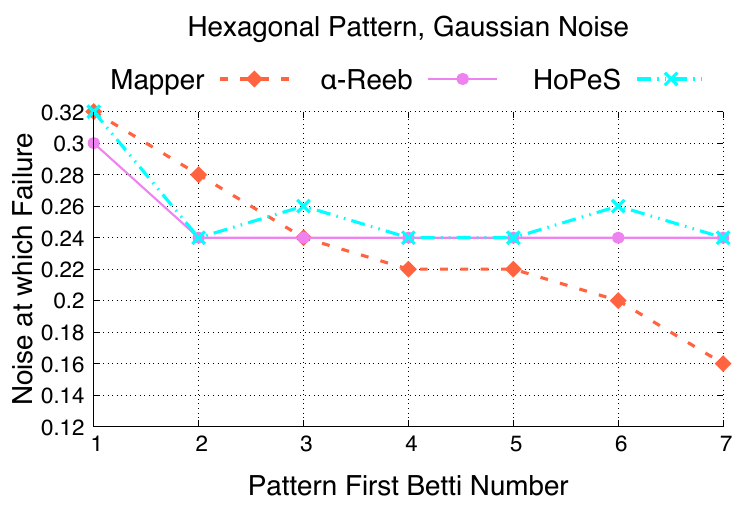}
	\caption{The magnitude of noise at which the Betti success rate drops below 95\%. Top-left: wheel pattern with uniform noise; Top-right: wheel pattern with Gaussian noise; Middle-left: grid pattern with uniform noise; Middle-right: grid pattern with Gaussian noise; Bottom-left: hexagons pattern with uniform noise; Bottom-right: hexagons pattern with Gaussian noise.}
	\label{fig:noise95}
\end{figure}

\begin{figure}[H]
	\centering
	\def\svgwidth{\columnwidth}
	\includegraphics[height=19mm]{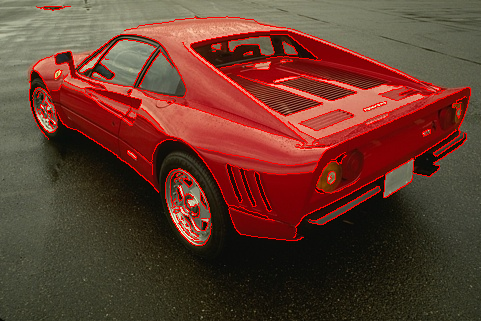}
	\includegraphics[height=19mm]{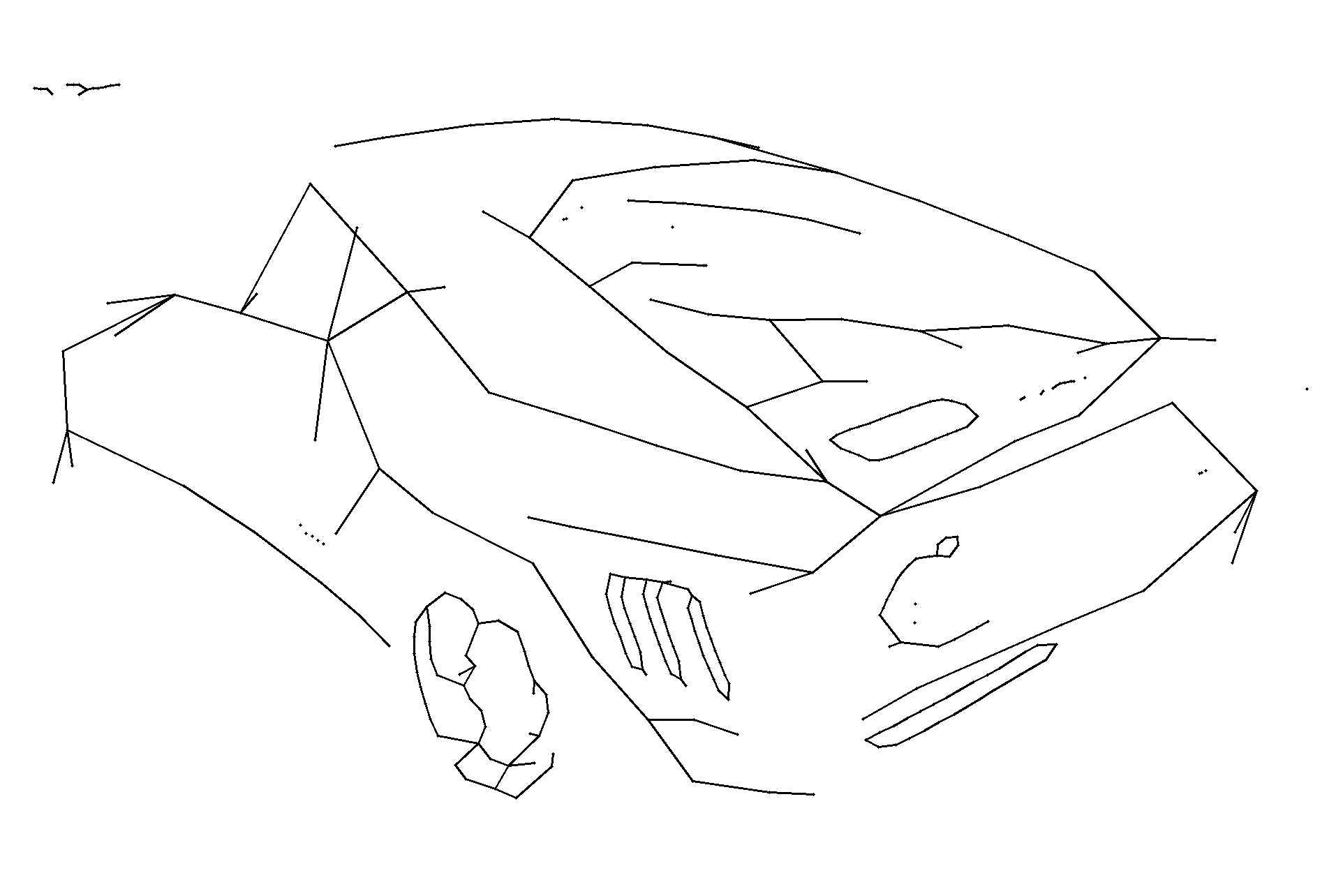}
	\includegraphics[height=19mm]{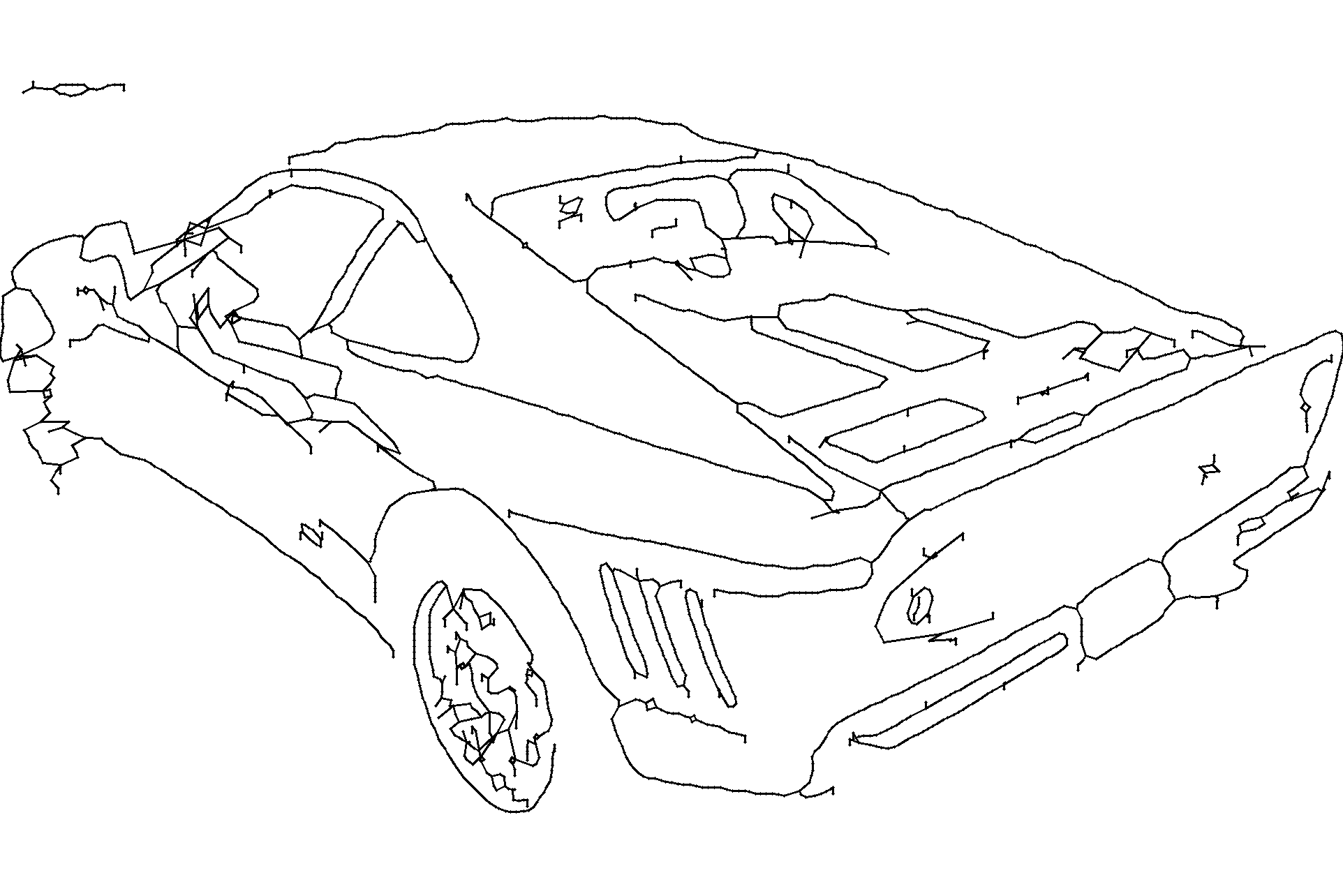}
	\includegraphics[height=19mm]{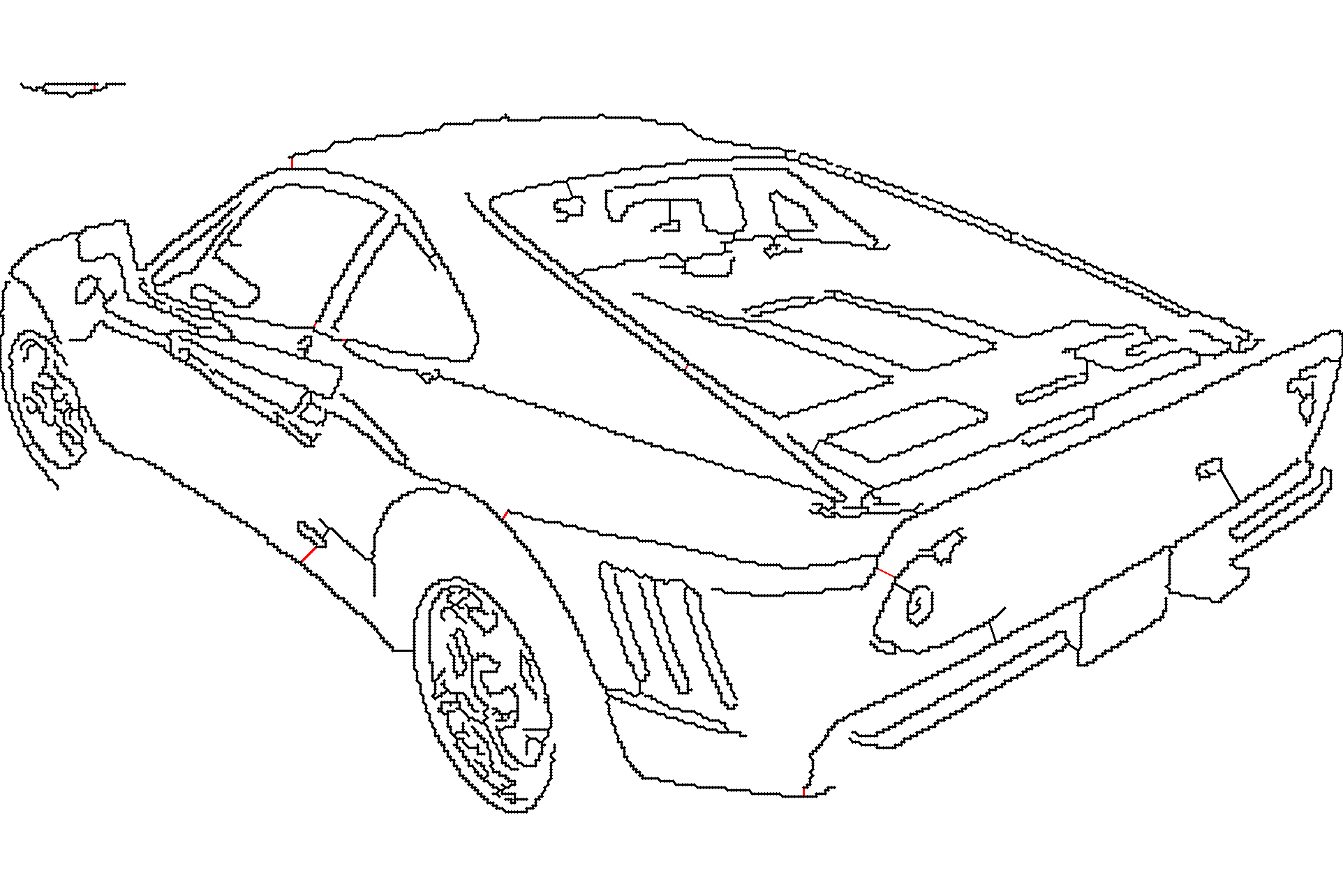}
	\caption{\BSDcaption{29030}}
	\label{fig:car}
\end{figure}

\begin{figure}[H]
\centering
\def\svgwidth{\columnwidth}
\includegraphics[height=19mm]{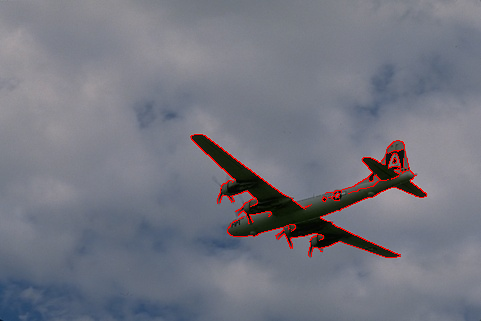}
\includegraphics[height=19mm]{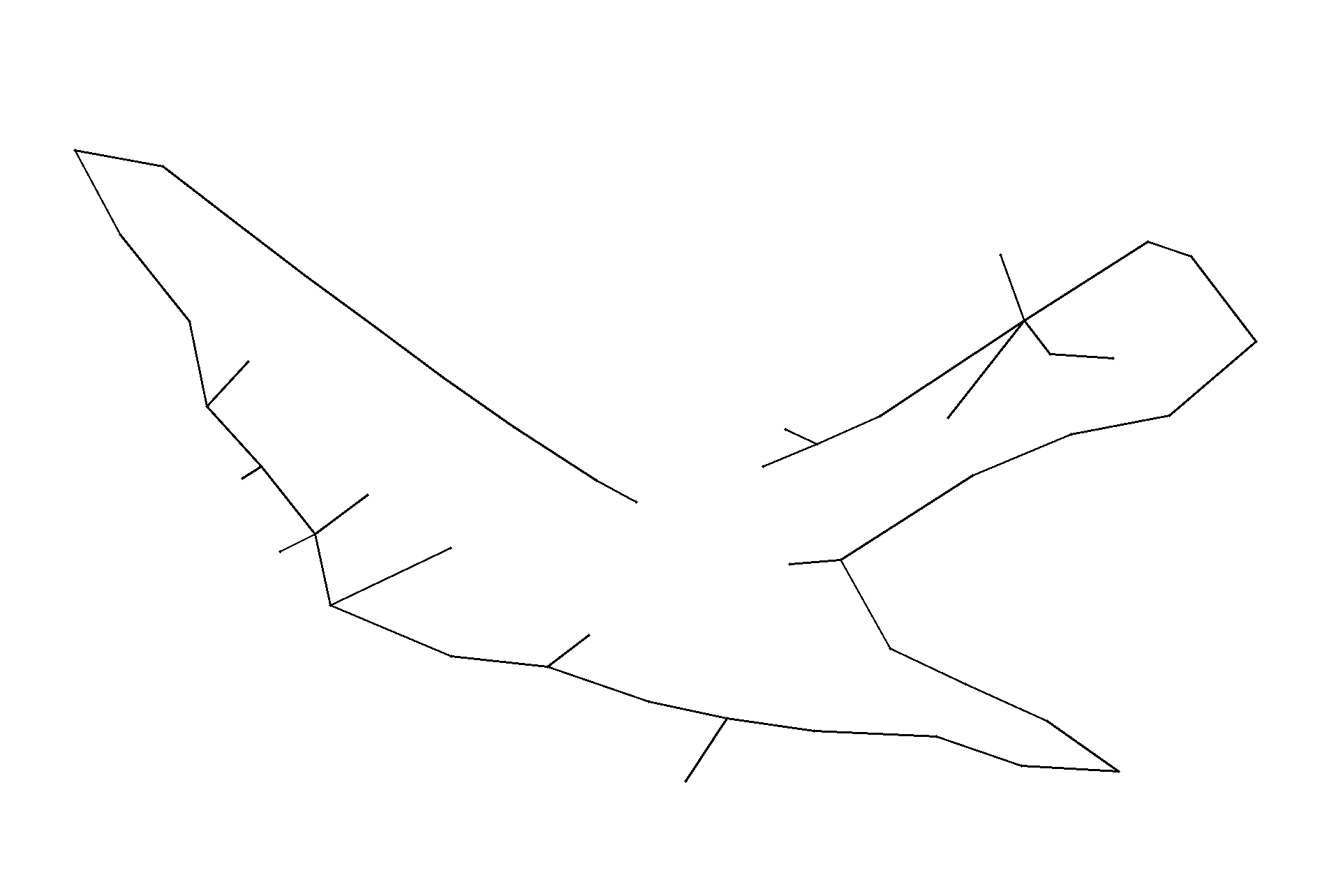}
\includegraphics[height=19mm]{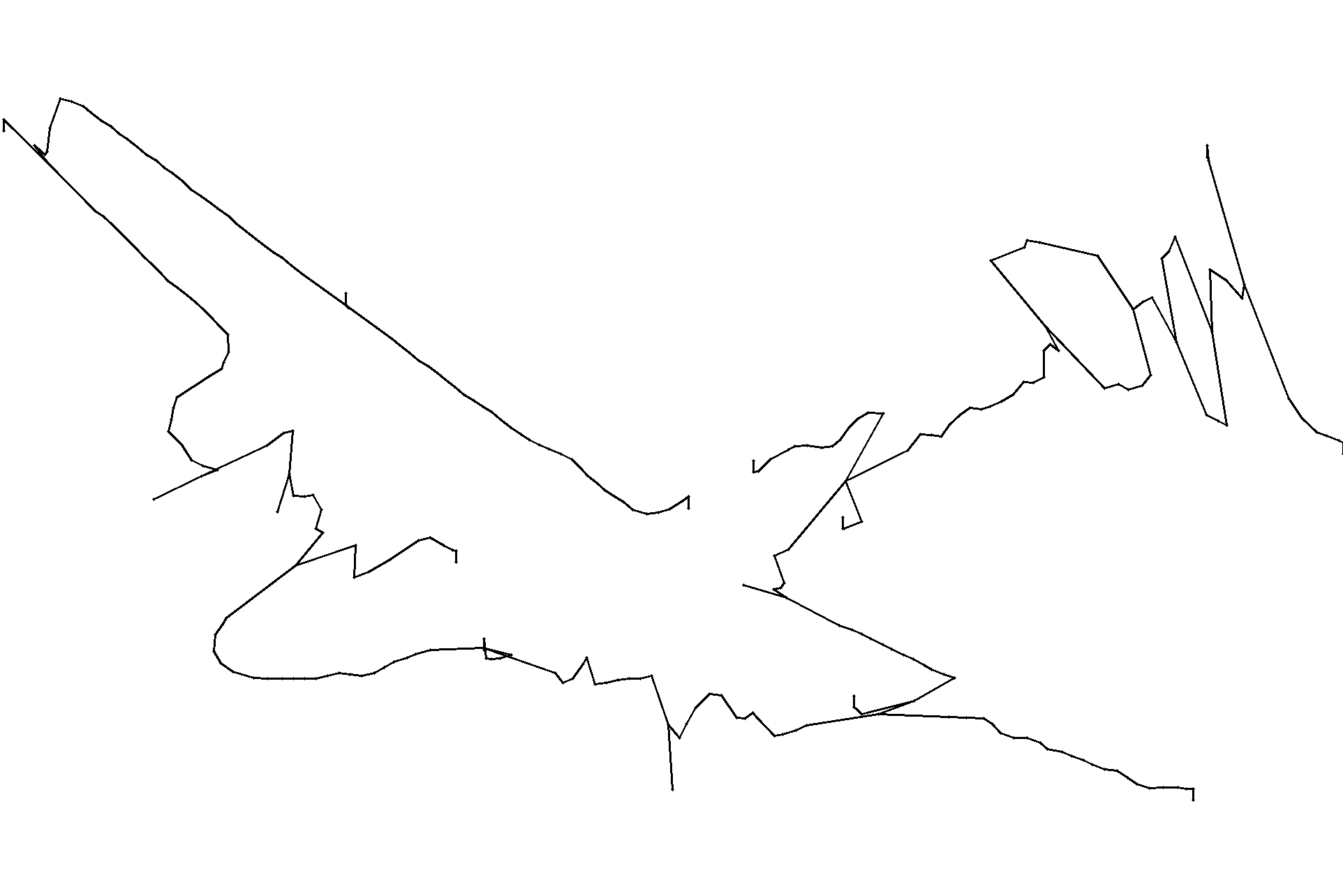}
\includegraphics[height=19mm]{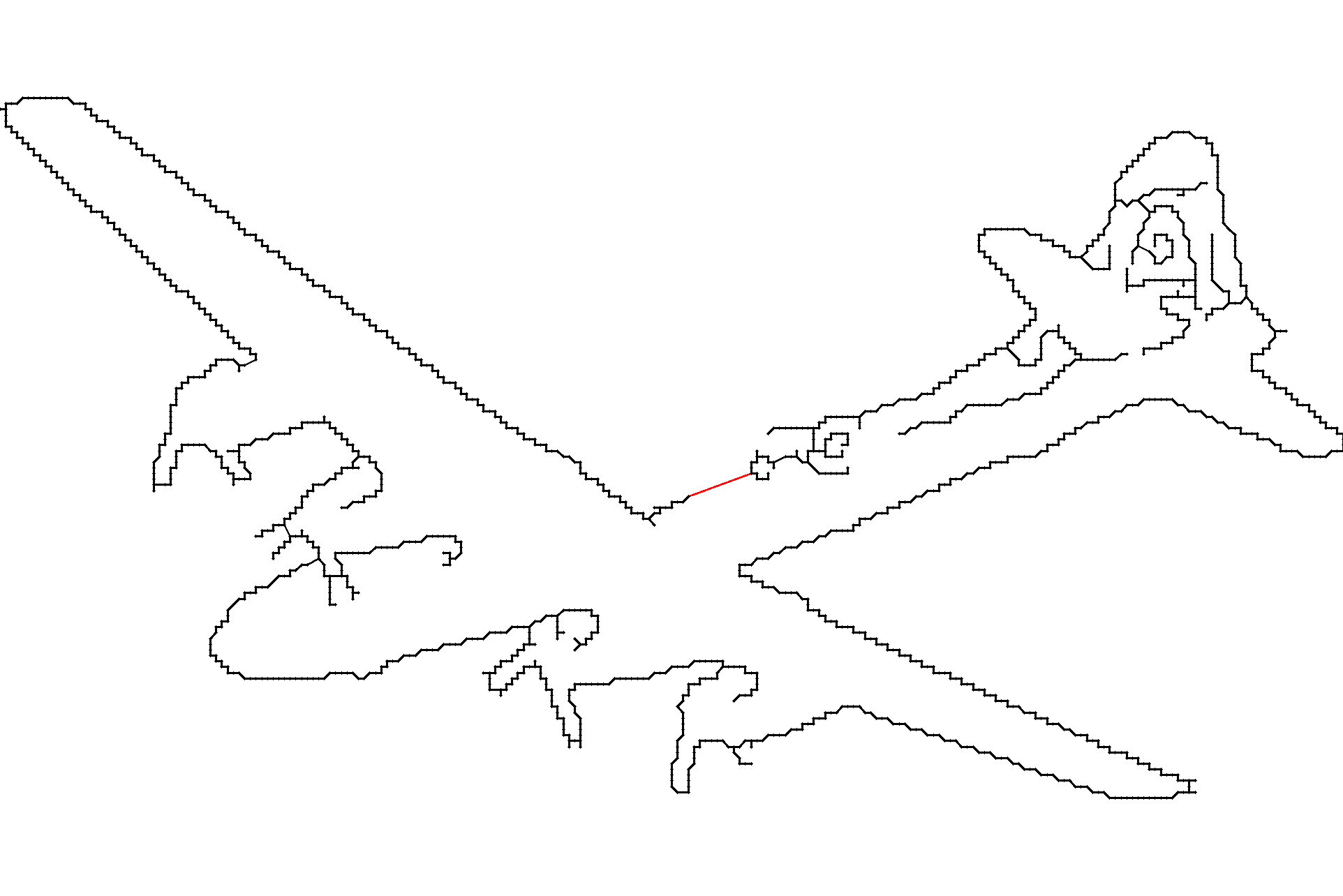}
\caption{Image 3096 with edge pixels in red, the outputs of Mapper, $\al$-Reeb, $\hopes_{1,1}$.}
\label{fig:plane}
\end{figure}
\vspace*{-4mm}

\begin{figure}[H]
\centering
\def\svgwidth{\columnwidth}
\includegraphics[scale = 0.26]{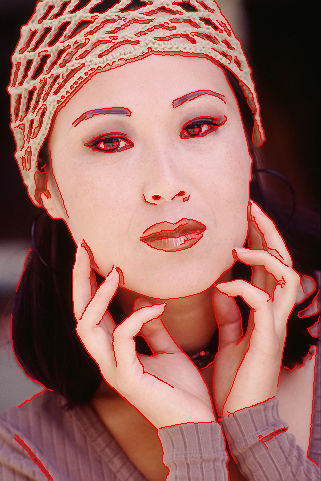}
\includegraphics[scale = 0.0652]{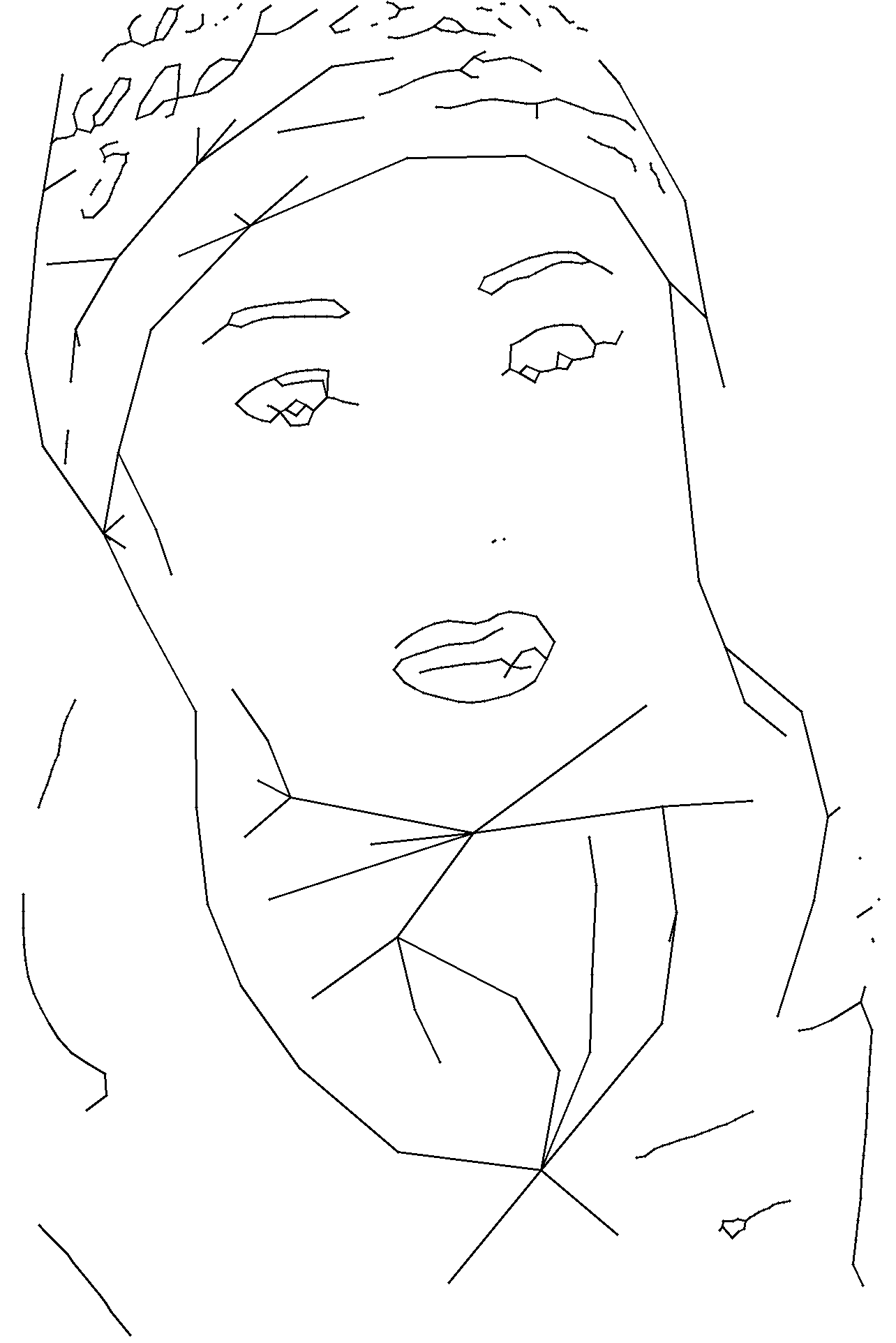}
\includegraphics[scale = 0.0652]{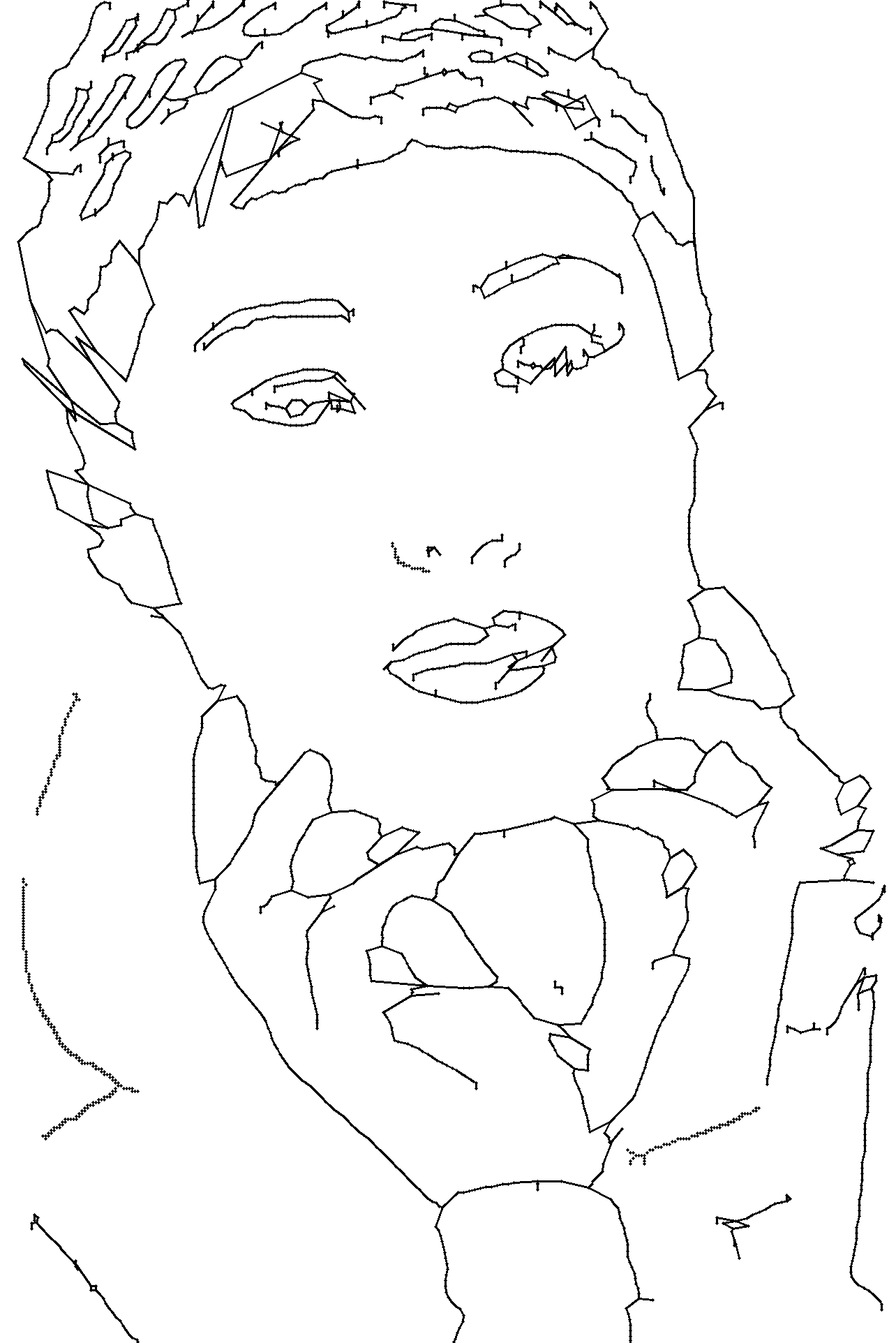}
\includegraphics[scale = 0.0652]{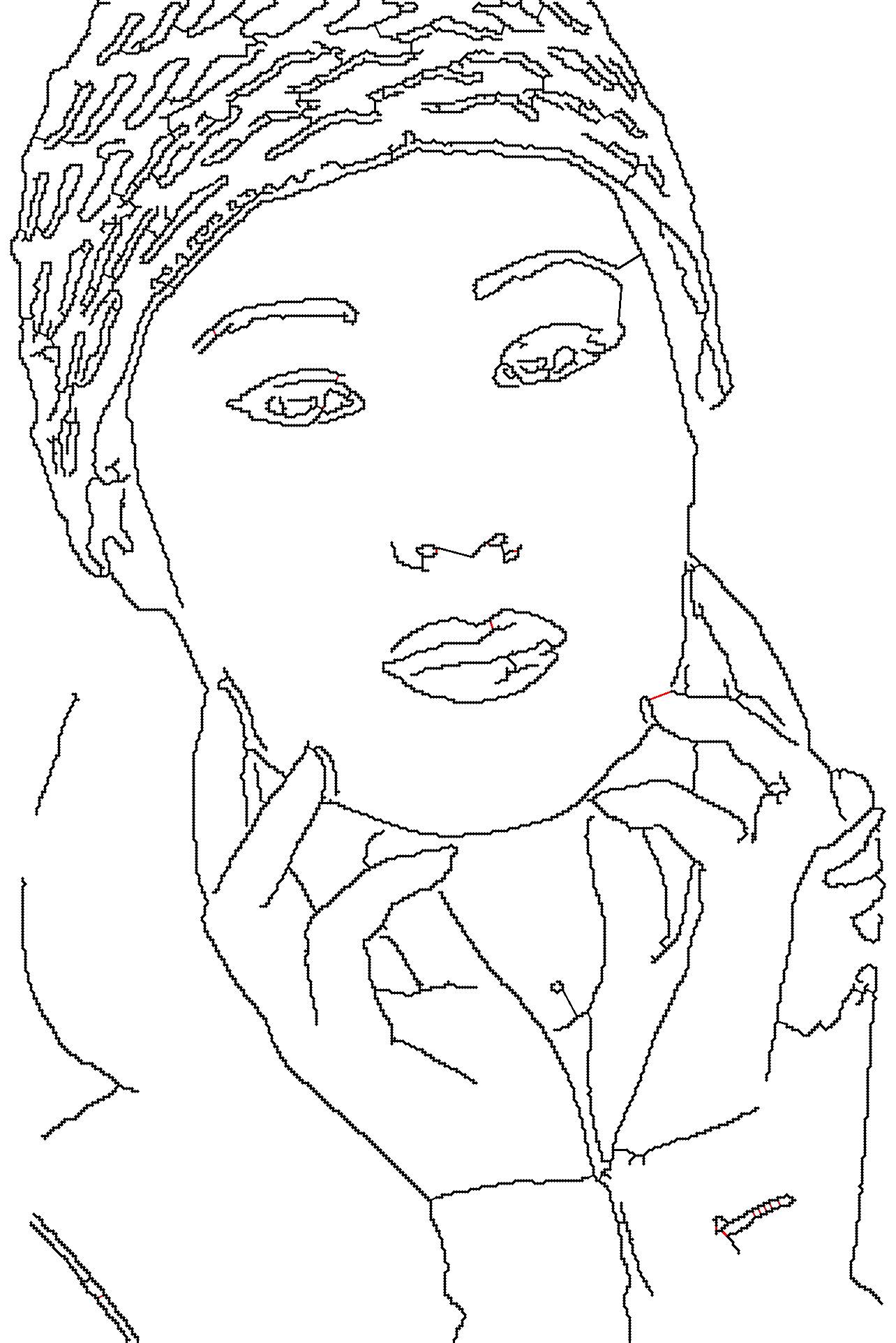}
\caption{Image 302003 with edge pixels in red, the outputs of Mapper, $\al$-Reeb, $\hopes_{1,1}$.}
\label{fig:woman}
\end{figure}
\vspace*{-4mm}


The key advantage of the Mapper algorithm is its flexibility due to many parameters, e.g. any clustering algorithm can be combined with Mapper.
When a suitable scale parameter $\al$ for a given $n$-point cloud $C$ can be quickly guessed, the $\al$-Reeb graph is computed with the very fast time $O(n\log n)$.
The experiments on 79K clouds and 500 BSD images in section~\ref{sec:experiments} have optimised the Mapper and $\al$-Reeb graphs over wide ranges of their important parameters.

Practically, by always using the justified heuristics (the 1st widest diagonal gap for $\hopes_{1,1}(C)$ and $\ep$ equal to the max death for simplification), the skeleton $\shopes(C)$ has comparable or better results without optimising any parameters in Fig.~\ref{fig:noise90}, \ref{fig:noise95} and Table~\ref{tab:bsd}.
All algorithms can be improved by a further minimisation of the RMS distance from a skeleton to a cloud $C\subset\R^d$.

Theoretically, a Homologically Persistent Skeleton $\hopes(C)$ has the advantage of being a parameter-free graph that is also embedded into the space of a cloud $C$.
Hence the time-consuming tuning of parameters is avoided and no intersections of edges are guaranteed. 
At the same time $\hopes$ remains flexible in the sense that after the full $\hopes(C)$ is found, one can quickly extract reduced subgraphs (for a fixed scale $\al$) and derived subgraphs $\hopes_{k,l}(C)$.

The paper gives detailed proofs of Optimality Theorem~\ref{thm:optimality} and Reconstruction Theorems~\ref{thm:reconstruction_1stgap},~\ref{thm:reconstruction_derived} for the first time.
The last two results are yet to be extended to higher dimensions.
Corollaries~\ref{cor:stability_hopes} and~\ref{cor:simhopes_homology} are new results.

The C++ code for all three algorithms is available at
https://github.com/Phil-Smith1/Cloud\_Skeletonization\_3.git.
The dataset of 79K point clouds (2GB) for comparison with other skeletonisation algorithms is available by request.

\bibliography{skeletonization-algorithms}

\newpage

\begin{center}
{\Large Supplementary materials: appendices A, B, C}
\end{center}

\setcounter{section}{0}
\renewcommand{\thesection}{\Alph{section}}
\renewcommand{\thesubsection}{\Alph{section}.\arabic{subsection}}

\section{Basic definitions and results from geometry and topology}
\label{sec:definitions}

\begin{definition}[metric space]
\label{dfn:metric}
A {\em metric space} $M$ is any set of points with a distance function  $d(p,q)$ between any points $p,q$ that satisfies these axioms:

\noindent 
$\bullet$ 
Positiveness: $d(p,q) \geq 0$ with equality if and only if $p = q$. 

\noindent 
$\bullet$ 
Symmetry: $d(p, q) = d(q, p)$ for any points $p,q\in C$.

\noindent 
$\bullet$ 
The triangle inequality: $d(p, q) + d(q, r) \geq d(p, r)$ for any $p,q,r\in C$.
\smallskip

\noindent 
For any subset $C$ of a metric space $M$ and a parameter $\al\geq 0$, the $\al$-offset $C^{\al}=\{p\in M \mid d(p,q)\leq\al \mbox{ for some }q\in M\}$ is the union of closed balls with the radius $\al$ and centres at all points of $C$, e.g. $C^0=C$. 
If a metric space $M$ consists of only finitely many points, this space will be called a {\em point cloud}.
\bsquare
\end{definition}

The Euclidean distance for points $p,q\in\R$ is $d(p,q)=|p-q|$. 
The function $(p-q)^2$ fails the triangle inequality above for the points $p=-1,q=0,r=1$. 

\begin{definition}[metric graph and neighbourhood graph]
\label{dfn:metric graph}
A \textit{graph} is a finite set of vertices with edges being unordered pairs of vertices.
A \textit{metric graph} is a graph that has a length assigned to each edge. 
The distance between two vertices is the minimum total length of any path (if it exists) from one vertex to the other.
An example of a metric graph is the \textit{neighbourhood graph} $N(C, \ep)$ of a point cloud $C$ with a threshold $\ep$. 
The vertex set of the graph $N(C, \ep)$ is the cloud $C$.
Two vertices $p,q\in C$ are joined by an edge in $N(C, \ep)$ if the distance $d(p,q)\leq\ep$. 
The length of this edge is fixed as the distance $d(p,q)$.
\bsquare
\end{definition}

\begin{definition}[graph homeomorphism]
\label{dfn:homeomorphism}
A graph $G$ is {\em connected} if any two vertices of $G$ are connected by a sequence of edges such that any successive edges share a common vertex.
Connected graphs $G,H$ are {\em homeomorphic}, $G\simeq H$, if, ignoring all vertices of degree 2, there exists a bijection $\psi$ between the vertices of $G$ and $H$ that respects edges, i.e. any vertices $u,v\in G$ are connected by $m$ edges in $G$ if and only if $\psi(u),\psi(v)\in H$ are connected by $m$ edges in $H$.  
\bsquare
\end{definition}

The graphs $W(4)$ and $G(2,2)$ in Fig.~\ref{fig:graphs} are homeomorphic, because four corner vertices of degree~2 in $G(2,2)$ are topologically trivial so that $G(2,2)$ can be mapped to $W(4)$ via a continuous bijection whose inverse is continuous.
\smallskip

The vertex set $C$ of any connected metric graph $G$ is a point cloud in the sense of Definition~\ref{dfn:metric}, where the distance $d(p,q)$ between any vertices $p,q\in C$ is the total length (sum of edge-lengths) of a shortest path from $p$ to $q$ in $G$.

\begin{definition}[simplicial complex]
\label{dfn:complex} 
A {\em simplicial complex} $Q$ with a finite vertex set $C$ is a collection of subsets $\{v_0, \dots, v_k\} \subset C$ called \textit{simplices} so that 
\smallskip

\noindent $\bullet$ any subset (called a \textit{face}) of a simplex is also a simplex and is included in $Q$; \smallskip

\noindent $\bullet$ any non-empty intersection of simplices is their common face included in $Q$. \medskip

\noindent Any simplex on $k+1$ vertices has the {\em dimension} $k$ and {\em geometric} realisation:
\[\De^k = \{(t_1, \dots, t_{k}) \in \R^{k} \mid t_1 + \dots + t_{k} \leq 1, \; t_i \geq 0\} \subset \R^{k}.\]
The \textit{geometric realisation} of $Q$ is obtained by gluing realisations of all its simplices along their common faces. Hence $Q$ inherits the Euclidean topology. The dimension of a complex $Q$ is the maximum dimension of its simplices.
\bsquare
\end{definition}

A 1-dimensional complex $Q$ is a graph $G$ without loops (edges connecting a vertex to itself) and multiple edges that have the same endpoints.
A geometric realisation of $G$ is a continuous function $f:G\to\R^k$ mapping edges to straight line segments in $\R^k$ that can meet only at common endpoints, see Fig.~\ref{fig:graphs}. 

\begin{definition}[\Cech and Vietoris-Rips complexes, Delaunay triangulations and $\al$-complexes]
\label{dfn:Cech+VR}
Let $C$ be any set in a metric space $M$. 
For abstract data, the space $M$ can coincide with a given cloud $C$.
The \textit{\Cech} complex $\Ch(C, M; \al)$ has a simplex on $v_0, \dots, v_k \in C$ if the full intersection of $k + 1$ closed balls with a radius $\al$ and centres $v_0, \dots, v_k$ is not empty, i.e. contains a point from the ambient space $M$, see \cite[section~4.2.3]{CdSO14} and Fig.~\ref{fig:Cech+VR}.
The \textit{Vietoris-Rips} complex $VR(C, M; \al)$ has a simplex on vertices $v_0, \dots, v_k \in C$ if each pairwise intersection of the above $k + 1$ balls is not empty in $M$.
For a cloud $C\subset\R^d$, a \emph{Delaunay} triangulation $\Del(C)$ consists of simplices (with vertices at points of $C$) whose minimal open circumballs contain no points of $C$.
For $\al\geq 0$, the \emph{$\al$-complex} $C(\al)\subset\Del(C)$ consists of all Delaunay simplices whose circumradii are at most $\al$, i.e. $C(\al)$ grows from $C$ at $\al=0$ to $\Del(C)$ when $\al\to+\infty$.
\bsquare
\end{definition}

Lemma~\ref{lem:nerve} says that the \Cech complex correctly represents the geometric shape of any union of balls.
The Vietoris-Rips complex is determined by its 1-dimensional skeleton, hence requires much less storage than the \Cech complex.
The $\al$-complex $C(\al)$ is preferred for low $d$, because any $C(\al)$ lives in $\R^d$, while \Cech and VR complexes can be only projected to $\R^d$ (often with intersections).

\begin{lemma}
\label{lem:nerve}
\cite[Theorem~3.2]{edelsbrunner1995union} 
Let $C$ be a subspace of a metric space $M$. Then any $\al$-offset $C^{\al} \subset M$ is homotopy equivalent to the \Cech complex $\Ch(C, M; \al)$.
For a finite cloud of points $C\subset\R^d$ in a Euclidean space, any $\al$-offset $C^{\al}\subset\R^d$ is homotopy equivalent to the $\al$-complex $C(\al)\subset\R^d$.
\end{lemma}

\begin{figure}
\centering
\begin{minipage}{0.6\textwidth}
\includegraphics[width=\textwidth]{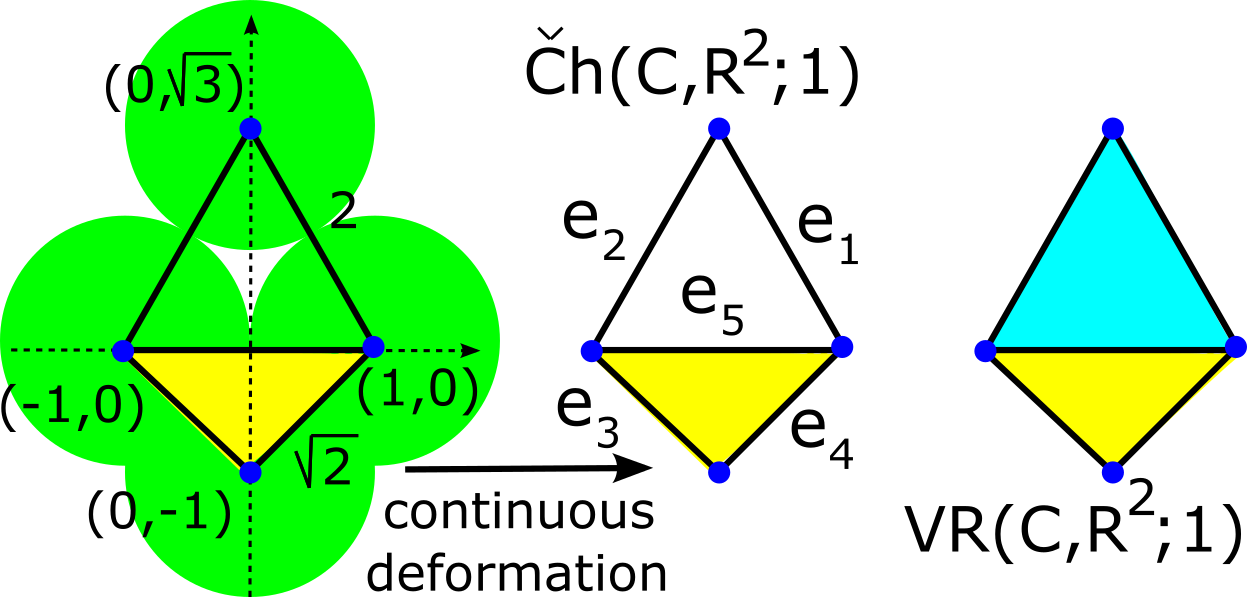}
\captionof{figure}{The offset $C^1$ of the 4-point cloud $C$ deforms to the \Cech complex, not to the Vietoris-Rips complex at scale 1, which coincides with $\Del(C)\subset\R^2$ in this case.}
\label{fig:Cech+VR}
\end{minipage} \hspace{1em}
\begin{minipage}{0.32\textwidth}
\includegraphics[width=\textwidth]{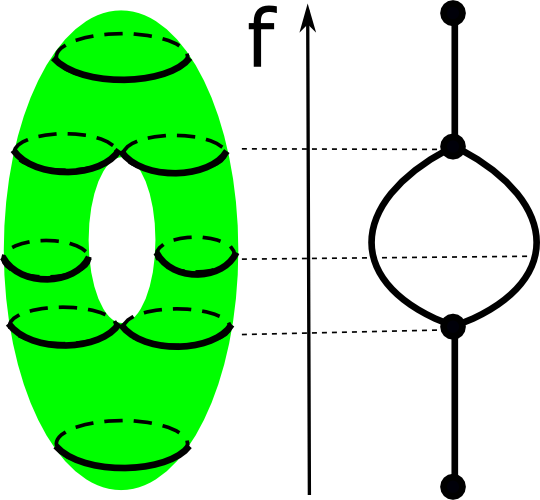}
\captionof{figure}{The Reeb graph of the height function on a torus.}
\label{fig:Reeb_graph}
\end{minipage}
\end{figure}

\begin{definition}[Reeb graphs]
\label{dfn:Reeb}
For a simplicial complex $Q$ and a function $f : Q \rightarrow \R$, the {\em level set} of $f$ corresponding to a value $t \in \R$ is the set of points $L_t(f) = \{x \in Q \mid f(x) = t\}$.
We define, for points $x, y \in Q$, the equivalence relation $\sim$ such that $x \sim y$ if and only if $f(x) = f(y)$ and $x$ and $y$ are in the same connected component of the level set of $f(x)$, see Fig.~\ref{fig:Reeb_graph}.
The \textit{Reeb graph} $\Reeb(Q, f)$ is the quotient space of $Q$ formed by mapping all points that are equivalent to each other under the relation $\sim$ to a single point.
\end{definition}

Definition~\ref{dfn:Reeb} makes sense when $Q$ is a topological space, which is more general than a simplicial complex.
A Homologically Persistent Skeleton ($\hopes$) is essentially based on the concept of homology introduced below.
A 1-dimensional cycle in a complex $Q$ is a sequence of edges $e_1, \dots, e_k$ such that $e_i$ and $e_{i + 1}$ have a common endpoint, where $e_{k + 1} = e_1$. 
We define the first homology group $H_1(Q)$ of a simplicial complex $Q$ only with coefficients in the group $\Z_2 = \Z / 2\Z = \{0, 1\}$.

\begin{definition}[homology $H_1(Q)$ of a complex]
\label{dfn:1dim-homology}
Cycles of a complex $Q$ can be algebraically written as finite linear combinations of edges (with coefficients only $0$ or $1$) and generate the vector space $C_1$ of cycles. The boundaries of all triangles in $Q$ are cycles of 3 edges and generate the subspace $B_1 \subset C_1$. The quotient space $C_1 / B_1$ is the \emph{homology} group $H_1(Q)$. 
The operation is the addition of cycles, the empty cycle is the zero. We define the $k$-th \emph{Betti number} of a complex $Q$ to be the rank of the $k$-th homology group. 
For instance, the first Betti number of a complex $Q$ is the rank of the first homology group, and is equivalent to the number of linearly independent 1-dimensional cycles in $Q$.
\end{definition}

The complex $Q=\Ch(C, \R^2; \al)$ in Fig.~\ref{fig:Cech+VR} has 5 edges $e_i$, $i=1,\dots,5$. The vector space $C_1$ is generated by two triangular cycles $e_1+e_2+e_5$ and $e_3+e_4+e_5$.
The subspace $B_1$ is generated by the boundary cycle $e_3+e_4+e_5$ of the yellow triangle.
Hence $H_1(Q)=C_1/B_1=\Z_2$ is generated by the cycle $e_1+e_2+e_5$.

\section{A detailed review of DBSCAN, Mapper and $\al$-Reeb algorithms}
\label{sec:algorithms}

\subsection{DBSCAN: the clustering algorithm used in our Mapper implementation}
\label{sub:DBSCAN}
	
DBSCAN is a density-based spatial clustering for applications with noise. 
Mapper in subsection~\ref{sub:Mapper} should work for any number of clusters in subclouds, hence we can not use clustering that needs a specified number of clusters, e.g. $k$-means. 
We have also tried the single-edge clustering and found that DBSCAN works better for noisy samples from section~\ref{sec:dataset}.
DBSCAN requires two parameters: (1) $\ep$ is the radius around a point within which we search for neighbours; and (2) minPts is the number of points required within a neighbourhood of a point before a cluster can be formed. 
Here are the stages of DBSCAN.
\smallskip
	
	\noindent
	$\bullet$
	\textbf{Stage 1:} A single point $p_1$ is selected at random, and we compute the set $\Nbhd(p_1)$ of all points within a distance $\ep$ of $p_1$. If $|\Nbhd(p_1)| < \text{minPts}$, then $p_1$ is labelled as noise. If however $|\Nbhd(p_1)| \geq \text{minPts}$, we label all points in $\Nbhd(p_1)$ (not already belonging to another cluster, but may have previously been labelled as noise) as belonging to the cluster of $p_1$. \smallskip
	
	\noindent
	$\bullet$
	\textbf{Stage 2:} Then, looping over all points $p_i$ in $\Nbhd(p_1)$ (apart from $p_1$), we compute $\Nbhd(p_i)$. All points in $\Nbhd(p_i)$ (not already belonging to another cluster) are labelled as belonging to the cluster of $p_1$, and, if $|\Nbhd(p_i)| \geq \text{minPts}$, we add all points in $\Nbhd(p_i)$ to $\Nbhd(p_1)$. Hence, this loop will finish only when all points in $p_1$'s cluster have been identified. \smallskip
	
	\noindent
	$\bullet$
	\textbf{Stage 3:} A new point $p_2$, that is not already labelled, is selected, and we continue as in the first two stages for $p_2$ instead of $p_1$. The process continues until all points are assigned to a cluster or are labelled as noise. \smallskip

\subsection{The Mapper graph: a network of interlinked clusters}
\label{sub:Mapper}

\noindent 
The Mapper algorithm is a skeletonisation framework introduced by G.~Carlsson et al. \cite{singh2007topological}, which aims to give a simple description of a dataset by a network of interlinked clusters. 
Mapper takes as input a point cloud $C$ given by pairwise distances with extra parameters below, and outputs a simplicial complex. 
\smallskip

\noindent
\textbf{A filter function} is a function $f : C \to Y$ whose value should be given for all points in a data cloud $C$. 
The parameter space $Y$ will often be $\R$, but could also be $\R^2$ or $S^1$. The type of a filter function is up to the user, with common examples being a density estimator, the Euclidean distance from a base point, or the distance from a root point within a neighbourhood graph on $C$. 
\smallskip

\noindent
\textbf{A covering of the range of $f$}. 
The range of the filter function $f$ must be covered by overlapping regions, e.g. line intervals for $Y=\R$. 
A cover usually has two parameters -- the number of regions and the ratio of overlap -- which can be used to control the resolution of the output simplicial complex. 
\smallskip

\noindent
\textbf{A clustering algorithm}. 
The Mapper algorithm requires subsets of the input point cloud $C$ to be clustered.
So a clustering algorithm must be chosen, with a good choice again being dependent on the type of data being analysed. 
\smallskip

\noindent
The main stages of the Mapper algorithm are now described below. \smallskip

\noindent
\textbf{Stage 1}. 
If values of a given filter function $f$ are not yet given explicitly, then the function $f$ is computed on all points of a point cloud $C$, see Fig.~\ref{fig:Mapper}.
\smallskip

\noindent
\textbf{Stage 2}. 
For each region $I$ in the covering of the range of $f$, we cluster the preimage $f^{-1}(I)=\{p\in C \mid f(p)\in I\}$ according to the given clustering algorithm. 
Each cluster is represented as a vertex in the output complex $Q$. 
\smallskip

\noindent
\textbf{Stage 3}. 
If any $k\geq 2$ resulting clusters (across all regions in the covering) share a common point of $C$, the corresponding $k$ vertices (representatives of $k$ clusters) span a $(k-1)$-dimensional simplex in the Mapper complex $Q$.
\smallskip

A simple way to visualise the Mapper complex $Q$ is to project $Q$ to the space where the cloud $C$ lives by mapping every vertex to the average point of the corresponding cluster.
If edges are drawn as straight-line segments, e.g. in $\R^2$, then intersections may be unavoidable, see the 2nd picture in Fig.~\ref{fig:outputsg32}.
\smallskip

Mapper is a versatile tool that if used rightly can be a useful method in visualising large datasets. A significant difficulaty of the method however is that with so many choices to be made, the user often requires existing knowledge of the dataset in order to select parameters that will give meaningful outputs.
\smallskip

The experiments in Section~\ref{sec:experiments} use a filter function mapping a point cloud $C$ to $\R$, and therefore Mapper outputs a one-dimensional complex, which is a graph. 
The filter function is the Euclidean distance from a base point, which is the most distant from a random point of $C$. 
The overlap percentage of two adjacent intervals in the covering is $50\%$.
The clustering algorithm is DBSCAN,
which has two parameters: 1) a radius $\ep$ around a point within which we search for neighbours, and 2) the minimum number of points in a cluster.
We found it acceptable to fix this minimum at 5 and optimised $\ep$ in the experiments.

\begin{figure}
	\centering
	\begin{minipage}{0.47\textwidth}
		\includegraphics[width=\textwidth]{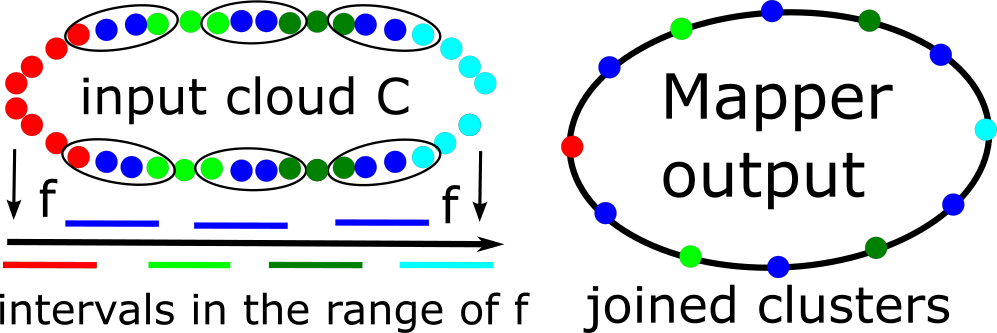}
		\captionof{figure}{\textbf{Left}: cloud $C$ with a filter function $f:C\to\R$ has range is covered by overlapped intervals. \textbf{Right}: Mapper outputs a graph whose nodes represent clusters and links join overlapped clusters.}
		\label{fig:Mapper}
	\end{minipage} \hspace{1em}
	\begin{minipage}{0.47\textwidth}
		\centering
		\includegraphics[width=\textwidth]{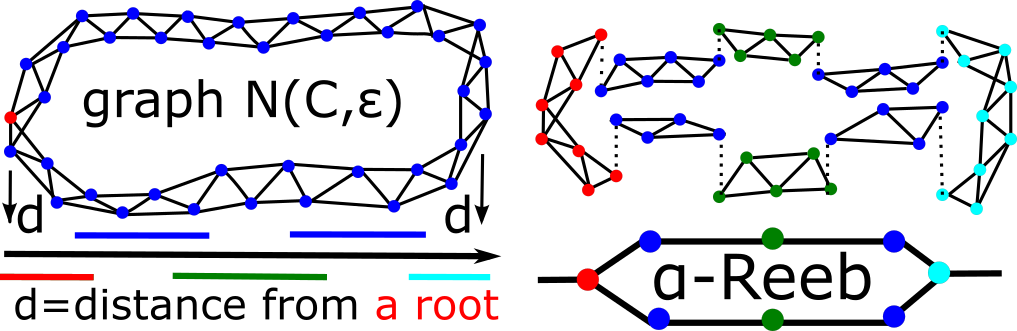}
		\captionof{figure}{\textbf{Left}: a neighbourhood graph $N(C,\ep)$ with a distance function $d$ from a root. \textbf{Right}: nodes of the $\al$-Reeb graph represent connected subgraphs in preimages $d^{-1}(I_i)$ of overlapped intervals $I_i$.}
		\label{fig:alpha-Reeb}
	\end{minipage}
\end{figure}

\subsection{The $\al$-Reeb graph is a parametric version of the Reeb graph}
\label{sub:Reeb}

The Reeb graph, introduced in Definition~\ref{dfn:Reeb}, is a simplified representation of a simplicial complex. If, instead of having an entire simplicial complex $Q$, we only have a finite number of points sampled from $Q$, even approximating the Reeb graph of $(Q, f)$ is not straightforward. 
Therefore, to bridge this barrier, F.~Chazal et al. introduced a parametric version called the $\al$-Reeb graph \cite{chazal2015gromov}.

\begin{definition}
	\label{dfn:alpha-Reeb}
	Let $(Q,d)$ be a simplicial complex $Q$ with a continuous function $d:Q\to\R$. 
	Let $\al > 0$ and let $\mathcal{I} = \{I_i\}$ be a covering of the range of $d$, where each $I_i$ is a closed interval of the length $\al$. 
	Consider the transitive closure of the following relation on points of $Q$:
	$x \sim_\al y$ if and only if $d(x) = d(y)$ and $x$ and $y$ are in the same connected component of $d^{-1}(I_i)= \{x \in Q \: | \: d(x) \in I_i\}$ for some interval $I_i \in \mathcal{I}$. 
	Then the $\al$-\emph{Reeb graph} associated to the covering $\mathcal{I}$ of a simplicial complex $Q$, is the quotient space formed from $Q$ by mapping all points that are equivalent to each other under $\sim_\al$ to a single point.
\end{definition}

Here are the stages of the $\al$-Reeb algorithm \cite[section~6]{chazal2015gromov}.
\smallskip

\noindent
\textbf{Stage 0}.
To get a simplicial complex $Q$ needed for the $\al$-Reeb graph in Definition~\ref{dfn:alpha-Reeb}, usually one creates a neighbourhood graph $N(C,\ep)$ where all points of $C$ at a distance less than $\ep$ are connected by an edge, see Fig.~\ref{fig:alpha-Reeb}. 
\smallskip

\noindent
\textbf{Stage 1}. 
Select a root vertex in $N(C,\ep)$, e.g. as the most distance vertex from a random one.
For each vertex, calculate its distance from the root within the graph, which gives the function $d : V \to \R$ on the vertex set $V$ of $N(C, \ep)$. 
\smallskip

\noindent
\textbf{Stage 2}. 
Let $\max(d)$ be the maximal value of $d:V\to\R$. 
For any $0<\al<\max(d)$, the range of $d$ is covered by the overlapped intervals $\mathcal{I} = \{I_i\}_{0 \leq i \leq m}$, where $I_i = \{[\frac{i\alpha}{2}, \frac{i\alpha}{2} + \alpha]\}$, $m$ is the smallest integer such that $m \geq \frac{2(\max(d) - \al)}{\al}$. \smallskip

\noindent
\textbf{Stage 3}.
For each interval in the covering $\mathcal{I}$, we consider its preimage $d^{-1}(I_i)$ in the vertex set $V$.
We add an edge between two vertices in the preimage if an edge also existed between these vertices in the neighbourhood graph $N(C,\ep)$. 
So each interval $I_i$ gives rise to a subgraph of $N(C,\ep)$. 
Every connected component of these subgraphs corresponds to a vertex in an intermediate graph $G$. 
\smallskip

\noindent
\textbf{Stage 4}. 
Two vertices in the intermediate graph $G$ are connected by an edge if their corresponding components share a common vertex in $N(C,\ep)$.
Such common vertices are joined by dotted arcs in the last picture of Fig.~\ref{fig:alpha-Reeb}.
\smallskip

\noindent
\textbf{Stage 5}. 
For each vertex $v$ in our intermediate graph $G$, we take a copy of the interval $J_v$, which are split by $v$ into two halves and partially ordered according to $\mathcal{I} = \{I_i\}_{0 \leq i \leq m}$. 
The $\alpha$-Reeb graph is the quotient of the disjoint union of these intervals $J_v$. 
The top half of $J_v$ is identified with the bottom half of all higher intervals $J_u$ such that the vertices $u,v$ are joined by an edge of $G$. 
\smallskip

Similarly to Mapper, the $\al$-Reeb graph aims to join close clusters, however clustering in preimages of intervals is replaced by finding connected subgraphs.
Informally, the scale $\al$ helps to identify points that can be connected by paths of lengths at most $\al$ within a neighbourhood graph.
In the limit $\al\to 0$, the $\al$-Reeb graph becomes the Reeb graph from Definition~\ref{dfn:Reeb}.
Hence the $\al$-Reeb graph essentially depends on $\al$ whose wide range is explored in subsection~\ref{sub:parameters}.

\section*{Appendix C: more experimental results on synthetic and real data}

\subsection*{C.1. Comparisons of algorithms on noisy samples of graphs}

\begin{figure}[H]
	\centering
	\def\svgwidth{\columnwidth}
	\includegraphics[width=\textwidth]{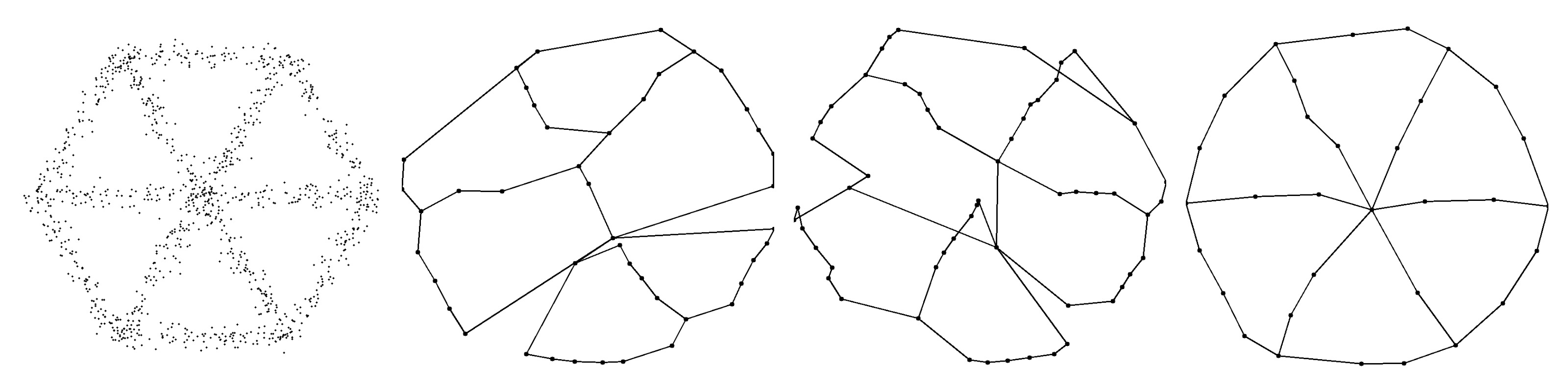}
	\caption{\textbf{1st}: a cloud $C$ sampled from the $W(6)$ graph with Gaussian noise with $\sigma = 0.04$; 
\textbf{2nd}: Mapper output on $C$; 
\textbf{3rd}: $\al$-Reeb$(C)$; 
\textbf{4th}: $\shopes(C)$ by Algorithm~\ref{alg:contract}.}
	\label{fig:wheel6_Gaussian}
\end{figure}

\begin{figure}[H]
	\centering
	\def\svgwidth{\columnwidth}
	\includegraphics[width=\textwidth]{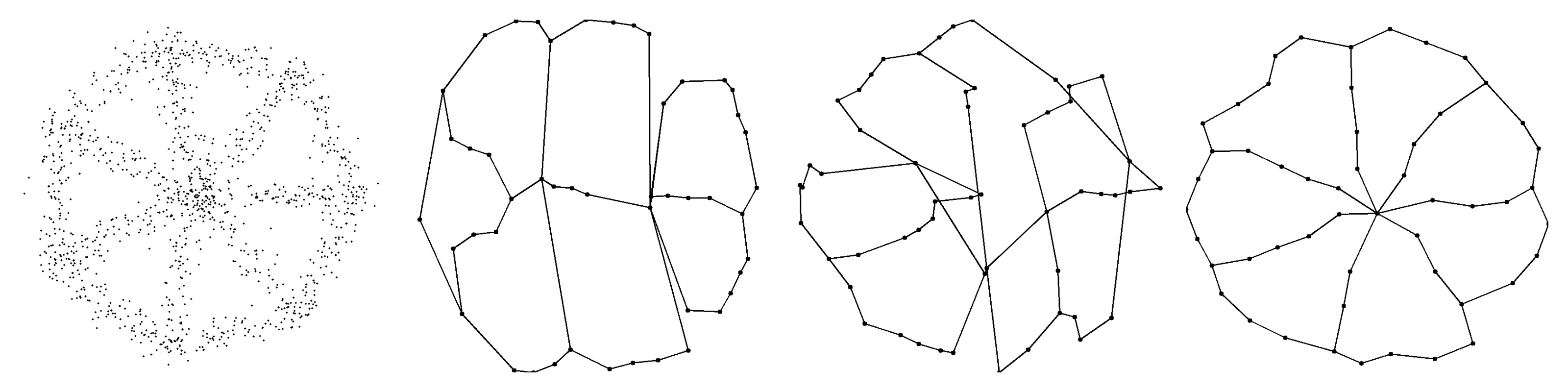}
	\caption{\textbf{1st}: a cloud $C$ sampled from the $W(7)$ graph with Gaussian noise with $\sigma = 0.06$; 
\textbf{2nd}: Mapper output on $C$; 
\textbf{3rd}: $\al$-Reeb$(C)$; 
\textbf{4th}: $\shopes(C)$ by Algorithm~\ref{alg:contract}.}
	\label{fig:wheel7_Gaussian}
\end{figure}

\begin{figure}[H]
	\centering
	\def\svgwidth{\columnwidth}
	\includegraphics[width=\textwidth]{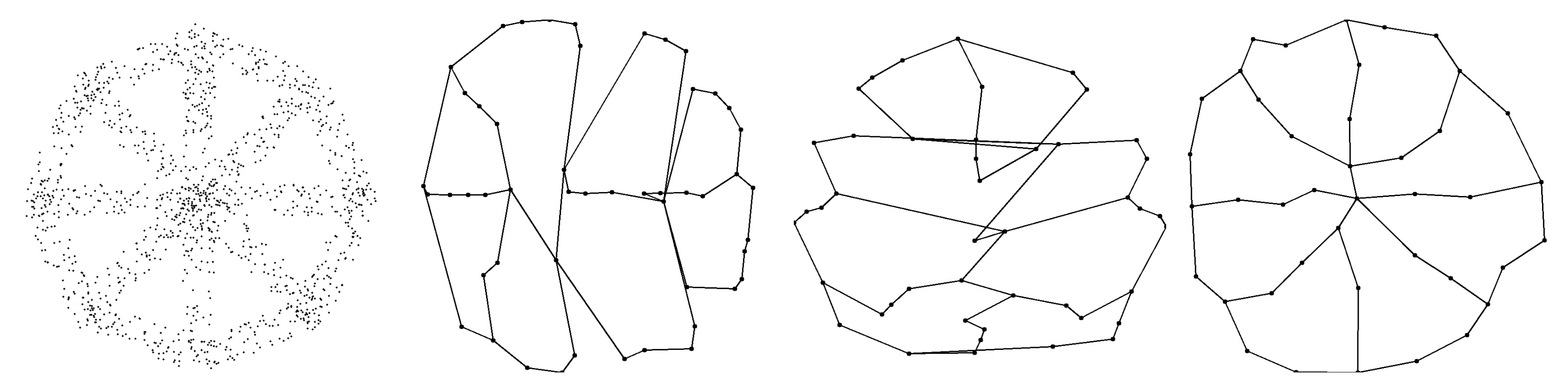}
	\caption{\textbf{1st}: a cloud $C$ sampled from the $W(8)$ graph with uniform noise with $\mu = 0.1$; 
\textbf{2nd}: Mapper output on $C$; 
\textbf{3rd}: $\al$-Reeb$(C)$; 
\textbf{4th}: $\shopes(C)$ by Algorithm~\ref{alg:contract}.}
	\label{fig:wheel8_uniform}
\end{figure}

\begin{figure}[H]
	\centering
	\def\svgwidth{\columnwidth}
	\includegraphics[width=\textwidth]{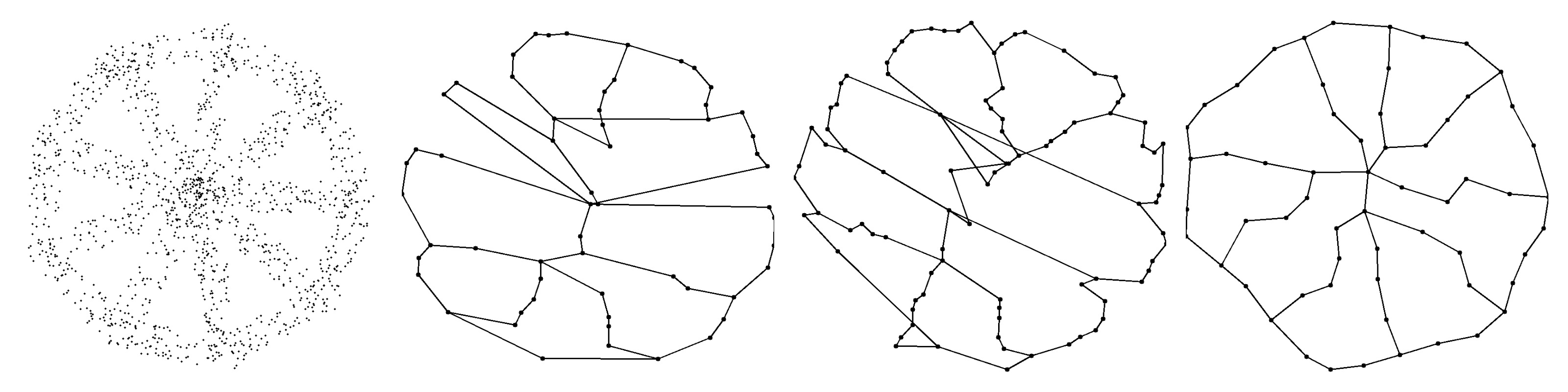}
	\caption{\textbf{1st}: a cloud $C$ sampled from the $W(9)$ graph with uniform noise with $\mu = 0.1$. 
\textbf{2nd}: Mapper output on $C$; 
\textbf{3rd}: $\al$-Reeb$(C)$; 
\textbf{4th}: $\shopes(C)$ by Algorithm~\ref{alg:contract}.}
	\label{fig:wheel9_uniform}
\end{figure}

\begin{figure}[H]
	\centering
	\def\svgwidth{\columnwidth}
	\includegraphics[width=\textwidth]{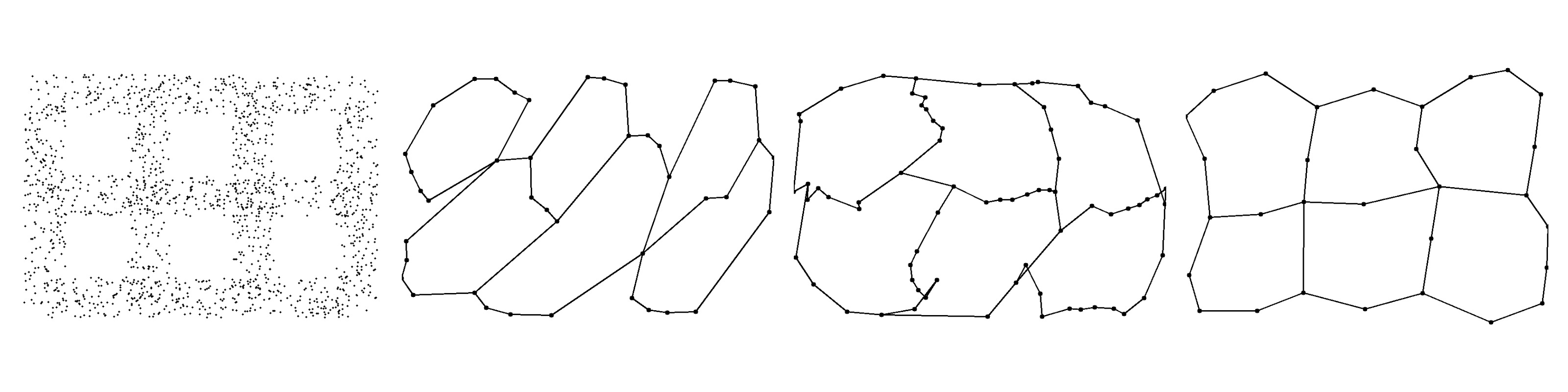}
	\caption{\textbf{1st}: a cloud $C$ sampled from the $G(3,2)$ graph with uniform noise with $\mu = 0.2$. 
\textbf{2nd}: Mapper output on $C$; 
\textbf{3rd}: $\al$-Reeb$(C)$; 
\textbf{4th}: $\shopes(C)$ by Algorithm~\ref{alg:contract}.}
	\label{fig:grid32_uniform}
\end{figure}

\begin{figure}[H]
	\centering
	\def\svgwidth{\columnwidth}
	\includegraphics[width=\textwidth]{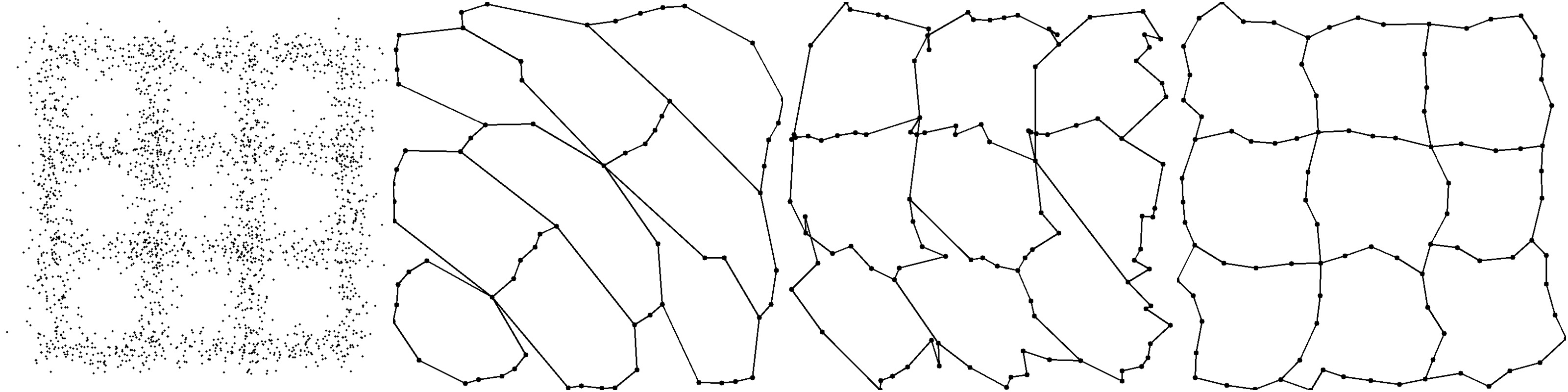}
	\caption{\textbf{1st}: a cloud $C$ sampled from the $G(3,3)$ graph with Gaussian noise with $\sigma = 0.12$. 
\textbf{2nd}: Mapper output on $C$; 
\textbf{3rd}: $\al$-Reeb$(C)$; 
\textbf{4th}: $\shopes(C)$ by Algorithm~\ref{alg:contract}.}
	\label{fig:grid33_Gaussian}
\end{figure}

\begin{figure}[H]
	\centering
	\def\svgwidth{\columnwidth}
	\includegraphics[width=\textwidth]{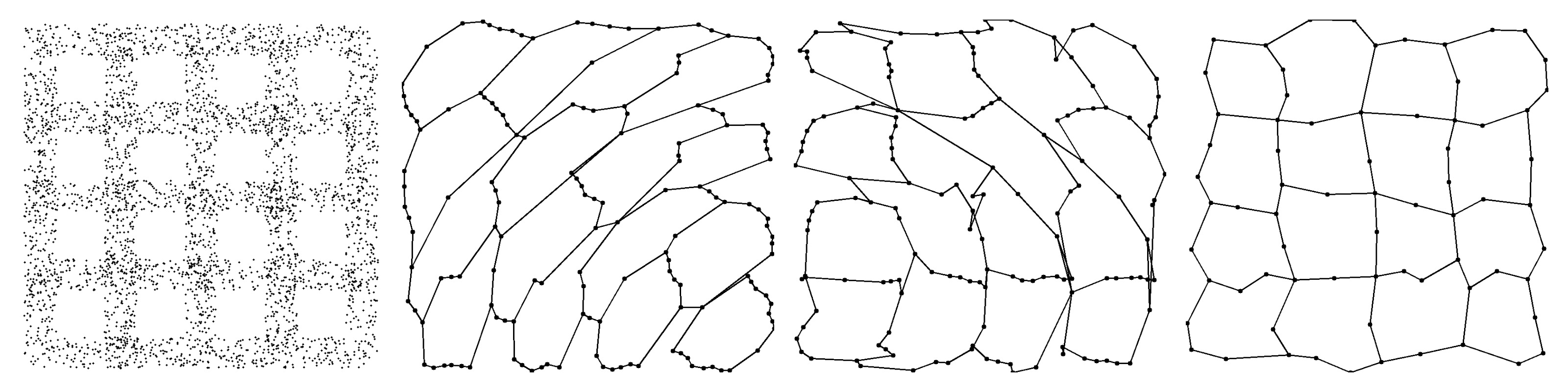}
	\caption{\textbf{1st}: a cloud $C$ sampled from the $G(4,4)$ graph with uniform noise with $\mu = 0.2$. 
\textbf{2nd}: Mapper output on $C$; 
\textbf{3rd}: $\al$-Reeb$(C)$; 
\textbf{4th}: $\shopes(C)$ by Algorithm~\ref{alg:contract}.}
	\label{fig:grid44_uniform}
\end{figure}

\begin{figure}[H]
	\centering
	\def\svgwidth{\columnwidth}
	\includegraphics[width=\textwidth]{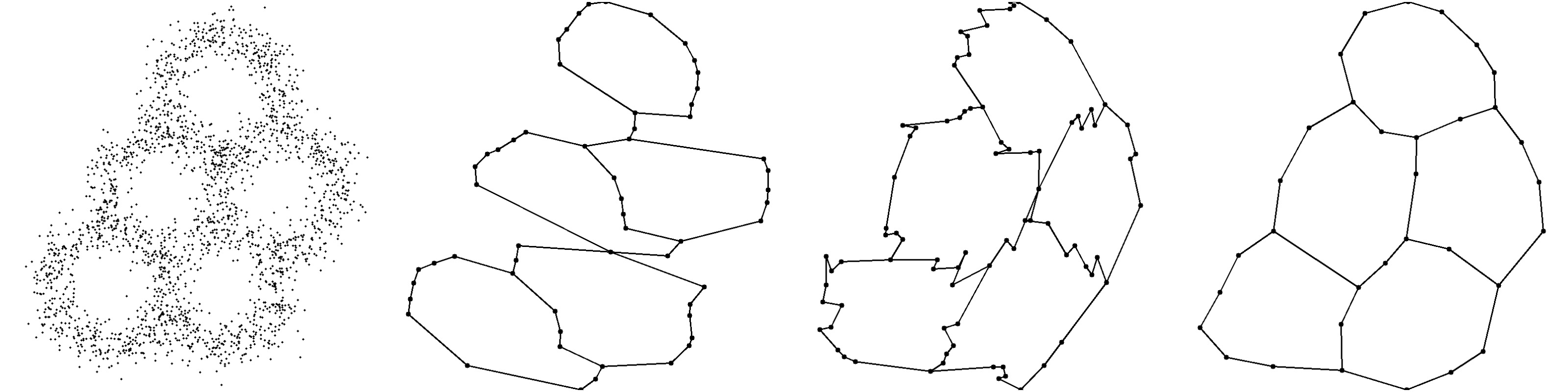}
	\caption{\textbf{1st}: a cloud $C$ sampled from the $H(5)$ graph with Gaussian noise with $\sigma = 0.2$. 
\textbf{2nd}: Mapper output on $C$; 
\textbf{3rd}: $\al$-Reeb$(C)$; 
\textbf{4th}: $\shopes(C)$ by Algorithm~\ref{alg:contract}.}
	\label{fig:hex5_Gaussian}
\end{figure}

\begin{figure}[H]
	\centering
	\def\svgwidth{\columnwidth}
	\includegraphics[width=\textwidth]{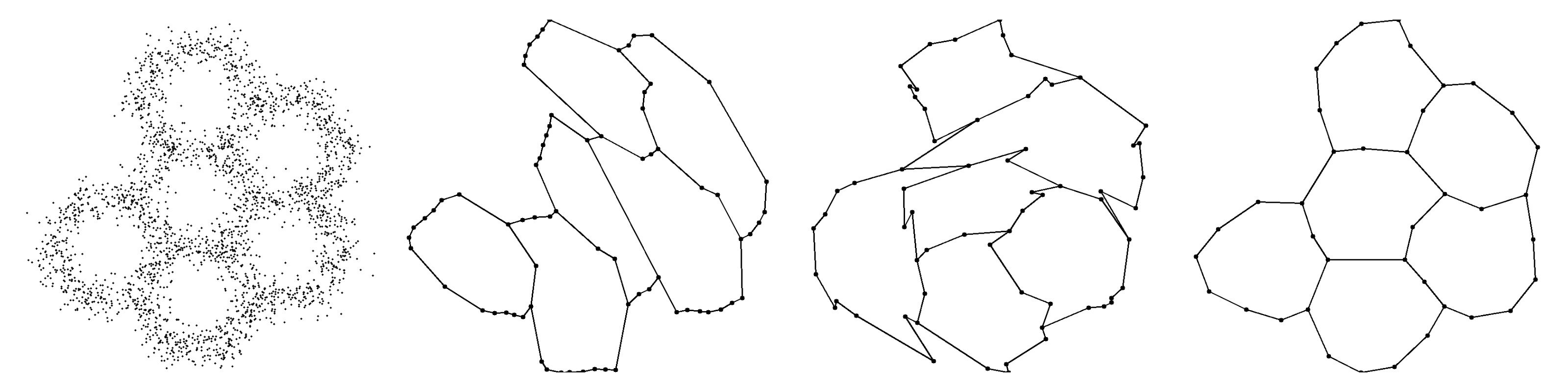}
	\caption{\textbf{1st}: a cloud $C$ sampled from the $H(6)$ graph with Gaussian noise with $\sigma = 0.18$. 
\textbf{2nd}: Mapper output on $C$; 
\textbf{3rd}: $\al$-Reeb$(C)$; 
\textbf{4th}: $\shopes(C)$ by Algorithm~\ref{alg:contract}.}
	\label{fig:hex6_Gaussian}
\end{figure}

\begin{figure}[H]
	\centering
	\def\svgwidth{\columnwidth}
	\includegraphics[width=\textwidth]{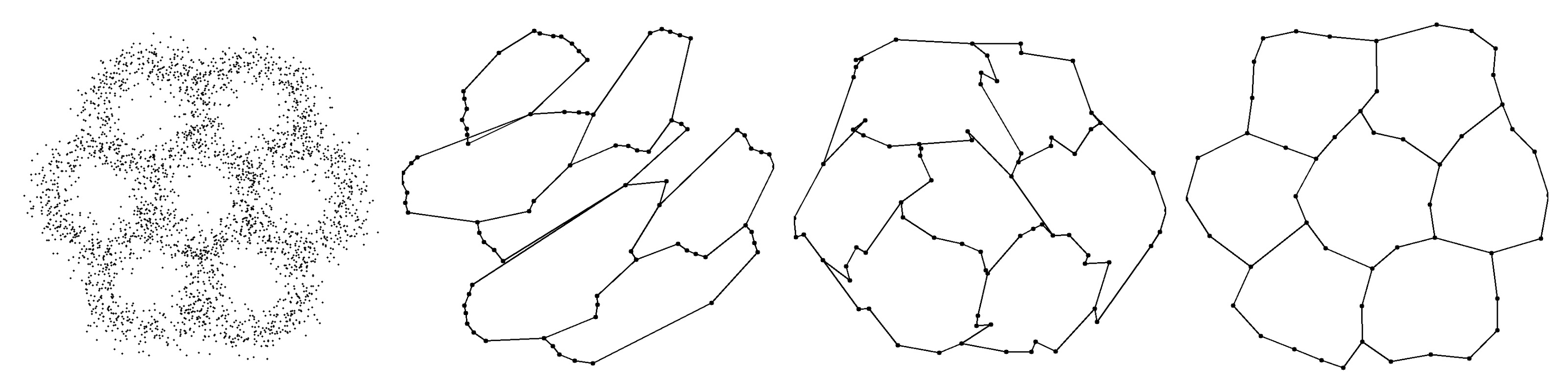}
	\caption{\textbf{1st}: a cloud $C$ sampled from the $H(7)$ graph with Gaussian noise with $\sigma = 0.2$. 
\textbf{2nd}: Mapper output on $C$; 
\textbf{3rd}: $\al$-Reeb$(C)$; 
\textbf{4th}: $\shopes(C)$ by Algorithm~\ref{alg:contract}.}
	\label{fig:hex7_Gaussian}
\end{figure}

\subsection*{C.2. Comparisons of algorithms on edge pixels in BSD images}

\begin{figure}[H]
	\centering
	\def\svgwidth{\columnwidth}
	\includegraphics[height=42mm]{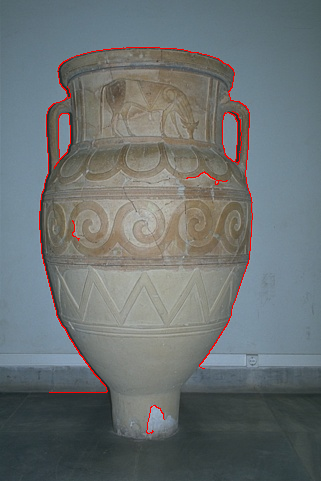}
	\includegraphics[height=42mm]{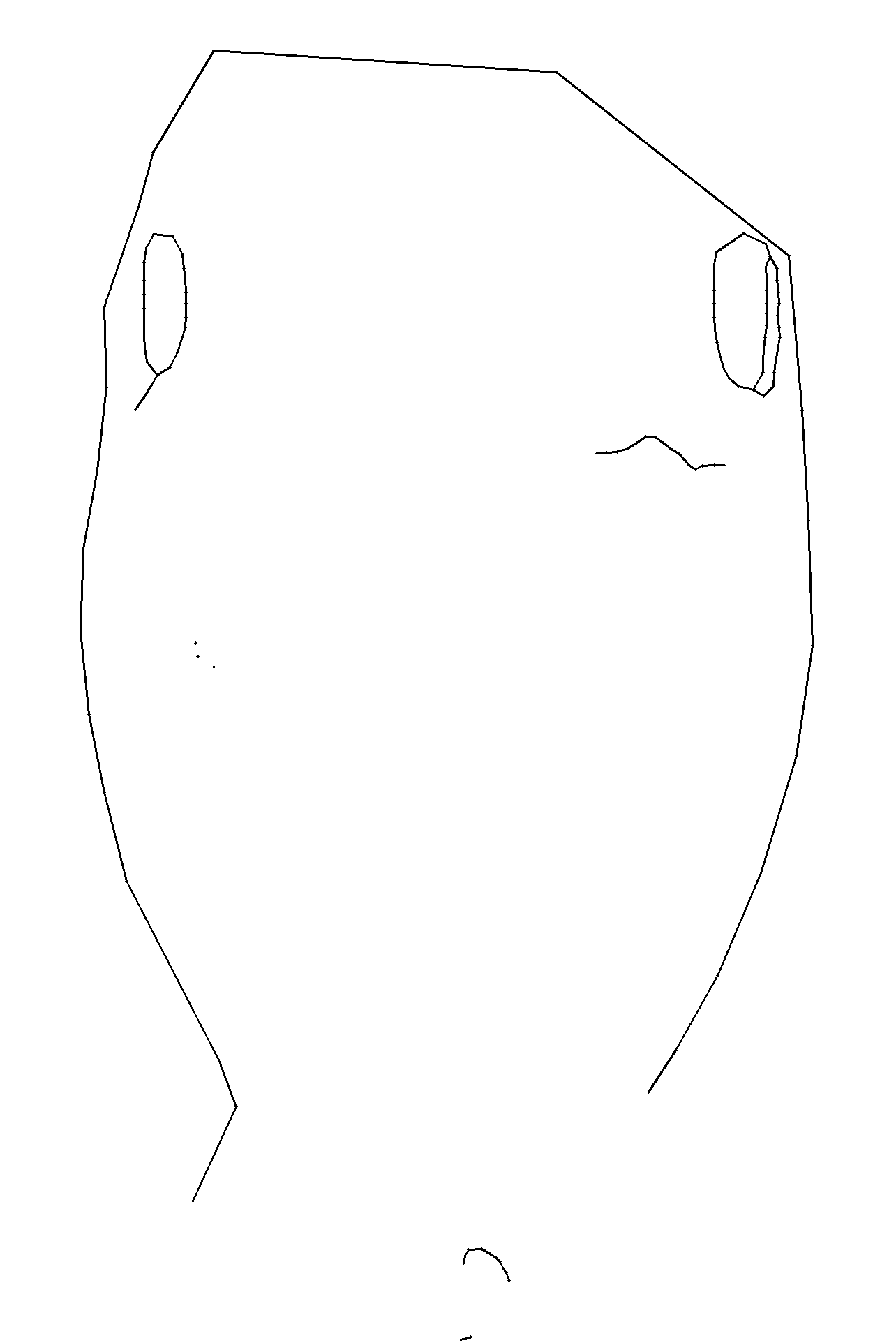}
	\includegraphics[height=42mm]{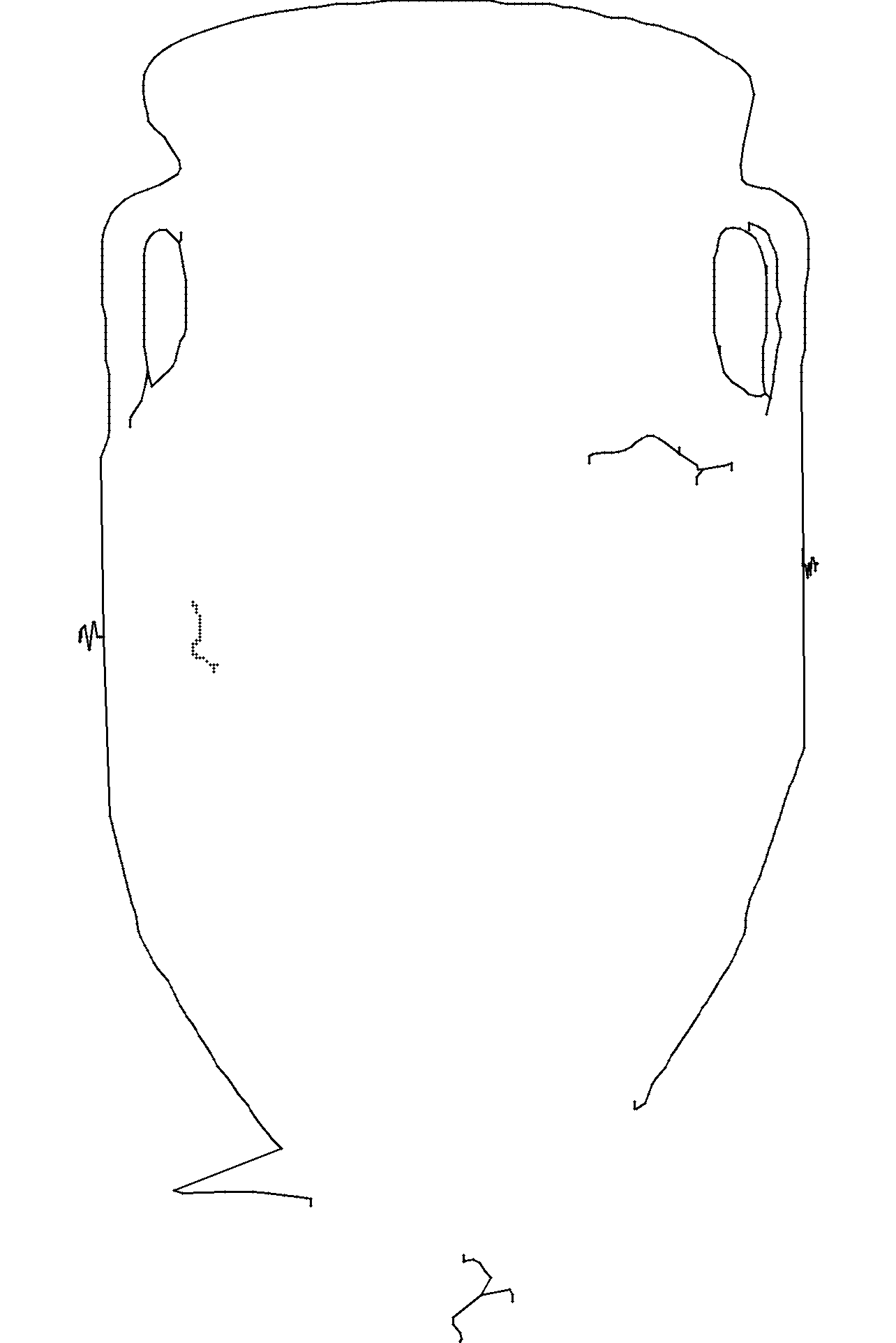}
	\includegraphics[height=42mm]{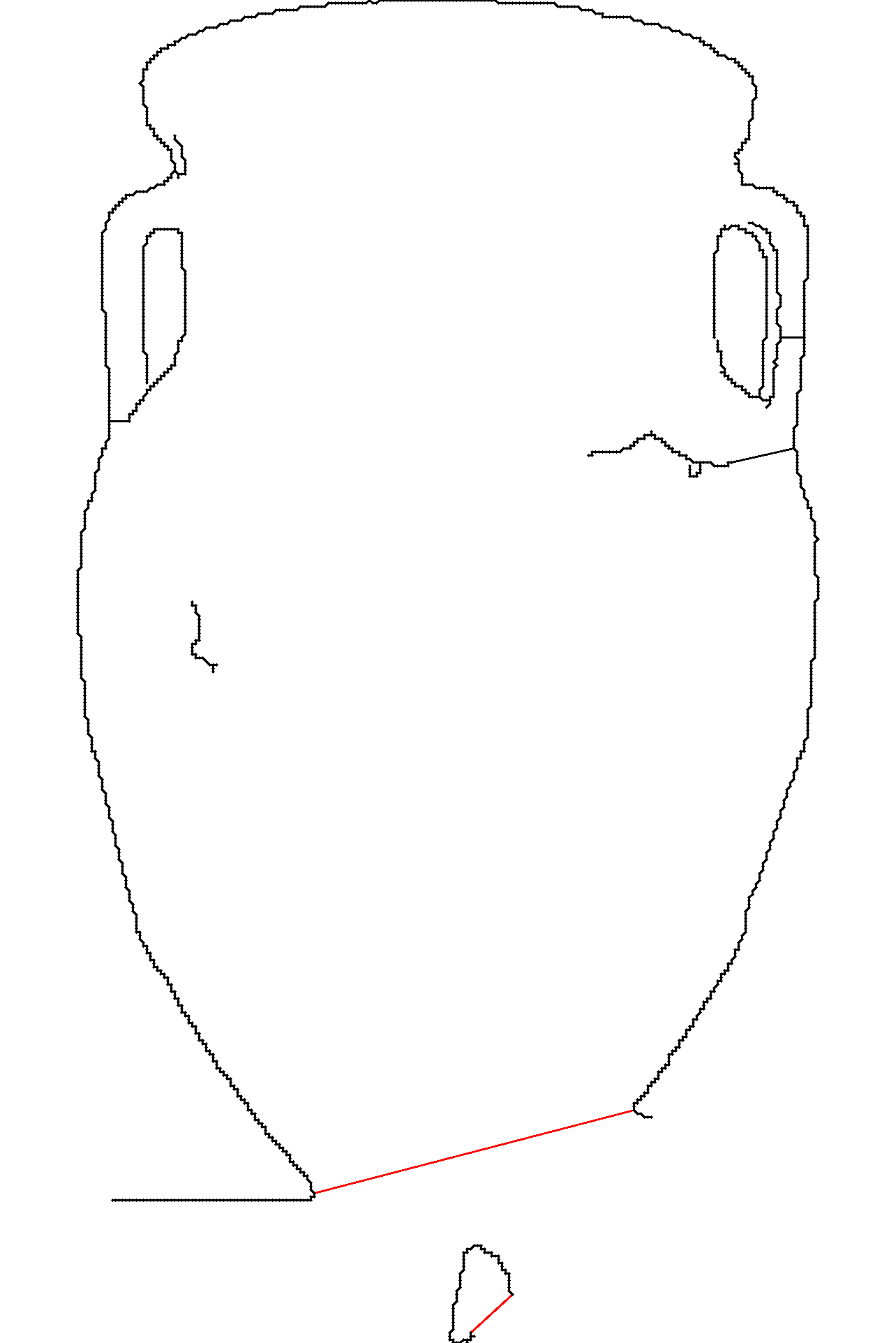}
	\caption{\BSDcaption{227092}}
	\label{fig:pottery}
\end{figure}

\begin{figure}[H]
	\centering
	\def\svgwidth{\columnwidth}
	\includegraphics[height=44mm]{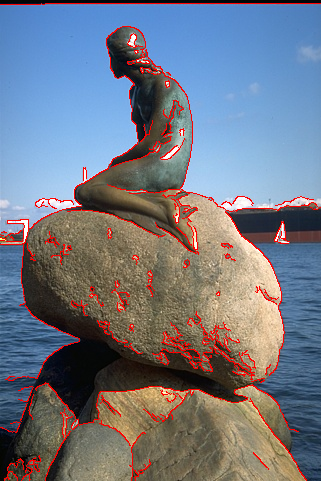}
	\includegraphics[height=44mm]{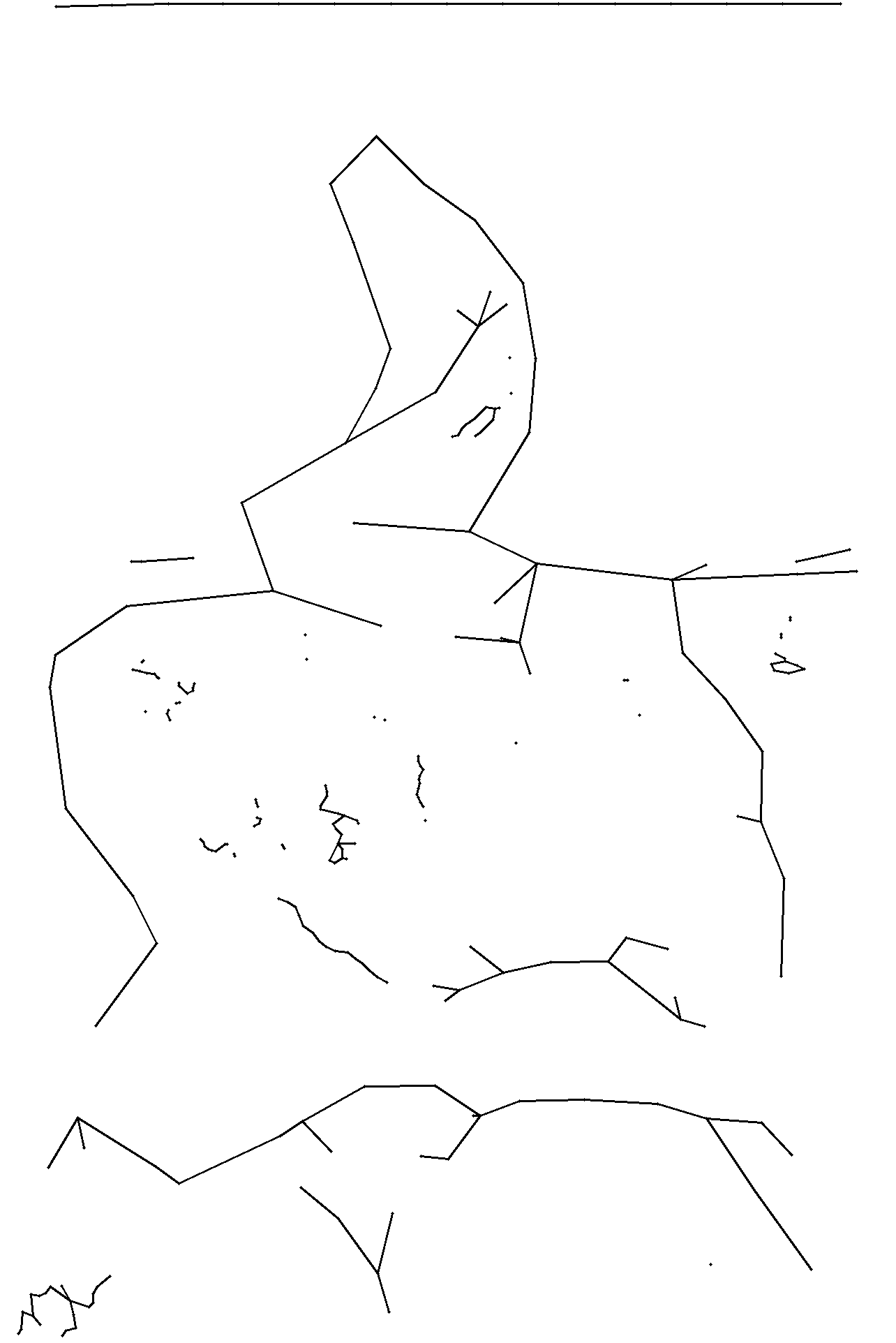}
	\includegraphics[height=44mm]{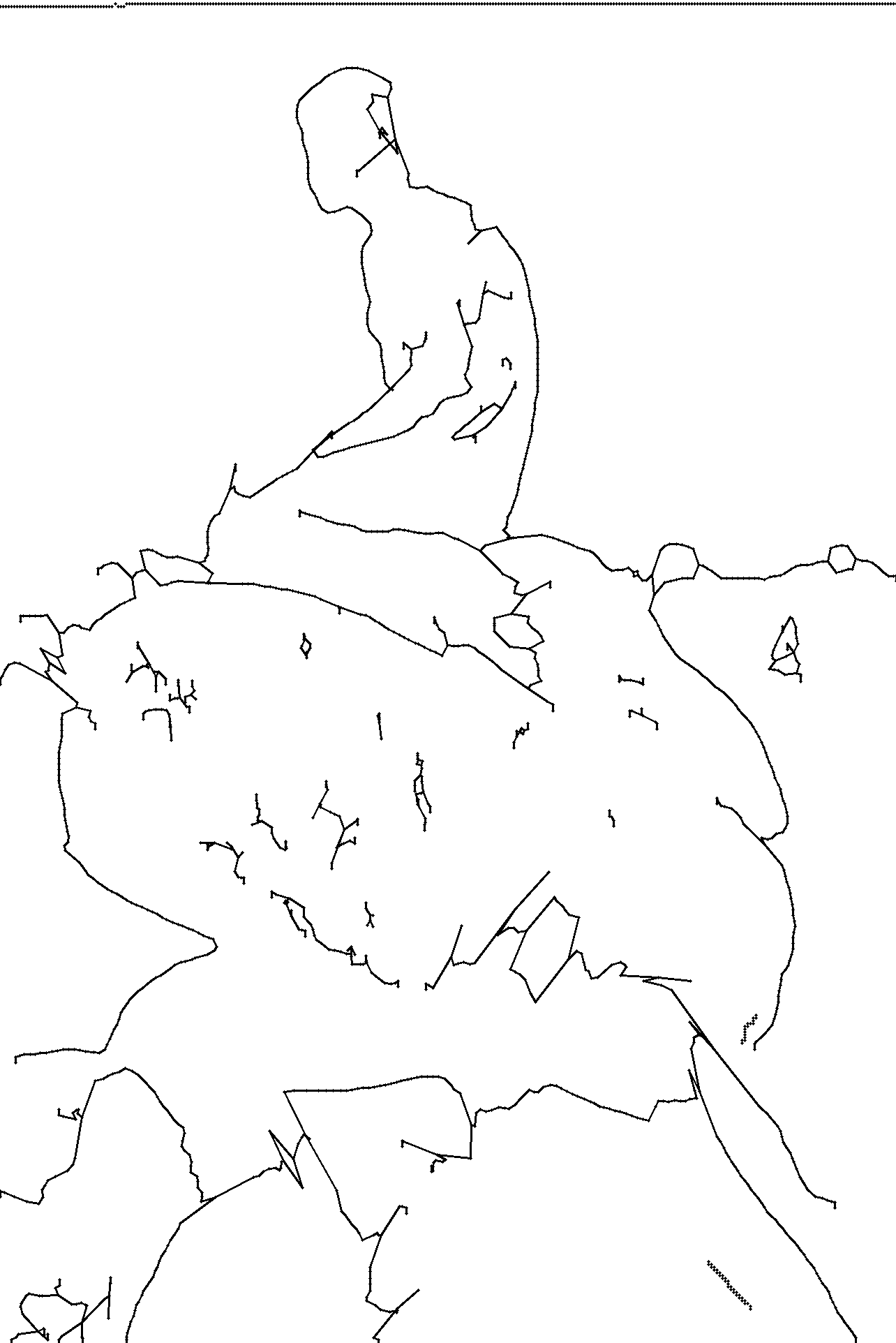}
	\includegraphics[height=44mm]{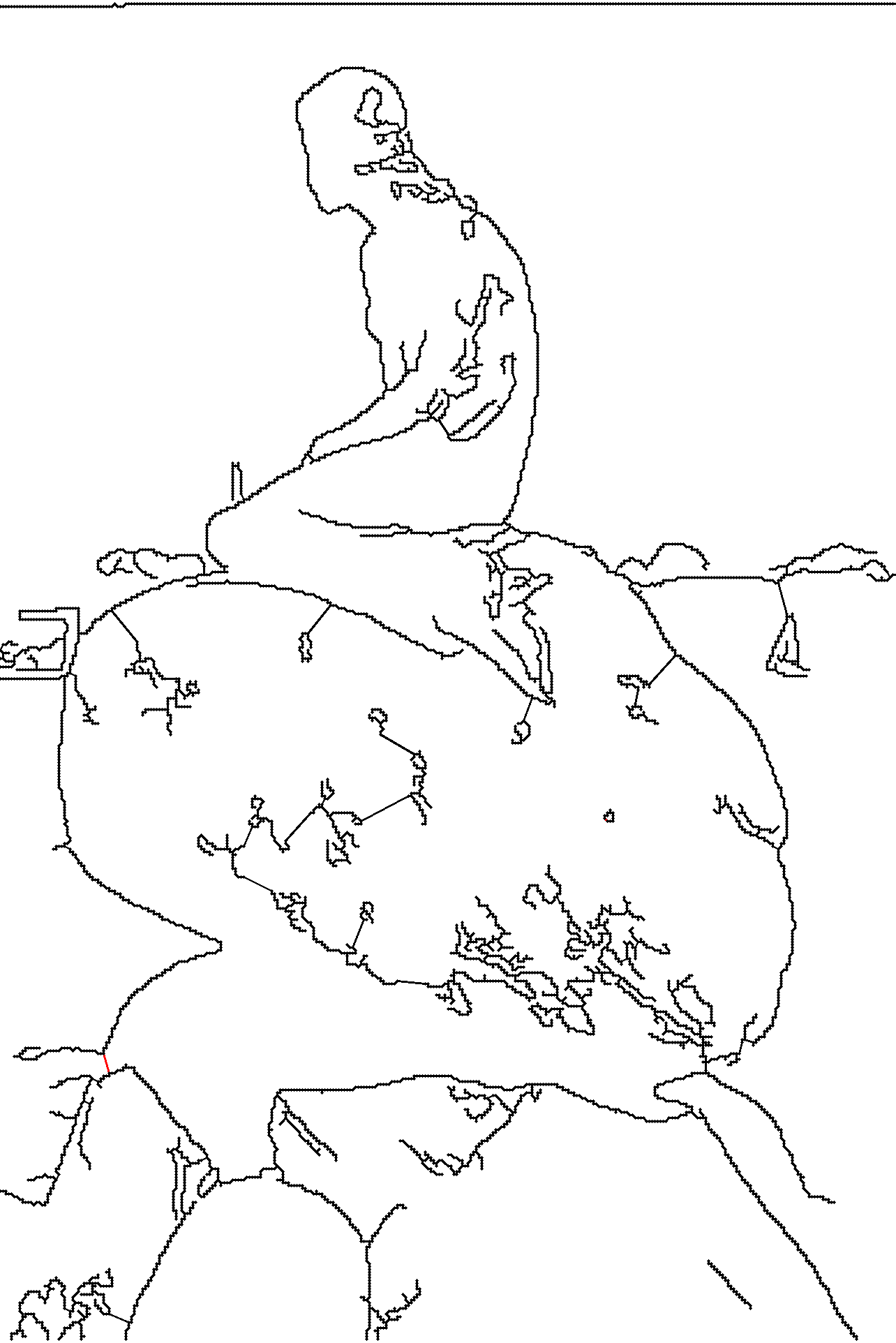}
	\caption{\BSDcaption{372019}}
	\label{fig:rock}
\end{figure}

\begin{figure}[H]
	\centering
	\def\svgwidth{\columnwidth}
	\includegraphics[height=39mm]{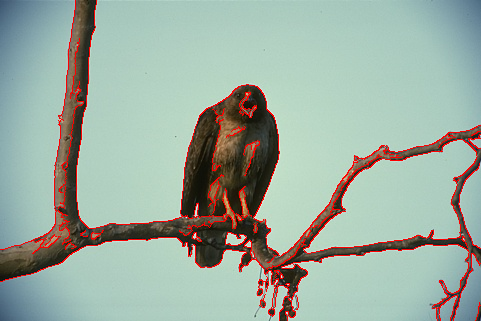}
	\includegraphics[height=39mm]{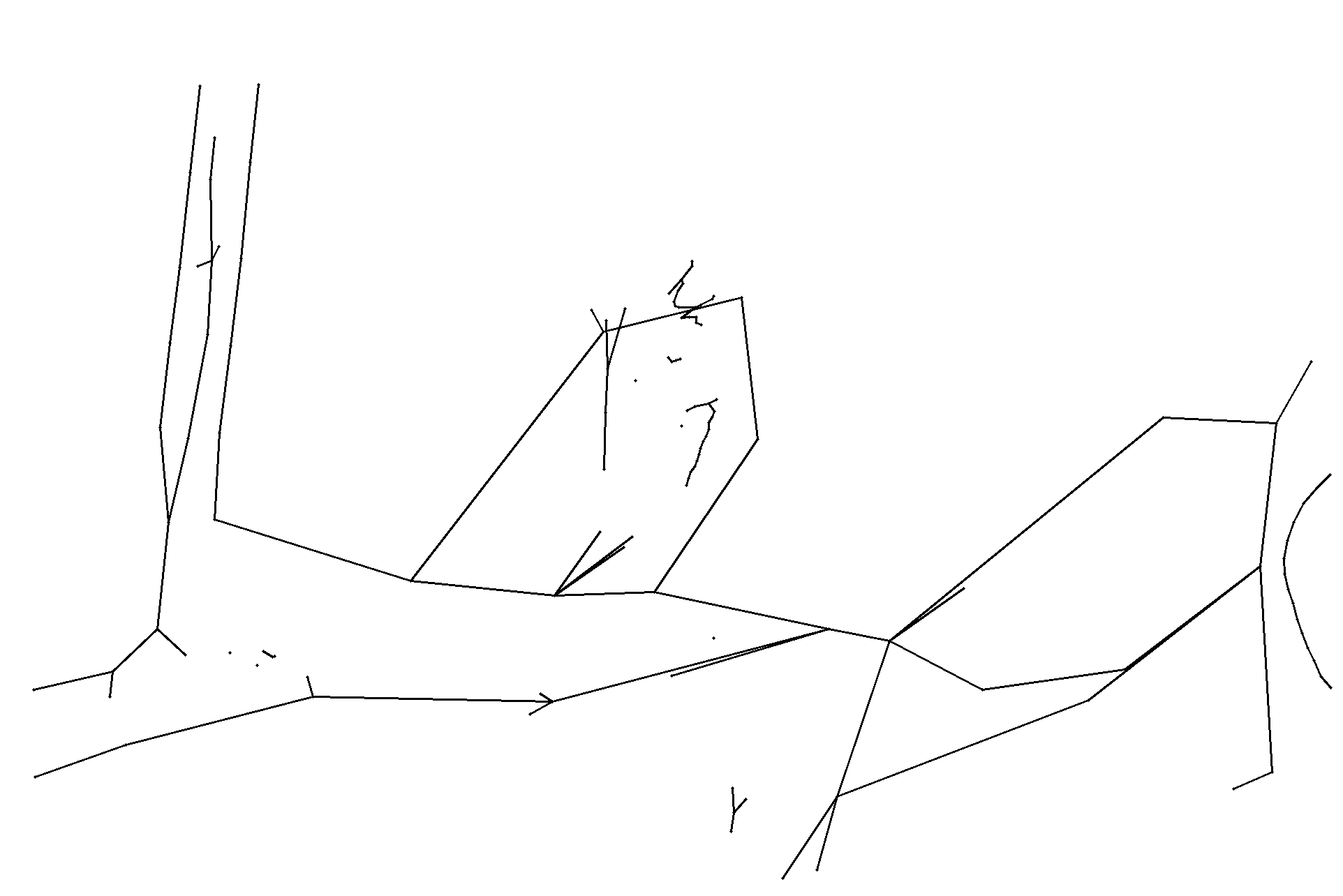}
	\includegraphics[height=39mm]{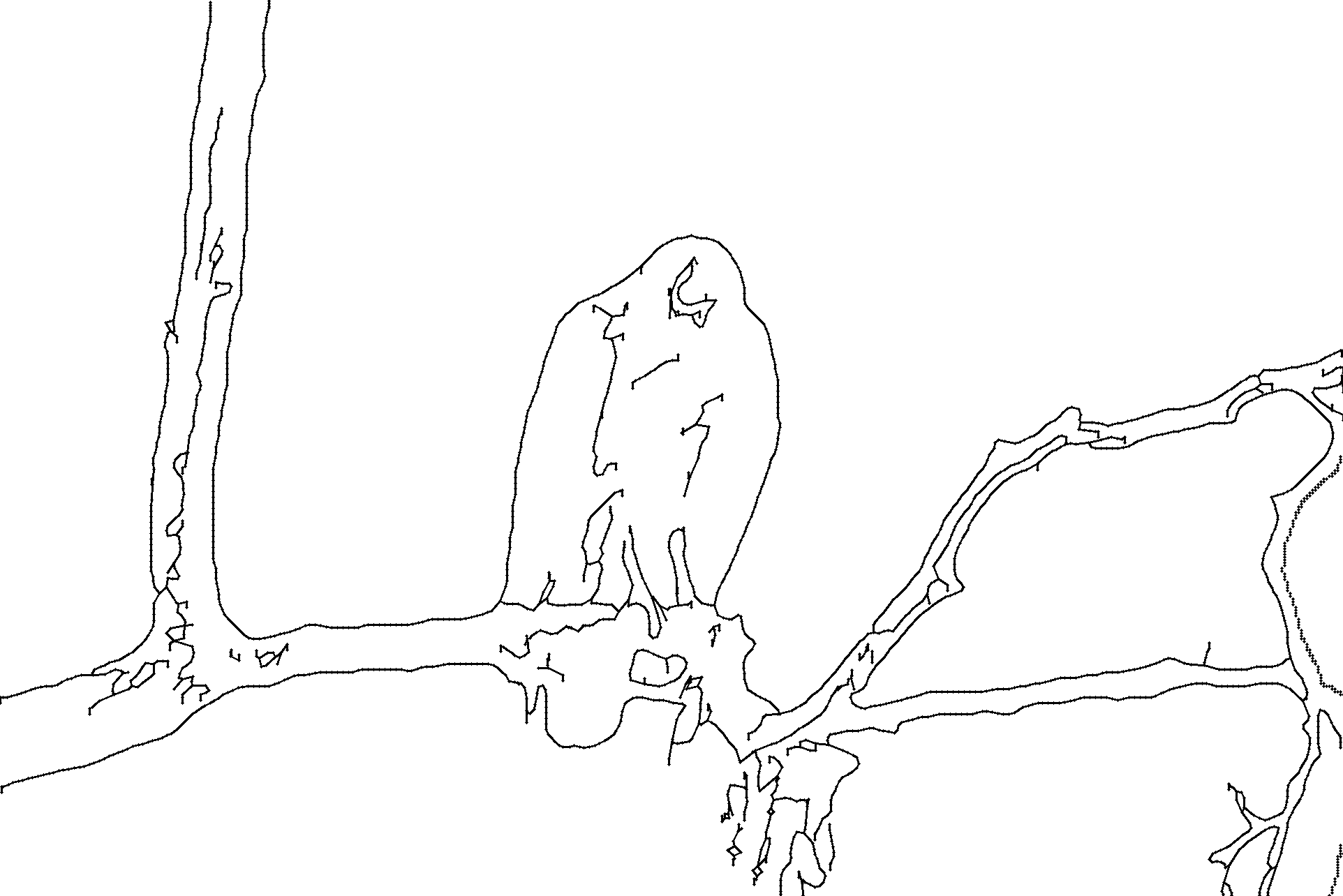}
	\includegraphics[height=39mm]{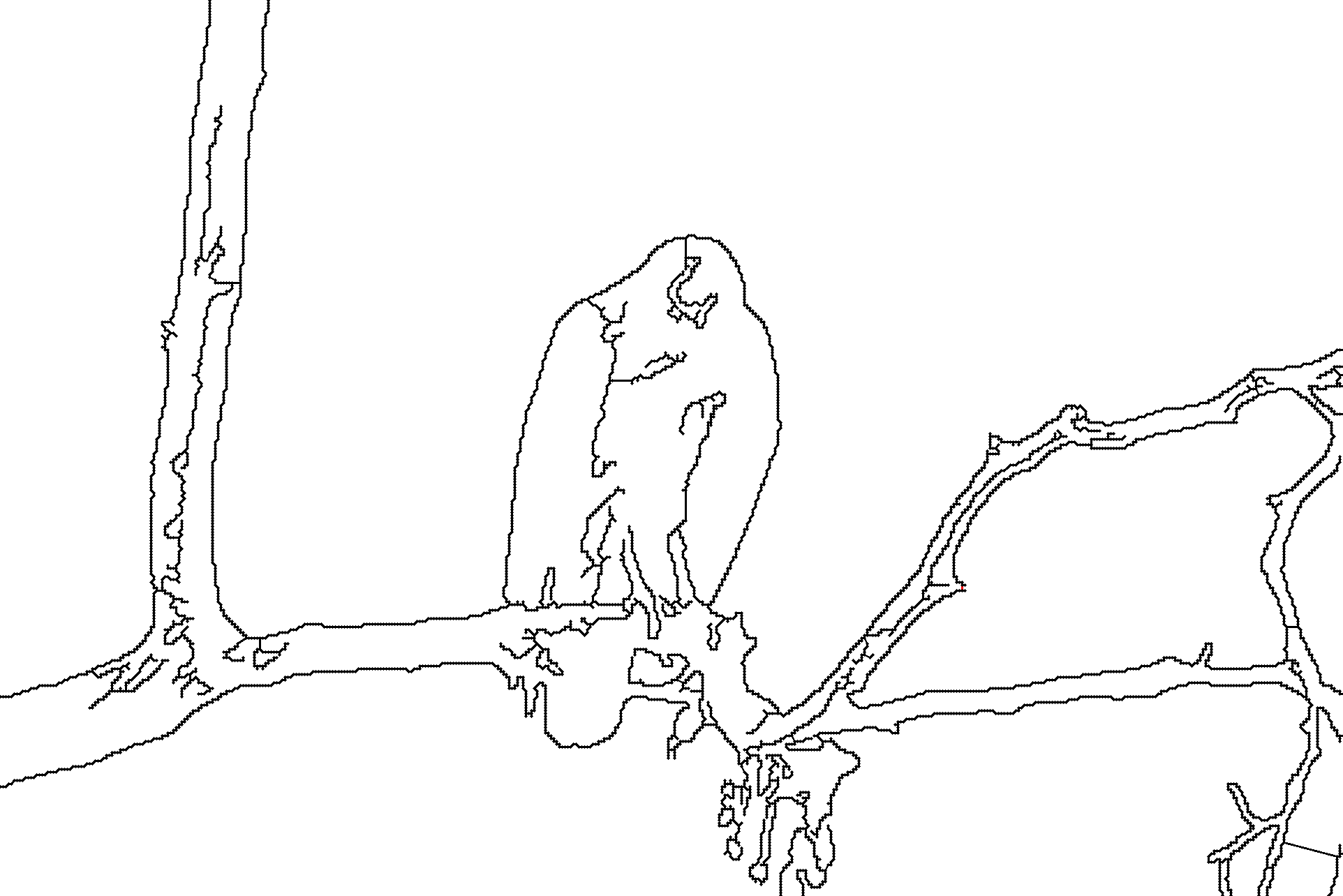}
	\caption{\BSDcaption{42049}}
	\label{fig:bird_1}
\end{figure}

\begin{figure}[!h]
\centering
\includegraphics[width=0.32\textwidth]{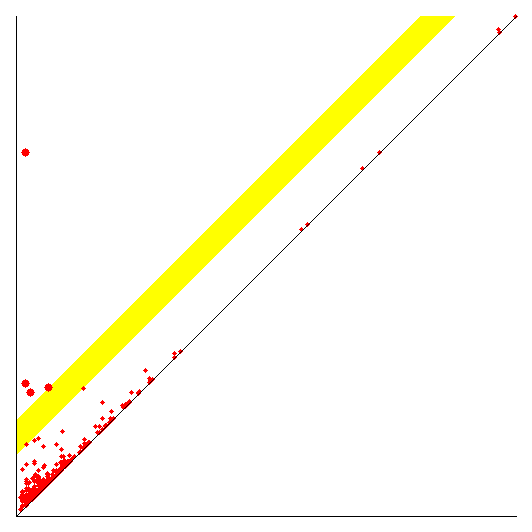}
\includegraphics[width=0.32\textwidth]{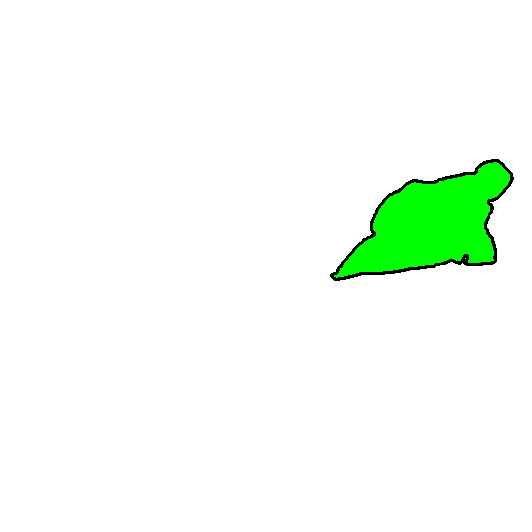}
\includegraphics[width=0.32\textwidth]{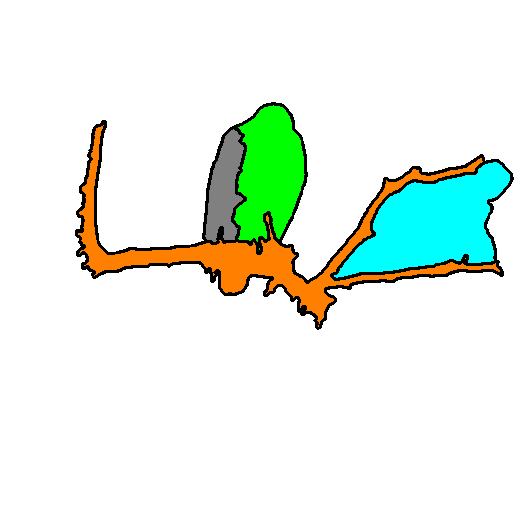}
\caption{\textbf{1st}: for the cloud $C$ from Fig.~\ref{fig:bird_1} the diagram $\PD\{C^{\al}\}$ has 4 dots above the 2nd widest gap. 
\textbf{2nd}: the region enclosed by the only cycle in $\hopes_{1,1}(C)$.
\text{3rd}: the 4 regions enclosed by the 4 cycles in $\hopes_{2,1}(C)$ corresponding to 4 dots above the 2nd widest gap.}
\label{fig:42049_hopes}
\end{figure}

Final Fig.~\ref{fig:42049_hopes} highlights the applicability of $\hopes(C)$, which is computed for a cloud $C$ without any parameters.
Then a user may explore derived skeletons $\hopes_{k,l}(C)$ only by looking at gaps in the persistence diagram $\PD\{C^{\al}\}$ and choosing a certain number of persistent cycles without extra computations. 

\end{document}